\pdfoutput=1

\documentclass[%
 reprint,
 amsmath,amssymb,
 aps,
]{revtex4-2}
\usepackage{physics}
\usepackage{bbold}
\usepackage{accents}
\usepackage{xcolor}

\usepackage{graphicx}
\usepackage{dcolumn}
\usepackage{bm}
\usepackage{hyperref}
\usepackage{amsmath,amsthm,verbatim,amssymb,amsfonts,amscd,graphicx,bm}
\usepackage{enumerate}

\newtheorem{thm}{Theorem}

\theoremstyle{remark}
\newtheorem{prop}{Proposition}[subsection]

\newcommand{\q}[1]{\lvert #1 \rangle}
\newcommand{\qd}[1]{\langle #1 \rvert}

\newcommand{\opus}[1]{ #1 }

\usepackage[normalem]{ulem}
\usepackage{soul}
\DeclareMathSymbol{\shortminus}{\mathbin}{AMSa}{"39}

\newcommand{\NDelta}{\Delta}

\newcommand{\bestguess}[2]{\bar{#1}_{#2}}
\newcommand{\uncertain}[1]{\delta #1}

\begin{document}


\title{Quantum dynamical decoupling by shaking the close environment}

\author{Michiel Burgelman}
\author{Paolo Forni}
\altaffiliation{now at Cambridge Mechatronics Ltd}
\author{Alain Sarlette}
\altaffiliation[Also at ]{Dept.\ of Electronics and Information Systems, Ghent University, Belgium.}
\affiliation{QUANTIC Team, INRIA Paris, 2 Rue Simone Iff, 75012 Paris, France}%

\begin{abstract}
Quantum dynamical decoupling is a procedure to cancel the effective coupling between two systems by applying sequences of fast actuations, under which the coupling Hamiltonian averages out to leading order(s). One of its prominent uses is to drive a target system in such a way as to decouple it from a less protected one. The present manuscript investigates the dual strategy: acting on a noisy ``environment'' subsystem such as to decouple it from a target system. The potential advantages are that actions on the environment commute with system operations, and that imprecisions in the decoupling actuation are harmless to the target. We consider two versions of environment-side decoupling:
adding an imprecise Hamiltonian drive which stirs the environment components; and, increasing the decoherence rates on the environment. The latter can be viewed as driving the environment with pure noise and our conclusions establish how, maybe counterintuitively, isolating the environment from noise sources as much as possible is often not the best option. We explicitly analyze the induced decoherence on the target system and establish how it is influenced by the parameters in both cases. The analysis combines Lindbladian derivation, adiabatic elimination, and Floquet modeling in a way that may be of independent interest.
\end{abstract}

\maketitle

\section{INTRODUCTION}

Several experimental realizations for quantum hardware encounter the situation where a target system is directly coupled to a finite-dimensional ``environment'' system whose decoherence is identified as the main source of induced decoherence on the target. One example of such environment systems are so-called TLS (two-level system) defects in the oxide layer of superconducting Josephson junctions, which decohere typically through phonon channels and are a main mechanism inducing decoherence of superconducting qudits~\cite{Lisenfeld2016,PhysRevB.92.035442}. Another such identified environment would 
be spurious box modes that show some residual coupling to the target modes in microwave resonators. Similar spurious degrees of freedom are likely present in atomic systems.

The idea of Quantum Dynamical Decoupling (QDD, see~\cite{DDofOpenQsystems} and a large set of follow-up work) is to reduce the effective coupling between two quantum systems by using tailored control actions at a faster timescale than the Hamiltonian coupling. Starting from this idea, the present paper proposes to reduce induced decoherence on the target system by applying actions, in a very broad sense, on the \emph{environment side}. The potential advantages are that those actions need not be particularly precise, and that they commute with any system operations one may want to do. In fact, we compute how even adding as noisy dynamics as decoherence channels on the environment, can decrease the induced dissipation on the target system.

\enlargethispage{2.2\baselineskip}
Environment-actuated decoupling also opens the door to refined contributions on analyzing the decoherence induced on the target system. The timescale separation between the effective inter-system coupling and all the dominant dynamics acting on the environment, allows for treating the induced decoherence experienced by the target in a perturbative manner, through the method of adiabatic elimination\cite{Azouit2017b}. This mathematical approach remains fully compatible with control actions applied to the environment over all ranges of magnitudes. More so, since the goal is to \emph{reduce} the effective coupling between the target system and the environment, the validity of the adiabatic elimination approach actually increases.
Using an extension of the adiabatic elimination formalism (see app.~\ref{sec:ae_floquet}), plus Floquet-Markov-type~\cite{GRIFONI1998229} adjustment of the environment decoherence channels themselves when accounting for ultra-strong driving (see Appendix~\ref{app:lindbladian_derivation}), we calculate the induced decoherence rate on the target when applying coherent drives or further decoherence channels on the environment, paving the way for an optimization of the setting.

No control actions whatsoever can hope to decouple from purely Markovian decoherence; or more concretely, in mathematical terms: adding Hamiltonian
actions on a system does not enable to reduce the effect of a purely Lindbladian dissipation channel on the same system. Therefore, QDD has been
considered to cancel spurious effects in two cases. In the first proposal~\cite{DDofOpenQsystems}, the goal is directly formulated as reducing the
coupling to a spurious finite-dimensional  ``environment'' system. The target and spurious environment are both modeled as Hamiltonian systems.
Control sequences are designed to make the effective Hamiltonian coupling vanish up to a certain order, the successive orders typically being given by
a Magnus expansion~\cite{Blanes2009b,Magnus1954} or an equivalent Hamiltonian averaging technique~\cite{Eckardt_2015}. As a result of the QDD controls, the state of the target system undergoes a fast trajectory and its quantum information is preserved in a so-called toggling frame which must be safely followed. In a second type of approach, it is acknowledged that Lindbladian dissipation models are in fact often an approximation, stemming from a direct interaction with a large bath. Identifying the environment with this large bath in a Hamiltonian model and introducing the QDD drives before making the typical Lindblad approximations (Born-Markov, secular approximations), one obtains that decoupling actions are able to counter low-frequency noise, thus effectively modifying and reducing the Lindbladian decoherence channels on target when QDD controls act faster than the cut-off frequency of the noise-spectrum of the bath~\cite{Szczygielski2015,FonsecaRomero,Fanchini2007a}. The decoherence model in the present work is somehow intermediate to these two viewpoints, as it considers a target system coupled to a small effective environment, which itself undergoes Lindbladian decoherence. The small environment thus captures memory effects in the decoherence of the target, as motivated by physically relevant examples like those mentioned in the first paragraph.

More precisely, we analyze the reduction of induced dissipation with two approaches, taking a TLS as the simplest prototypical environment subsystem.

In section~\ref{sec:coherent_decoupling}, we consider the environment subject to periodic drives. While acting on the
environment comes with the security of not deteriorating the target state directly in the case of control imprecision,
we also cannot expect to control an environment system in a well-calibrated manner. Neither can we expect to have
accurate knowledge of the bare environment Hamiltonian. Using a simple model for both these uncertainties, we show that
for the case of a TLS environment, using sufficient time-scale separation in the applied drive enables efficient QDD
    despite control imprecisions. As a trade-off for requesting a strong time-scale separation, we consider a very simple control signal, consisting of only one harmonic tone. The analysis is performed with a generalization of adiabatic elimination adapted to periodically-driven systems, inspired by the basic Floquet property, and which we believe to be novel. This analysis method also differs from the more standard QDD analysis based on Magnus expansion in a purely Hamiltonian setting. We obtain an explicit Lindblad model for the leading-order induced decoherence on the target. The procedure, explained in Appendix~\ref{sec:ae_floquet}, would allow in principle to obtain further perturbative corrections in powers of the coupling strength.

In section~\ref{sec:noisy_decoupling}, we consider the limit of extremely disorganized QDD actions, by adding decoherence channels instead of Hamiltonians to the environment subsystem. Indeed, increasing the decoherence strength on the environment also decreases its effective coupling with the target system, and the scaling for induced dissipation on the target often turns out to be favorable at higher environment decoherence. Using second-order adiabatic elimination formulas, we analyze the resulting behavior in detail, providing some general results and characterizing the optimal choice for typical settings with a TLS environment.
\enlargethispage{1.0\baselineskip}

\section{MODEL DESCRIPTION}\label{sec:model}

As a main setting throughout this work, we consider a general target system T undergoing Hamiltonian dynamics,
and whose main source of decoherence is an undesired Hamiltonian interaction with an environment E which itself undergoes fast, Lindbladian decoherence.
In a rotating frame around the bare frequencies of both T and E, the general evolution is described by:
\begin{equation}\label{eq:genmodel}
\tfrac{d}{dt}\rho = -i[ H_T + H_E + H_{TE},\, \rho] + \sum_k \kappa_k \mathcal{D}_{L_k}(\rho) \;.
\end{equation}
Here we have introduced the general Lindbladian dissipator
$$\mathcal{D}_{X}(\rho) = X \rho X^\dag - \frac{1}{2} \qty(X^\dag X \rho + \rho X^\dag X).$$
The $L_k$ represent various decoherence channels of E, at respective rates $\kappa_k$.
The Hamiltonians $H_T, H_E$ and $H_{TE}$ respectively act on T, on E, and couple T with E.
The objective is to protect quantum information stored in the target system T.
Standard QDD works by applying well-designed sequences of control Hamiltonians $H_T$.
The present paper rather assumes $H_T=0$ and analyzes how one can decrease the induced decoherence on T, by acting on the environment through $H_E$ on the one hand,
or through addition or modification of the $\kappa_k$ on the other hand.

A prototypical example for E is a set of two-level-systems (TLS's), like defects in the oxide layer of superconducting Josephson junctions~\cite{Lisenfeld2016,PhysRevB.92.035442}.
At the dominating order, we can consider the contribution to the overall induced decoherence of each such TLS individually~\cite{Forni2019}.
In a rotating frame of both the target system and TLS, we consider a general stationary coupling
\begin{equation}\label{eq:coupling_model}
   H_{TE} = g \qty(T_x \otimes \sigma_x + T_y \otimes \sigma_y + T_z \otimes \sigma_z) \; .
\end{equation}
Here, $g$ is a small coupling rate with the dimension of a frequency (units where $\hbar = 1$), $T_x$, $T_y$ and $T_z$ are arbitrary Hermitian operators acting on the target system,
and $\sigma_x, \sigma_y, \sigma_z$ are the Pauli operators on the TLS.

The TLS's themselves are thus assumed poorly protected and quickly dissipate according to a Lindbladian model, as described in \eqref{eq:genmodel}.
When adding coherent drives in Section \ref{sec:coherent_decoupling}, we typically assume the dominating dissipation channels:
\begin{equation}\label{eq:pm_dissipation}
    L_k \in \{ \sigma_-, \; \sigma_+ \} \; ,
\end{equation}
corresponding to loss and excitation in the $\sigma_z$-basis of E.
When adding/tuning dissipation channels in section~\ref{sec:noisy_decoupling},
the environment side is treated purely on the basis of a given set of dissipation operators $L_k$ whose rates $\kappa_k$ may be adjustable in some range.

In this way we mainly consider the Lindbladian dissipation operators $L_k$ as fixed, independently of the mechanisms added to reduce the coupling between T and E.
Since our goal towards QDD is to drive strongly, we also compute corrections to the dissipation on E for the case where
ultra-strong driving has an effect on the dissipation model itself.
For this, in Section~\ref{sec:dissipation_model}, we rederive a modified Lindbladian starting from a model where E interacts with a large bath.

In the remainder of this work, we compute and analyze the decoherence that this setting induces on the target system T.

\section{Coherent decoupling with drives}\label{sec:coherent_decoupling}

In this section, we pursue the strategy of applying coherent QDD controls $H_E$ to a TLS-type and decohering environment, in order to decouple it from a general target system.
The section is organized as follows.
We start in Section \ref{sec:qdd_related} by recalling the concept of QDD more explicitly, including previous work
concerning continuous bounded-strength decoupling drives in particular, applied through $H_T$. Next, we translate the
application of QDD drives to the environment side $H_E$. In section~\ref{sec:qdd_proposal} we propose a continuous QDD
control signal accounting for inevitable control imprecision when acting on the environment. In
section~\ref{sec:reduced_model}, we then calculate an explicit Lindbladian model for the decoherence induced on the
target when applying the QDD controls. This involves an extension of the adiabatic elimination approach to time-periodic
couplings which we summarize in Appendix~\ref{sec:ae_floquet}. We analyze the obtained expressions, highlighting the efficiency of applying environment-side QDD drives.
For further consistency, in section~\ref{sec:dissipation_model}, we re-discuss the dissipation channels on the environment when the QDD drives $H_E$ become significant compared to bare system frequencies.

\subsection{QDD and related work}\label{sec:qdd_related}

Established QDD approaches consist in applying control pulses to the target system T that send its state quickly wandering around its Hilbert space. The explicit objective is that the average effect over one wandering cycle of all relevant coupling operators goes to zero.
The simplest example is the case of a target qubit T with only one coupling term involving $\sigma_z$. In this case,
one can periodically apply $\pi$-pulses around the $\sigma_x$-axis of T, such that it effectively accumulates phase around $\pm \sigma_z$ half of the time each, thus canceling the coupling effect on average if there is no other motion in the meantime. The shorter the period between subsequent pulses, the better T is being decoupled from E. This is the well-known spin echo sequence \cite{Hahn1950}.
The generalization of this idea to general systems with arbitrary stationary couplings was introduced in~\cite{DDofOpenQsystems}, and versions replacing the instantaneous pulses with bounded drives in group-based decoupling schemes were established in~\cite{Chen2006,Viola2003,Khodjasteh2009a,Wocjan2006}.

For the case of a target qubit T, a different type of bounded-drive QDD scheme has been devised, using the
combination of a static field and a simple monochromatic drive~\cite{Fanchini2007a,Fanchini2007,Chaudhry2012}.
Explicitly, their control Hamiltonian to decouple a single qubit takes the form
\begin{equation}\label{eq:mfd_simple}
    H_T(t) = \frac{\omega}{2} \sigma_z + \frac{\omega}{4} \qty(\cos(\omega t) \sigma_x + \sin(\omega t) \sigma_y) .
\end{equation}
Under this drive, the qubit state is made to rotate around the $\sigma_x$-axis in a frame which itself rotates around the $\sigma_z$-axis at double the frequency. Indeed, $H_T(t)$ has been designed to generate the unitary evolution
$$ U_T(t) = e^{- i \frac{\omega}{2} \sigma_z t} e^{- i \frac{\omega}{4} \sigma_x t} $$
of the target qubit in absence of any further dynamics. We can clearly see the composition of two rotations around orthogonal axes in the Bloch sphere.
The effectiveness of this QDD scheme can be analyzed in a frame that eliminates the QDD controls, called the toggling frame.
Indeed, it is easy to verify that the first-order decoupling condition is satisfied~\cite{DDofOpenQsystems},
namely that any coupling operator averages out to $0$ under this unitary evolution:
\begin{equation}
    \frac{\omega}{2 \pi} \int_0^{\frac{2 \pi}{\omega}} U_T^\dag(t) \sigma_a U_T(t) \; dt = 0, \;\;\text{ for }  a \in \{ x,y,z\} \, .
\end{equation}
When this first-order decoupling condition is satisfied, the effect of any coupling between T and E can be made arbitrarily small 
by ramping up $\omega$.
This is proven by identifying the average coupling as the first and leading order of a Magnus expansion of the effective dynamics in powers of $\frac{g}{\omega}$.

Such results are hence typically established by focusing on the Hamiltonian part of the model, i.e.~discarding the $L_k$ in \eqref{eq:genmodel} and showing that the effective coupling between T and E is canceled up to some order(s). In such setting, the QDD treats T and E in a symmetric way, and one could in principle consider applying the QDD drives to either system. The advantages of acting on E rather than on T would be that (i) we minimize the danger of perturbing quantum information with actuation imprecisions and (ii) we can keep applying QDD drives irrespective of the system operations on T.
Indeed, standard QDD acting on T requires specific adaptations when T is also subject to actions operating the quantum information system,
like logical gates \cite{Khodjasteh2009a}. On the downside, of course we can hope to act on E only if it is well identified and of reasonably small dimension, like for instance spurious TLS's \cite{Lisenfeld2016,PhysRevB.92.035442}.
In addition, the situation is not as symmetric between T and E when one explicitly introduces that E is a strongly decohering environment,
i.e.~when introducing the $L_k$ in \eqref{eq:genmodel}. We therefore provide an analysis that explicitly considers the decoupling Hamiltonian and the decoherence operators together.

In the remainder of this section we thus address three main points in which the QDD methodology needs to be extended,
to show how it still works with environment-side driving.
Firstly, we need to include a significant amount of control imprecision into the QDD drives, since a TLS environment cannot be assumed as precisely addressable as the target system.
Secondly, since the fastest timescale is embodied on E, we propose an analysis of the model~\eqref{eq:genmodel} including the decoherence channels $L_k$.
With adiabatic elimination techniques we eliminate the fast subsystem E and directly compute the induced decoherence on $T$, rather than going through the computation of effective couplings with Hamiltonian averaging techniques like the Magnus expansion.
Lastly, the model with dissipation channels $L_k$ acting on E has to be rediscussed under ultra-strong QDD driving,
as this model ultimately stems from interaction of E with further external degrees of freedom in a way that can also be affected by the driving.

\subsection{Double-timescale QDD proposal}\label{sec:qdd_proposal}

The E subsystem, i.e.~the spurious TLS, is not an accurately addressable subsystem. First, we will not assume to know the eigenfrequency $\Omega_E$ of E exactly. To account for this, we split up $\Omega_E$ into its best-guess value $\bestguess{\Omega}{E}$ and an uncertain constant
    deviation $\uncertain{\Omega_E}$:
    $$\Omega_E = \bar{\Omega}_E + \delta \Omega_E.$$
    With this decomposition, the model~\eqref{eq:coupling_model} is defined in a rotating frame
    w.r.t. $\bar{\Omega}_E$, and $H_E$ features a residual unknown detuning:
    $$ H_E(t) = \frac{\delta \Omega_E}{2} \sigma_z + H_c(t).$$
Here, $H_c(t)$ stands for the applied control Hamiltonian.

As a second point of control imperfection, we will not assume that a calibration is carried out for the actual amplitude reaching E upon applying a signal in the lab. Hence for the definition of $H_c$, we introduce the same decomposition for the control parameters into best-guess quantities and unknown deviations thereof. We propose to use a simple continuous signal similar to \eqref{eq:mfd_simple}, meant to cancel the general coupling~\eqref{eq:coupling_model}:
    $$H_c(t) := \frac{\omega_1}{2} \sigma_z + \frac{\omega_2}{2} \qty(\cos(\bestguess{\omega}{1} t) \sigma_x + \sin(\bestguess{\omega}{1} t) \sigma_y),$$
    with
    \begin{align}
        \omega_1 &= \bar{\omega}_1 + \delta \omega_1,\\
        \omega_2 &= \bar{\omega}_2 + \delta \omega_2.
    \end{align}
    Note that the drive frequency $\bar{\omega}_1$ is well-known, whereas the amplitudes of the static field and of the $\sigma_{x,y}$-drive are only roughly known, involving uncertainties $\delta \omega_1$ and $\delta \omega_2$ respectively.
    Defining 
    $$\Delta = \delta \omega_1 + \delta \Omega_E,$$
    the total Hamiltonian can be written as
    \begin{equation}\label{eq:EAD_control}
        H_E(t) := \frac{\NDelta + \bar{\omega}_1}{2} \sigma_z + \frac{\omega_2}{2} \qty(\cos(\bestguess{\omega}{1} t) \sigma_x + \sin(\bestguess{\omega}{1} t) \sigma_y).
    \end{equation}

Although our actual analysis will consider the full model with decoherence channels, we can already take a look at the implications of such control in a purely Hamiltonian setting. 

$\bullet$ The evolution of E under $H_E(t)$ alone can be understood by first moving to a rotating frame w.r.t. $\frac{\bestguess{\omega}{1}}{2} \sigma_z$, yielding a remaining constant Hamiltonian $ \frac{\NDelta}{2} \sigma_z + \frac{\omega_2}{2} \sigma_x$.
In this frame the state will rotate at a speed
$$\Lambda :=\sqrt{\NDelta^{2} + \omega_{2}^{2}},$$
around the axis  $$\sigma_{\alpha x} = \cos(\alpha)  \sigma_x + \sin(\alpha)  \sigma_z,$$
where we have defined
$$ \cos(\alpha) = \frac{\omega_2}{\Lambda}, \quad   \sin(\alpha) = \frac{\NDelta}{\Lambda}.$$
Back in the original frame, the associated propagator thus reads
\begin{equation}\label{eq:control_propagator}
    U_E(t) := e^{-i \bestguess{\omega}{1} \sigma_z t / 2} e^{- i \Lambda \sigma_{\alpha x} t / 2} \; .
\end{equation}
The E subsystem thus undergoes two composite rotations around axes in the Bloch sphere which would be orthogonal in absence of the detuning $\NDelta$.
We see that the presence of $\NDelta$  prevents us from applying exact $\sigma_x$ rotations, as would be required in a continuous-time analog of the spin echo strategy. As the angle is determined by $\NDelta/ \omega_2$, we should favor a large value of $\omega_2$.
Considering $\NDelta$ of possibly the same order as $\bestguess{\omega}{1}$, this would suggest to take $\omega_2 \gg \bestguess{\omega}{1} \gg g$,
 where the latter is the strength of the coupling Hamiltonian \eqref{eq:coupling_model}.

$\bullet$ Next, applying the propagator associated to $H_E(t)$ on the coupling Hamiltonian \eqref{eq:coupling_model}, it is easy to verify that
$U_E^\dagger(t) \sigma_{x,y} U_E(t)$ only involve terms oscillating at frequencies $\pm \bestguess{\omega}{1}$ and $\Lambda \pm \bestguess{\omega}{1}$,
while $$U_E^\dagger(t) \sigma_z U_E(t) =
\sin(\alpha) \sigma_{\alpha x} - \cos(\alpha) (e^{i \Lambda t} \sigma_{\alpha +} + e^{- i \Lambda t} \sigma_{\alpha -}), $$
where $\sigma_{\alpha \pm}$ are lowering and raising operators with respect to the eigenstates of $\sigma_{\alpha x}$.
Having $ \Lambda \gg \bestguess{\omega}{1} \gg g $, we can perform a rotating-wave approximation (RWA) and obtain the non-zero average coupling
\begin{equation}\label{eq:avHam}
g \sin(\alpha) \sigma_{\alpha x} = g \frac{\NDelta}{\omega_2} \sigma_{\alpha x} + g \mathcal{O}\qty(\frac{\NDelta^3}{\omega_2^3}) \; .
\end{equation}
Thus taking $\omega_2 \gg \NDelta \sim \bestguess{\omega}{1}$ in this formula, and $\bestguess{\omega}{1} \gg g$ to justify the RWA, indeed appears to reduce the effective coupling between T and E.

\subsection{Analysis of decoherence on target}\label{sec:reduced_model}

In a rotating frame w.r.t. $\frac{\bestguess{\omega}{1}}{2} \sigma_z$, and defining $T_\pm = T_x \pm i T_y$, the joint evolution of the target and TLS is described by the master equation
\begin{align}
    \tfrac{d}{dt}\rho& = \kappa_- \mathcal{D}_{\mathbb{1}_T \otimes \sigma_-}(\rho) + \kappa_+ \mathcal{D}_{\mathbb{1}_T \otimes \sigma_+}(\rho)\label{eq:model_strong}\\
    &- i \frac{\Lambda}{2} \comm{\mathbb{1}_T \otimes \sigma_{\alpha x}}{\rho}\nonumber\\
    &- ig \comm{T_z \otimes \sigma_z + e^{i \bestguess{\omega}{1} t} T_- \otimes \sigma_+ + e^{- i \bestguess{\omega}{1} t} T_+ \otimes \sigma_-}{\rho},\nonumber
\end{align}
when assuming drive-independent decoherence channels $L_k \in \{\sigma_-, \sigma_+\}$ on E.
In the present section, we analyze the induced decoherence on T, by obtaining explicit formulas for its reduced dynamics thanks to adiabatic elimination of the environment E. For this we rely on a timescale separation as E dissipates with rates $\kappa_k \gg g$ dominating the coupling Hamiltonian. The work of \cite{Azouit2017b} explains how to obtain the reduced dynamics of T as a power expansion in $g / \kappa_k$, considering a stationary coupling Hamiltonian as a perturbation. In appendix~\ref{sec:ae_floquet}, we have derived a general extension of this adiabatic elimination approach for the case where the coupling Hamiltonian is time-periodic. The related formulas could be of independent interest to treat other cases where
first performing a lowest-order RWA, then adiabatically eliminating the fastly decohering
degrees of freedom, does not yield the correct leading-order induced dissipation.

Before moving to the QDD results, we start by summarizing the time-periodic adiabatic elimination extension in the context of our bipartite T-E system.

In the absence of any coupling ($g=0$), (\ref{eq:model_strong}) features an invariant subspace with zero dynamics,
consisting of all the states of the form $\rho_T \otimes \bar{\rho}_E$, for an arbitrary state $\rho_T$ of the target system and where $\bar{\rho}_E$ is the unique steady state of the Lindbladian acting on E.
The remaining degrees of freedom in this subspace can thus be trivially identified with the state space of the target system. Moreover, any state quickly relaxes towards this invariant subspace.
For a non-zero but weak constant coupling $g$, this invariant subspace is slightly perturbed~\cite{Azouit2017b}: there remains an invariant subspace of the same dimension,
in which the dynamics is slow (perturbed eigenvalues of the superoperator),
and where the target subsystem is slightly hybridized with E (perturbed eigenspaces of the superoperator). 

In Appendix~\ref{sec:ae_floquet} we show how for a time-periodic coupling of period $\frac{2 \pi}{\bestguess{\omega}{1}}$ like in (\ref{eq:model_strong}),
we can still identify an invariant subspace --- i.e.~a subspace $\mathcal{M}$ such that $\rho(0) \in \mathcal{M}$ implies $\rho(t) \in \mathcal{M}$ for all $t$ ---
but this subspace moves periodically in time with period $\frac{2 \pi}{\bestguess{\omega}{1}}$.
Moreover, much like in the classical Floquet theorem for periodic linear systems,
the total dynamics on $\mathcal{M}$ can be decomposed into slow (i.e.~order $g$), stationary Markovian dynamics on the one hand,
and a fast periodic motion of the invariant subspace as a whole on the other hand.
The periodic motion of the subspace can be described by a global change of variables, completely agnostic of the actual state or its dynamics.
The slow Markovian dynamics can in turn be parametrized by a state $\rho_s$ living in a space of the same dimension as
T, and thus essentially describes the effective decoherence of T.

This picture leads to the following Ansatz for the solution of~\eqref{eq:model_strong} (and later \eqref{eq:model_ultrastrong}) within the invariant subspace:
\begin{equation}\label{eq:anszats}
    \rho(t) = \mathcal{K}_g(\rho_s(t), t),
\end{equation}
with
\begin{equation}\label{eq:reduced_model}
    \tfrac{d}{dt} \rho_s(t) = \mathcal{L}_{s,g}(\rho_s(t)).
\end{equation}
Here, $\rho_s$ is a state of the same dimension as T to represent its slightly hybridized version; $\mathcal{K}_g(\cdot,t)$ is a $\frac{2 \pi}{\bestguess{\omega}{1}}$-periodic superoperator close to $\rho_s \mapsto \rho_s \otimes \bar{\rho}_E$ defining the embedding of the invariant subspace in the total system space; and $\mathcal{L}_{s,g}$ is a stationary Lindbladian representing the slow Markovian dynamics occurring within the invariant subspace. In order to identify $\mathcal{L}_{s,g}$ and $\mathcal{K}_{g}$,
like in \cite{Azouit2017b}, we write both as a power expansion.
The small expansion parameter is $\varepsilon = \frac{g}{\bestguess{\omega}{1}} \ll 1$, with $\bestguess{\omega}{1}$ the
frequency of the driving as in~\eqref{eq:EAD_control}, and we write
\begin{align}
    \mathcal{K}_{g}(\cdot,t) &:= \sum_{k = 0}^\infty \varepsilon^k \mathcal{K}_k(\cdot, t),\\
    \mathcal{L}_{s,g} &:= \sum_{k = 1}^\infty \varepsilon^k \mathcal{L}_{s,k}.
\end{align}
Substituting this Ansatz into~\eqref{eq:model_strong} and identifying equal powers of $\varepsilon$ then
allows to solve for the unknowns $\mathcal{L}_{s,k}$ and $\mathcal{K}_k$ order by order, as shown in Appendix~\ref{sec:ae_floquet}.
    In line with standard adiabatic elimination, the convergence of the series is ensured provided $\frac{g}{\kappa} \ll 1$,
    with $\kappa$ the typical dissipation rate of E. The validity of the expansion thus depends on the timescale
    separation $\bestguess{\omega}{1} \gg g$ and $\kappa \gg g$. However, we do not have to assume either
    $\bestguess{\omega}{1}$ or $\kappa$ to be larger than the other; in other words, we do not have to perform standard adiabatic elimination with $\kappa$ before averaging over $\bestguess{\omega}{1}$ or conversely.

Equation~\eqref{eq:reduced_model} can rightfully be called a reduced model for the induced decoherence on T,
since we have eliminated both the coupling to the TLS from the description, as well as a fast periodic micromotion given by $\mathcal{K}_{g}(\cdot,t)$.
We observe (see appendix~\ref{reduced_formulas}) that the first-order slow dynamics $\mathcal{L}_{s,1}$ is purely Hamiltonian.
Since Hamiltonian contributions
can by definition be calibrated and do not represent the decoherence we want to study, we will not discuss them here.
The leading-order decoherence process is of second order, represented by $\mathcal{L}_{s,2}$.
The remainder of this section will thus focus on the effectiveness of drives on E in reducing induced dissipation on T,
by examining the dependence of the decoherence operators in $\mathcal{L}_{s,2}$ on the QDD parameters of our proposal \eqref{eq:EAD_control}. A discussion of the Hamiltonian terms in $\mathcal{L}_{s,1}$ and $\mathcal{L}_{s,2}$ including Hamiltonian terms can be found in Appendix~\ref{reduced_formulas}.

\subsubsection{Strong driving}\label{sec:main_strong_driving}

A full derivation of the second-order reduced model corresponding to~\eqref{eq:model_strong},\eqref{eq:reduced_model}
can be found in Appendix~\ref{sec:reduc_strong}. It takes the following form, with some Hamiltonian $H_s$ which we do
not discuss here, and dissipation in operators inherited from the coupling $H_{TE}$:
\begin{align}
    \hspace{-0.1cm} \tfrac{d}{dt} \rho_s & \simeq -i g [H_s,\rho_s]\nonumber\\
    \hspace{-0.1cm} &+ \kappa_{s,z} \mathcal{D}_{T_z}(\rho_s) + \kappa_{s,-} \mathcal{D}_{T_-}(\rho_s) + \kappa_{s,+} \mathcal{D}_{T_+}(\rho_s).\label{eq:id1a}
\end{align}
The decoherence rates $\kappa_{s,z}, \kappa_{s,\pm}$ are given by
\begin{subequations}
    \begin{align}
        \kappa_{s,z} &= - 2 g^2 \mathrm{Re} \qty(\Tr(\sigma_z X_z)), \label{eq:ASn1} \\
        \kappa_{s,{\pm}} &= - 2 g^2 \mathrm{Re} \qty(\Tr(\sigma_{\pm} X_{\mp})),
    \end{align}
\end{subequations}
where $X_z$ and $X_{\mp}$ respectively satisfy the following matrix equations:
\begin{subequations}
    \label{eq:def_kpm}
    \begin{align}
        \qty(\sigma_z - \Tr(\sigma_z \bar{\rho}_E)) \bar{\rho}_E &= - \frac{i}{2} \comm{\omega_2 \sigma_x + \NDelta \sigma_z}{X_z}\label{eq:def_kpm_top}\\
                                                                    &+ \kappa_- \mathcal{D}_{\sigma_-}(X_z) + \kappa_+ \mathcal{D}_{\sigma_+}(X_z), \nonumber\\
        \qty(\sigma_{\mp} - \Tr(\sigma_{\mp} \; \bar{\rho}_E)) \bar{\rho}_E &=- \frac{i}{2} \comm{\omega_2 \sigma_x + \NDelta \sigma_{z}}{X_{\mp}}\label{eq:def_kpm_bottom}\\
                                                                   \pm i \bestguess{\omega}{1} X_{\mp} &+ \kappa_- \mathcal{D}_{\sigma_-}(X_{\mp}) + \kappa_+ \mathcal{D}_{\sigma_+}(X_{\mp})\nonumber.
    \end{align}
\end{subequations}
Here, $\bar{\rho}_E$ is the unique steady state of the Lindbladian acting on E, namely:
\[- \frac{i}{2} \comm{\omega_2 \sigma_x + \NDelta \sigma_z}{\bar{\rho}_E}
+ \kappa_- \mathcal{D}_{\sigma_-}(\bar{\rho}_E) + \kappa_+ \mathcal{D}_{\sigma_+}(\bar{\rho}_E) = 0.\]
A unique solution for $X_{z,\pm}$ is guaranteed by the formalism in Appendix~\ref{sec:ae_floquet}. Given the number of variables in play, expressions
for the dissipation rates are algebraically complicated and computed with the help of a computer algebra system (SymPy~\cite{10.7717/peerj-cs.103}).
As a concrete result of this section, and in line with the double-timescale QDD proposal detailed in section~\ref{sec:qdd_proposal},
we can focus on the limiting case of strong driving, where $\omega_2$ dominates the other parameters.
\begin{thm}\label{thm:first}
    Define $\frac{1}{\Omega_2^k}$ to signify any dimensionless term consisting of the product of $\frac{1}{\omega_2^k}$ with positive powers of the other rates $\bestguess{\omega}{1}, \kappa_{\pm}$ or $\NDelta$ excluding $\omega_2$.
    The decoherence rates defined by~\eqref{eq:def_kpm} display the following asymptotic behavior for large $\omega_2$:
    \begin{subequations}
    \label{eq:id1b}
    \begin{eqnarray}
        \kappa_{s,z} &=& \frac{(\kappa_- + \kappa_+) g^2}{\omega_2^2} + 4 \frac{\NDelta^2}{\omega_2^2} \frac{{g}^2}{\kappa_- + \kappa_+}\\
        &+& \frac{g^2}{\omega_2} \mathcal{O}\qty(\frac{1}{\Omega_2^3}) \nonumber \\
        \kappa_{s,\pm} &=&  \frac{\qty(\kappa_- + \kappa_+) g^{2}}{{(\kappa_- + \kappa_+)}^{2} + 4 \bestguess{\omega}{1}^{2}}\\
        &+& \frac{g^2}{\omega_2} \mathcal{O}\qty(\frac{1}{\Omega_2}) \; . \nonumber
    \end{eqnarray}
    \end{subequations}
\end{thm}

All these rates vanish in the limit $\bestguess{\omega}{1} \ll \omega_2 \rightarrow \infty$, quantitatively confirming the QDD benefits.
The general form of these expressions can be understood intuitively as follows. The expressions involve the sum
$(\kappa_- +\kappa_+)$ because the strongest drive $\omega_2 \sigma_x$ constantly exchanges the roles of ground and
excited states in E. This also explains why $\kappa_{s,+} \simeq \kappa_{s,-}$. The rates $\kappa_{s,\pm}$ then take the
standard Purcell-type expression resulting from Jaynes-Cummings type coupling under detuning
$\frac{\bestguess{\omega}{1}}{2}$. The main QDD effect here is just the $\bestguess{\omega}{1}$-detuning reducing the
effective coupling between T and E. The first term of $\kappa_{s,z}$ in fact has a similar form, where $\kappa_-,
\kappa_+$ terms don't appear in the denominator because they are dominated by $\omega_2^2$. This is no coincidence,
since the Hamiltonian part is like the usual Jaynes-Cummings coupling; up to exchanging the roles of $\sigma_x$ and
$\sigma_z$.  Indeed, neglecting the detuning $\NDelta$, we are applying a constant drive along the $\sigma_x$ direction
(in the $\bestguess{\omega}{1}$ rotating frame), orthogonal to the coupling in the $\sigma_z$ direction. Those two
contributions would not be present if we were only considering the average coupling as derived in \eqref{eq:avHam}. They thus express the limitations, in presence of $\kappa_{\pm}$, of the RWA performed in Section~\ref{sec:qdd_proposal}. The effect of the average coupling remaining in \eqref{eq:avHam} is captured by the second term of $\kappa_{s,z}$. One can recognize the standard induced dissipation formula of type ``$\tilde{g}^2 / \tilde{\kappa}$'' where $\tilde{g}$ is replaced by the average coupling $g \frac{\NDelta}{\omega_2}$ as derived in Section~\ref{sec:qdd_proposal}.

We recall that, behind these interpretations, purely mathematical derivations of the formulas \eqref{eq:id1a}, \eqref{eq:id1b} are detailed in Appendix~\ref{sec:reduc_strong}.

\subsubsection{Case of ultra-strong driving}\label{sec:dissipation_model}

A Lindbladian dissipation model like Eq.\eqref{eq:genmodel} is an idealization meant to summarize interactions of the
TLS with further external degrees of freedom, e.g.~a large bath involving phonon modes. Therefore, when significantly
modifying the system Hamiltonian, in other words when we choose to add ``ultra-strong'' QDD drives on the TLS, the
dissipation model may have to be revised, depending on the type of bath and noise spectrum behind its derivation. One
might be tempted to design QDD drives to purposefully modify the Lindbladian
itself~\cite{Szczygielski2015,FonsecaRomero,Fanchini2007a}. However, in the context of the present work this is
typically a secondary effect. The present section provides explicit formulas for such bath reconsideration, in order to
check to which point our conclusions of Thm.~\ref{thm:first} remain consistent.

We thus leave aside system T for a while and go back to the lab frame for the TLS system E in order to reconsider its decoherence channels.
We can safely neglect the coupling of E and T at this stage, as it involves a weaker Hamiltonian,
even weakened by the QDD drives, and it would thus only appear at higher orders in any possible modification of the Lindbladian dissipator of E.
We model the TLS relaxation as stemming from an interaction of E with a large bath B that can be assumed memoryless.
For the sake of concreteness, the interaction Hamiltonian is taken to be $\gamma \sigma_x \otimes R$, thus
\begin{equation}\label{eq:asLindbath}
H_{EB} = \frac{\Omega_E}{2} \sigma_z  + \tilde{H}_c(t) + \gamma \sigma_x \otimes R + H_B \; .
\end{equation}
Here R is a Hermitian operator acting on the bath Hilbert space, $\gamma$ is some small positive coupling rate,
$H_B$ is the bare bath Hamiltonian and $\tilde{H}_c(t)$ is the QDD drive, expressed back in the lab frame. For this reason, \eqref{eq:asLindbath} also includes the TLS bare frequency $\Omega_E$.
As is common practice, we can consider a bath of harmonic oscillators, for which the coupling along $\sigma_x$ leads to
a Jaynes-Cummings-type interaction with the different modes; similar conclusions hold for more general couplings and
baths~~\cite{cohenT_book2F,Breuer2007,Qnoise}.
We next summarize the results, while details of their derivation can be found in Appendix~\ref{app:lindbladian_derivation}.

\enlargethispage{\baselineskip}
As a first step in obtaining a Lindbladian model, we perform
the Born-Markov approximation in the interaction frame of the TLS and the bath. This interaction frame must include all the dominant Hamiltonians,
it thus involves a rotating frame w.r.t\ the bath Hamiltonian, but also the toggling frame defined by \eqref{eq:control_propagator}, to include the drives on the TLS part.
Next, we perform a standard secular approximation (RWA), averaging over the frequencies $\pm \bestguess{\Omega}{E}$.
The RWA introduces an error of order $\frac{\kappa^2}{\bestguess{\Omega}{E}}$, where $\kappa$ is the typical dissipation rate obtained in the end.
Since we assume the bare frequency of the TLS to be much larger than the dissipation rate, we can neglect this term.
A final approximation is needed to obtain a Lindbladian model.
There are two possibilities for this, and for any value of $\omega_2$, at least one of them is valid in the context of our QDD protocol.

As a first possible condition,  when the drive amplitude $\omega_2$ is dominated by the bare qubit frequency $\Omega_E$,
the noise spectral density $G$ of the bath (defined in~\eqref{eq:defn_spectral_density}) can typically be considered flat in the ranges $\pm [\Omega_E-\omega_2, \; \Omega_E+\omega_2]$.
The Jaynes-Cummings type coupling assumed in \eqref{eq:asLindbath} then yields stationary dissipators in $\sigma_-$ and
$\sigma_+$, as assumed in Section~\ref{sec:model}.:
\begin{equation*}
    L_k \in \qty{ \sigma_{-},\quad \sigma_{+}\; },
\end{equation*}
with respective rates $\kappa_{\mp} \simeq 2 \gamma^2\, G(\pm \Omega_E)$.

The second possible approximation for obtaining a Lindbladian model is a second RWA, now over frequencies $\pm \Lambda$.
This approximation remains valid as long as $\Lambda$ is much larger than the obtained dissipation rate, to be checked a posteriori.
For our TLS system coupled to the bath, this yields (see appendix) decoherence through the three dissipation operators 
\begin{equation}
    L_k \in \qty{ \sigma_{\alpha x}, \quad \sigma_{\alpha -},\quad \sigma_{\alpha +}\; } ,
\end{equation}
as defined in Section \ref{sec:qdd_proposal}, with respective decoherence rates:
\begin{subequations}
\label{eq:ultrastrong_dissipators}
\begin{align}
    \kappa_{\alpha x} &= \frac{\gamma^2}{2} (G(\bestguess{\Omega}{E} + \bestguess{\omega}{1}) + G(- \bestguess{\Omega}{E} - \bestguess{\omega}{1})) \cos^2(\alpha), \\ \nonumber
    \kappa_{\alpha-} &= \frac{\gamma^2}{2} G(\qty(\bestguess{\Omega}{E} + \bestguess{\omega}{1}) + \Lambda) {(1 + \sin(\alpha))}^2\nonumber\\
                    &+\frac{\gamma^2}{2} G(- \qty(\bestguess{\Omega}{E} + \bestguess{\omega}{1}) + \Lambda) {(1 - \sin(\alpha))}^2, \\ \nonumber
    \kappa_{\alpha+} &= \frac{\gamma^2}{2} G(- \qty(\bestguess{\Omega}{E} + \bestguess{\omega}{1}) - \Lambda) {(1 + \sin(\alpha))}^2 \nonumber \\
                    &+\frac{\gamma^2}{2} G(\qty(\bestguess{\Omega}{E} + \bestguess{\omega}{1}) - \Lambda) {(1 - \sin(\alpha))}^2\label{eq:laatste}  \; .
\end{align}
\end{subequations}

The choice between a model with fixed decoherence operators $L_k \in \{\sigma_- , \, \sigma_+ \}$, or with
drive-corrected ones $L_k \in \{ \sigma_{\alpha x}, \; \sigma_{\alpha -},\; \sigma_{\alpha +} \}$, depends on
whether it is a better approximation to consider $G$ flat on the scale of $\omega_2/\Omega_E$,
or to consider an RWA based on $\Lambda \gg \kappa_{\alpha x}, \kappa_{\alpha-}, \kappa_{\alpha+}$. The former approach leads to an error of order $\kappa \frac{\Lambda}{\Omega_E}$,
whereas the latter leads to an error of order $\frac{\kappa^2}{\Lambda}$.

The two approximations are compatible and commute with one another when both are justified, i.e.~when $\kappa_{\pm, \alpha \pm, \alpha x} \ll \omega_2 \ll \Omega_E$. Indeed, first assuming a locally flat bath spectrum, next transforming the $\sigma_-$ and $\sigma_+$ dissipators to the rotating frame w.r.t.~$\frac{\Lambda}{2} \sigma_{\alpha x}$, and finally performing RWA over frequencies $\pm\Lambda$,
yields exactly the dissipators associated to~\eqref{eq:ultrastrong_dissipators} with $\omega_1$ and $\Lambda \simeq \omega_2$ put to zero in the bath spectrum $G$. In contrast, we can also see that the two approaches do give different results in some situations. For instance, for $\alpha=0$ and $G$ depending on frequencies on the scale of $\omega_2$, the dissipation rates $\kappa_{\alpha \pm}$ along the $\pm 1$ eigenvectors of $\sigma_{\alpha x}=\sigma_x$, thus obtained using the second approximation, would differ (slightly). Such asymmetry cannot be retrieved as an average effect of $\omega_2 \sigma_x$ driving on given $\sigma_{\pm}$ dissipators, as would result from the first type of approximation. Thus in this case, applying the correct (second) type of approximation results in genuine corrections to the Lindbladian. 

In summary, when the first type of approximation is justified, we retrieve the original model and the induced
dissipation of Thm.~\ref{thm:first}. When $\omega_2$ becomes too large (ultra-strong driving) and only the second type of approximation is justified, we must revise the dissipation model. In the rest of this section, we derive formulas for the induced dissipation on T under this revised dissipation model and just considering general, non-vanishing rates $\kappa_{\alpha x,\alpha -,\alpha +}$.

Again in a rotating frame w.r.t. $\frac{\bestguess{\omega}{1}}{2} \sigma_z$, the joint evolution of the target and TLS is thus described by the master equation
\begin{align}
    \tfrac{d}{dt}\rho &=\kappa_{\alpha-} \mathcal{D}_{\mathbb{1}_T \otimes \sigma_{\alpha-}}(\rho)
         + \kappa_{\alpha+} \mathcal{D}_{\mathbb{1}_T \otimes \sigma_{\alpha+}}(\rho)\label{eq:model_ultrastrong}\\
    & - i \frac{\Lambda}{2} \comm{\mathbb{1}_T \otimes \sigma_{\alpha x}}{\rho}
     + \kappa_{\alpha x} \mathcal{D}_{\mathbb{1}_T \otimes \sigma_{\alpha x}}(\rho)\nonumber\\
    &- i g \comm{T_z \otimes \sigma_z + e^{i \bestguess{\omega}{1} t} T_- \otimes \sigma_+ + e^{- i \bestguess{\omega}{1} t} T_+ \otimes \sigma_-}{\rho}. \nonumber
\end{align}
A full derivation of the reduced model corresponding to~\eqref{eq:model_ultrastrong}
can be found in Appendix~\ref{sec:reduc_strong}, including all terms in
$\mathcal{K}_1, \mathcal{L}_{s,1}$ and $\mathcal{L}_{s,2}$.
Again the expressions are algebraically complicated and computed with the help of a computer algebra system (SymPy).
We here report simplified formulas in the limit where $\omega_2$ is the fastest timescale in the joint system.
The leading-order decoherence process contains the same dissipators as in~\eqref{eq:id1a}, thus
\begin{align}\label{eq:ASultrastrongDinddiss}
    \hspace{-0.1cm} \tfrac{d}{dt} \rho_s &\simeq  -i g [H_s,\rho_s]\nonumber\\
    \hspace{-0.1cm} &+  \kappa_{s,z} \mathcal{D}_{T_z}(\rho_s) + \kappa_{s,-} \mathcal{D}_{T_-}(\rho_s)+ \kappa_{s,+} \mathcal{D}_{T_+}(\rho_s).
\end{align}
Our main result consists of the formulas for the dominating order of the decoherence rates.
\begin{thm}
Consider the same notation $\frac{1}{\Omega_2^k}$ as in Thm.~\ref{thm:first}.
The induced decoherence rates associated to \eqref{eq:ASultrastrongDinddiss} for the model \eqref{eq:model_ultrastrong} display the following asymptotic behavior for large $\omega_2$:
\begin{subequations}
\label{eq:id2a}
\begin{eqnarray}
    \kappa_{s,z} &=&  \frac{(\kappa_{\alpha_\Sigma} + 4 \kappa_{\alpha x}) g^2}{\omega_2^2} \nonumber\\
    &+& 2 \frac{\NDelta^2}{\omega_2^2} \frac{g^2 \left(1 - \frac{\kappa_{\alpha_\Delta}^{2}}{\kappa_{\alpha_\Sigma}^{2}} \right)}{\kappa_{\alpha_\Sigma}}
    + \frac{g^2}{\omega_2} \mathcal{O}\qty(\frac{1}{\Omega_2^2}) ,\\ 
    \kappa_{s,\pm} &=& \frac{\kappa_{\alpha_\Sigma} g^{2} \left(1 - \frac{\NDelta^{2}}{\omega_{2}^{2}} \right)
            \left(1 - \frac{\kappa_{\alpha_\Delta}^{2}}{\kappa_{\alpha_\Sigma}^{2}} \right)}{2 \left(\kappa_{\alpha_\Sigma}^{2} + \bestguess{\omega}{1}^{2}\right)}\nonumber\\
    &&+ \frac{g^{2} \left(4 \kappa_{\alpha x} + \kappa_{\alpha_\Sigma}\right)}{4 \omega_{2}^{2}} + \frac{g^2}{\omega_2}\mathcal{O}\qty(\frac{1}{\Omega_2^2}),
\end{eqnarray}
\end{subequations}
with $\kappa_{\alpha \Sigma} = \kappa_{\alpha -} + \kappa_{\alpha +}$ and $\kappa_{\alpha \Delta} = \kappa_{\alpha -}  - \kappa_{\alpha +} $.
\end{thm}

These rates can be understood intuitively in a similar way as for \eqref{eq:id1b}.
The extra factor  $(1 - \frac{\kappa_{\alpha_\Delta}^{2}}{\kappa_{\alpha_\Sigma}^{2}} = 1 - x_{\alpha, \infty}^2)$
accounts for the generally nonzero average value $x_{\alpha, \infty}$ of $\sigma_{\alpha x}$ in the TLS steady state.
A larger $x_{\alpha, \infty}$ reduces the dissipative part at the expense of a deterministic, Hamiltonian term (see
Appendix ~\ref{sec:calculations_ultra_strong}).
In $\kappa_{s,\pm}$, we have now kept a term of order $1/\omega_2^2$ because the dominating contribution of $\kappa_{\alpha x}$ only appears at this order.

Taking into account the modified dissipation model for E thus does affect induced decoherence for T,
with significant changes if $\kappa_{\alpha -}  \gg \kappa_{\alpha +}$ such that $\kappa_{\alpha \Delta} \simeq \kappa_{\alpha \Sigma}$.
However, with a bath model at the origin of \eqref{eq:ultrastrong_dissipators}, this would only happen under very peculiar conditions.
The standard conclusions with a reasonably flat bath noise spectrum, and $\alpha \ll 1$, are not too different from \eqref{eq:id1b}. They quantitatively confirm the QDD benefits under this model too.

\subsubsection{Optimization: cold TLS and reducing $\omega_2$} \label{sec:optimization_lower_omega_2}

The general formulas \eqref{eq:id1b} and \eqref{eq:id2a} quantify how QDD controls containing two drives with amplitudes $\omega_2 \gg \omega_1$ reduce the decoherence induced on T under general conditions. They can guide parameter choices in particular situations, as long as we assume large $\omega_2$. Having large $\bestguess{\omega}{1}$ and $\omega_2$ is always beneficial.

However, this does not mean that driving strongly in both $\bestguess{\omega}{1}$ and $\omega_{2}$ is always the best choice. Indeed, in very particular settings, it may be even better to take some of the drives at their minimal value; in other words, intermediate values of the drives would be the worst case.
Assume for instance the extreme situation of dispersive coupling to a zero-temperature bath, i.e.~$T_x= T_y = \kappa_+ = 0$. Then, in absence of controls
(in fact as long as $\omega_2=0$), the TLS is attracted towards its ground state, and the resulting effect on T would be purely Hamiltonian.  This raises the question of how to choose $\omega_2$ to minimize the $T_z$-decoherence. We next answer this question, as an illustration of how to use our framework for design choices.

We therefore reconsider the exact rate of the $T_z$-decoherence channel at second order adiabatic elimination, valid as
long as $\kappa \gg g$ and $\bestguess{\omega}{1} \gg g$. This is the solution of \eqref{eq:ASn1},\eqref{eq:def_kpm_top}, thus assuming the model \eqref{eq:model_strong}, without considering the limit of large $\omega_2$:
\begin{widetext}
\begin{equation}\label{eq:k_sz_exact}
    \kappa_{s,z} = \frac{2 g^{2} \left(4 \NDelta^{2} + {\qty(\kappa_- + \kappa_+)}^{2}\right) \left( 4 \kappa_+ \kappa_- \left(16 \NDelta^{2} \omega_{2}^{2} + \left(4 \NDelta^{2} + {\qty(\kappa_- + \kappa_+)}^{2}\right)^{2}\right) + 4 {\qty(\kappa_- + \kappa_+)}^{2} \omega_{2}^{2} \left(2 \kappa_-^{2} + 2 \kappa_+^{2} + \omega_{2}^{2}\right)\right)}{{\qty(\kappa_- + \kappa_+)}^{3} \left(4 \NDelta^{2} + {\qty(\kappa_- + \kappa_+)}^{2} + 2 \omega_{2}^{2}\right)^{3}} \; .
\end{equation}
\end{widetext}
The bath temperature is characterized by $n_\textrm{th}$, the mean number of thermal photons, such that $\kappa_- = \kappa_1 (1 + n_\textrm{th})$ and $\kappa_+ = \kappa_1 n_\textrm{th}$.  Straightforward algebraic manipulations of~\eqref{eq:k_sz_exact} allow for an optimization study, which we summarize in the following result.
\begin{thm}\label{thm:cold_optimization}
The induced dissipation rate $\kappa_{s,z}$ as defined in~\eqref{eq:k_sz_exact} shows the following dependence on $\omega_2$:
    \begin{itemize}
    \item If  $\; n_\textrm{th} < \frac{\sqrt{3}}{3} - \frac{1}{2} \simeq 0.077 \; $,  then $\kappa_{s,z}$ displays a single local maximum as a function of $\omega_2$, for any values of $\NDelta$ and $\kappa_1$. The optimal value of $\omega_2$ is either zero or the maximal achievable one, depending on the experimentally achievable bound on the latter. 
    \item If $n_\textrm{th} > \frac{\sqrt{3}}{3} - \frac{1}{2} \simeq 0.077$ and  
    \begin{equation}\label{eq:detuning_threshold}
            \frac{\NDelta^2}{\kappa_1^2} < \tfrac{\left(2 n_\textrm{th} + 1\right)^{2} \left(2 \sqrt{3} \left(2 n_\textrm{th} + 1\right) + \sqrt{12 n_\textrm{th}^{2} + 12 n_\textrm{th} - 1}\right)}{4 \sqrt{12 n_\textrm{th}^{2} + 12 n_\textrm{th} - 1}}\; ,
       \end{equation}
    then $\kappa_{s,z}$ also displays a single local maximum as a function of $\omega_2$, with the same conclusions for its optimization.
    \item If $n_\textrm{th} > \frac{\sqrt{3}}{3} - \frac{1}{2} \simeq 0.077$ and \eqref{eq:detuning_threshold} is not satisfied, then $\kappa_{s,z}$ is monotonically decreasing in $\omega_2$.
    \end{itemize}
\end{thm}
In the last case, in other words when $\Delta$ is large, ramping up $\omega_2$ is always advantageous. In the first two cases, the value of $\omega_2$ minimizing $\kappa_{s,z}$ will thus depend on how its value at the maximal achievable $\omega_2$ compares to its value at $\omega_2=0$, which reads:
    \begin{equation}\label{eq:k_sz_0}
        \kappa_{s,z}(\omega_2 = 0) = \frac{8 g^{2} n_\textrm{th} \left(n_\textrm{th} + 1\right)}
        {\kappa_{1} \left(8 n_\textrm{th}^{3} + 12 n_\textrm{th}^{2} + 6 n_\textrm{th} + 1\right)} \; .
    \end{equation}
A numerical illustration of the dependence of $\kappa_{s,z}$ on $\omega_2$ and $n_{\textrm{th}}$ is provided on Figure \ref{fig:cold_bath}. 

\begin{figure}
    \centering
        \includegraphics[width=1.0\linewidth]{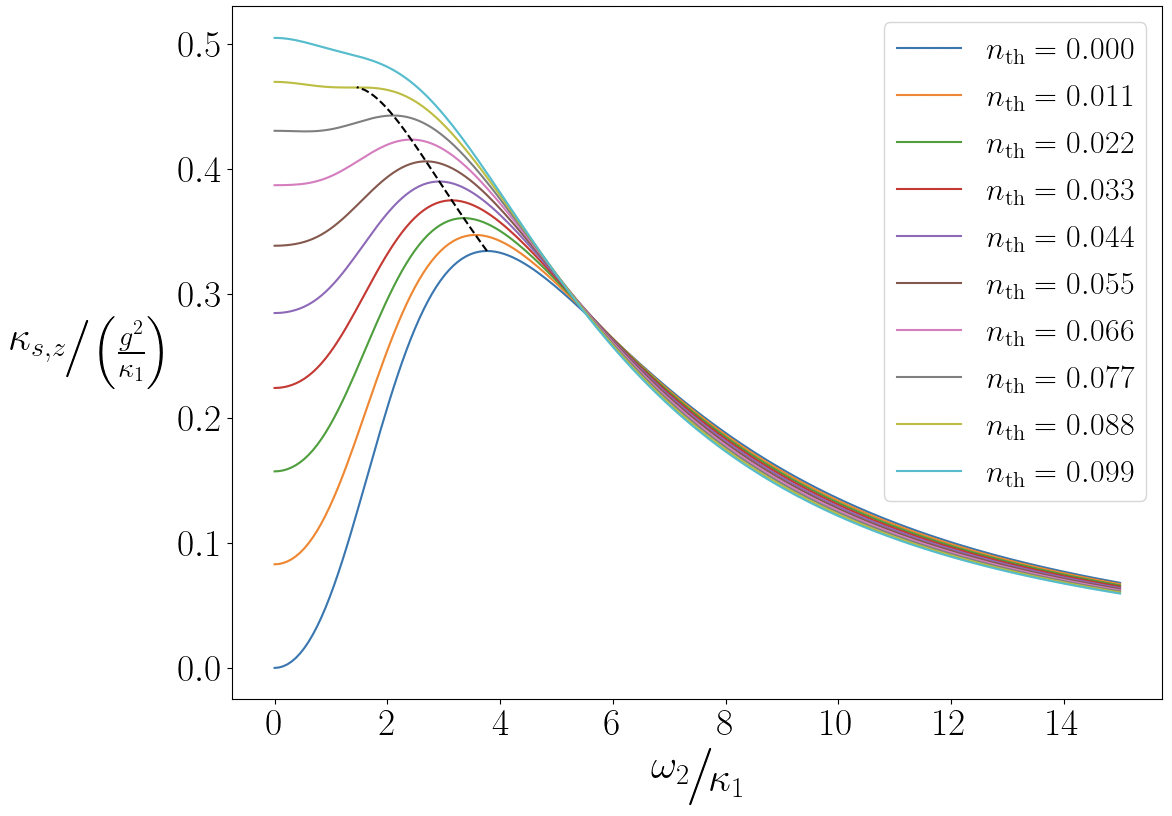}
        \caption{Dependence of $\kappa_{s,z}$ as defined in~\eqref{eq:k_sz_exact} on $\frac{\omega_2}{\kappa_1}$, for various values of  $n_\textrm{th}$ between $0$ and $0.1$ and at fixed $\frac{\NDelta}{\kappa_1} = {2}$. The dashed black line indicates the local maximum in $\omega_2$ when it is present. In the latter cases, induced dissipation $\kappa_{s,z}$ is minimized either at $\omega_2=0$ or at the largest achievable $\omega_2$, depending on its value.}\label{fig:cold_bath}
    \end{figure}

Note that Theorem~\ref{thm:cold_optimization} and Figure \ref{fig:cold_bath} have been established by analyzing a single algebraic formula.
Indeed, the adiabatic elimination method on one hand yields explicit formulas for the induced dissipation rates, preventing the need for solving differential equations for each parameter setting.
On the other hand, our extended formalism as explained  in appendix~\ref{sec:ae_floquet} does not require to select between either the dissipation being the largest time-scale
(standard adiabatic elimination), or drive frequencies $\bestguess{\omega}{1}, \omega_2$ being the largest timescale (domain of averaging techniques like RWA). A single formula thus allows to consistently cover the full range of parameter values.

We can also look at the values of induced decoherence rates $\kappa_{s,\pm}$ for $\omega_2=0$, yielding:
\begin{subequations}
\label{eq:om20rate}
\begin{align}
        \kappa_{s,-} &= \frac{4 \kappa_{-} g^{2}}{{(\kappa_- + \kappa_+)}^2 + 4 {(\bestguess{\omega}{1} + \NDelta)}^2},\\
        \kappa_{s,+} &= \frac{4 \kappa_{+} g^{2}}{{(\kappa_- + \kappa_+)}^2 + 4 {(\bestguess{\omega}{1} + \NDelta)}^2}.
\end{align}
\end{subequations}
As expected when taking $\omega_2=0$, the initial $\kappa_-$ and $\kappa_+$ remain separated, such that $\kappa_{s,+}$ remains small for a cold bath. We can also note that $\omega_2$ was decoupling the effect of $\NDelta\; \sigma_z$, and its absence reintroduces this detuning in addition to $\bestguess{\omega}{1}$ in \eqref{eq:om20rate}. Regarding the QDD effect, in both cases we rely on large $\bestguess{\omega}{1}$ to reduce the induced decoherence. In absence of $\omega_2$ however,
there may be a danger of being counterproductive by hitting $\bestguess{\omega}{1} \approx -\NDelta$.\\

Such calculations illustrate how our formulas could be used to optimize the parameter choices. In Section \ref{sec:noisy_decoupling}, we will instead optimize the relative strengths of \emph{dissipation} parameters in various situations, showing that the optimum can similarly jump from strongest possible to weakest possible ``shaking'' on some channels. For instance, expressions like \eqref{eq:k_sz_0} seem to indicate that even in absence of any drives, the lowest temperature (i.e. value of $n_{\text{th}}$) is not necessarily inducing the lowest $\kappa_{s,z}$; this specific example will be treated in Section~\ref{sssec:ex2:disp}. Before this, we investigate more generally how induced dissipation on T can be decreased by acting on E not only with coherent drives, but also with additional decoherence.

\section{Purely dissipative decoupling}\label{sec:noisy_decoupling}

In Section~\ref{sec:coherent_decoupling} we are arguing that adding not too precise drives on E can decrease the decoherence that it induces on the target system T,  as environment and system can be seen to dynamically decouple. It is tempting to push this idea one step further and ask: what happens if we drive E with a Hamiltonian $\opus{H}$ whose amplitude is pure noise? In fact, this brings us to asking whether we can achieve dynamical decoupling (or whether we can lower the induced decoherence, if you prefer to be more cautious with the naming QDD) by \emph{adding decoherence}, in the form of Lindblad operators, to the environment. Indeed, taking $dW_t$ the increment at time $t$ of a Brownian motion process, we have:
\begin{align*}
    \rho(t + dt) =& \mathbb{E}\left( e^{-i \opus{H} dW_t} \rho(t)  e^{i \opus{H} dW_t} \right)\\
     =& \rho(t)   -i [\opus{H},\rho(t)] \, \mathbb{E}(dW_t)\nonumber \\
     +&     \qty(\, \opus{H} \rho(t) \opus{H}  - \frac{\opus{H}^2}{2} \rho(t) - \rho(t) \frac{\opus{H}^2}{2} \,)  \; \mathbb{E}(dW_t^2) \\
     =& \rho(t) +  \mathcal{D}_H(\rho(t)) \, dt   \; ,
\end{align*}
which is a Lindblad equation with Hermitian decoherence channel $\opus{H}$. Adding such decoherence could be considered a ``legitimate hack'' in the sense that it increases entropy production on the environment. One can also consider adding non-Hermitian decoherence channels, like the qubit relaxation channel $\mathcal{D}_{\opus{\sigma_-}}$. While requesting to add a very strong such cooling on the environment is most probably not experiment-friendly, a more reasonable setting of this type could be: let E be subject to
\begin{equation}\label{Aseq:thermalba}
\mathcal{D}_{\text{total}} =  \kappa_1 (1+n_{th}) \mathcal{D}_{\opus{\sigma_-}} +  \kappa_1 n_{th} \mathcal{D}_{\opus{\sigma_+}}\; ,
\end{equation}
with lower and upper bounds on $\kappa_1$ and on $n_{th}$; which parameter choice minimizes the induced decoherence on T? More generally, we can consider settings where the environment is subject to decoherence
\begin{equation}\label{ASeq:generalkappas}
\mathcal{D}_{\text{total}} = \sum_k \, \kappa_k  \mathcal{D}_{\opus{L}_k} \; ,
\end{equation}
with the rates $\kappa_k$ of the decoherence channels jointly tunable within a given set. The way in which these $\kappa_k$ are tuned in practice can depend on the particular experiment. They might result from (noisy) drives and secular approximations, or for models like   \eqref{Aseq:thermalba} they might guide target values of $\kappa_1$ and $n_{th}$ at the experiment design stage.

We will thus focus on a setting where the joint state of the system and environment are described by a master equation like \eqref{eq:genmodel} with $H_T=0$, i.e.
\begin{equation}\label{eq:Psec:setting}
\tfrac{d}{dt} \rho = -i[\opus{H_{TE}}+\opus{H_E},\rho] + \mathcal{D}_{\text{total}}  \; ,
\end{equation}
where now $\opus{H_E}$ is fixed and dissipation takes the form \eqref{ASeq:generalkappas} with tunable $\kappa_k$ associated to $L_k$ operators acting on E only. Future work may want to add a tunable $\opus{H_E}$ as in Section~\ref{sec:coherent_decoupling}, but for a more efficient presentation we here study both effects separately (see~\cite{Forni2019} for examples on tuning both the $\kappa_k$ and a time-independent drive $H_E$). The question is again: what are the values of the $\kappa_k$ which minimize the induced decoherence on T? In absence of time-dependent drives, this induced decoherence can be computed directly with standard adiabatic elimination formulas~\cite{Azouit2017b}. We will show that the QDD principle carries through at this purely dissipative level. Namely, selecting large values of $\kappa_k$, which strongly shake the environment, can lead to much less induced decoherence than selecting the $\kappa_k$ which make the environment as pure as possible.

This section is organized as follows. After quickly recalling the required formulas, we establish some properties on general systems, then we address some typical settings for a TLS as environment. A preliminary and different presentation of these ideas can be found in the conference proceedings~\cite{Forni2019a}.

\subsection{Adiabatic elimination formulas}\label{ssec:formulas}

We here recall the standard adiabatic elimination setting in absence of drives, to make this self-contained for readers skipping Section \ref{sec:coherent_decoupling}.

The general purpose of adiabatic elimination is to eliminate all fast dissipative dynamics and only retain the degrees of freedom which evolve slowly, and which are thus best protected from decoherence. A standard setting is when a fast dissipating system (here E) is weakly coupled to another system (here T). Under appropriate conditions, the coupling induces a weak hybridization of the two subsystems, in which a subsystem close to T can be viewed as an autonomous system of state $\rho_s$ undergoing slow decoherence and slow Hamiltonian dynamics. Approximation formulas have been developed to compute this hybridization and slow dynamics at various orders \cite{Azouit2017b}.

We focus on the formulas expressing the \emph{dissipation} on $\rho_s$, thus induced on T by its coupling to E, taking the viewpoint that a constant hybridization and Hamiltonian can, by definition, be calibrated. The following procedure gives the dominating terms of the dissipation, provided the environment E alone has a unique steady state and the latter is strongly attractive compared to the coupling rate with T.
\begin{itemize}
\item Denote by $\bar{\rho}_E$ the unique steady state of the environment dynamics, thus taking $\opus{H_{TE}}=0$ in \eqref{eq:Psec:setting}.
\item Writing $\opus{H_{TE} = \sum_k  T_k \otimes E_k}$, for each $E_k$ compute $\tilde{E}_k \bar{\rho}_E = E_k \bar{\rho}_E - \Tr(E_k \bar{\rho}_E) \bar{\rho}_E$.
\item For each $k$, solve for a traceless operator $Q_k$ in  $\;\; -i[H_E, Q_k ] + \mathcal{D}_{\text{total}}(Q_k) = - \tilde{E}_k \bar{\rho}_E \;$.
\item Construct the matrix $X$ with components $X_{k,j} = \Tr(Q_j E_k^\dagger + E_j Q_k^\dagger)$. This matrix is positive semidefinite and the induced dissipation on T, at second order adiabatic elimination, is given by
$$\mathcal{D}_{\text{induced}} = \sum_k \mathcal{D}_{L_k}$$
where $L_k = \sum_j \; \Lambda_{j,k} T_j$ for any decomposition $X = \Lambda \Lambda^\dagger$.

In particular, the induced dissipation at this order of approximation involves just linear combinations of the coupling operators $T_j$ in $H_{TE}$.

In particular, when $H_{TE} = T \otimes E$ and thus $X$ is a scalar, this $X$ just gives the induced dissipation rate on T associated to dissipation operator $T$.
\end{itemize}

We will use the result of this procedure to analyze how the dissipation induced on a target system can be reduced by varying the $\kappa_k$ in \eqref{ASeq:generalkappas}. One should bear in mind that this is only the dominating term, in an approximate formula which is valid when dissipation on the environment is \emph{fast}. Thus, conclusions encouraging us to take minimal dissipation on the environment should be taken with caution. However, we will often encounter the conclusion that more dissipation in the environment is better for the target, and this regime is precisely the one well covered by the adiabatic elimination conditions.

Moreover, when treating the example of a TLS environment in more detail, we will illustrate how to adapt the adiabatic elimination procedure when the dynamics on E alone do not strongly attract it to a unique steady state.

\subsection{Some general properties}\label{ssec:properties}

Before moving to our running example of a TLS environment, we can give some general results on induced dissipation as computed with the above procedure. They are very much in line with the QDD viewpoint that shaking the environment \emph{more} should lead to \emph{less} effect on the target system.

\begin{prop}\label{prop:D_1}
When $\opus{H_E}$ is fixed and of the same order as $\opus{H_{TE}}$ or smaller, multiplying all the $\kappa_k$ by $\alpha>1$ decreases the dissipation induced on the target system by a factor $\alpha$. The same conclusion holds for any $\opus{H_E}$ if it can be multiplied by $\alpha>1$ together with the $\kappa_k$.
\end{prop}
\begin{proof}
There are two ways to consider $\opus{H_E}$. Under the first condition, we consider it as part of the perturbation, with a coupling
$\tilde{\opus{H}}_{TE} = \opus{H_{TE}} + \mathbb{1}_T \otimes  \opus{H_E}$. Under the second condition, we consider it as part of the fast dynamics, and just scale it together with all other Lindbladian contributions on the environment. 
In the adiabatic elimination procedure outlined above, both cases will not change $\bar{\rho}_E$ nor the $\tilde{E}_k$, yet the $Q_k$ will be $\alpha$ times smaller. Thus the induced dissipation matrix $X$ is $\alpha$ times smaller and so will be the rates deduced from it. Another viewpoint is to make a change of timescale re-establishing the initial dissipation rate on the environment. In this new timescale, the coupling to the target system is decreased by $\alpha$. According to the 2nd order adiabatic elimination formulas, the induced dissipation is \emph{quadratic} in the coupling strength, therefore the lower coupling more than compensates for the change of timescale.\\
\end{proof}

\begin{prop}\label{prop:D_2}
    Consider that $H_E=0$ and that all the decoherence channels $L_k$ on the environment are Hermitian and can be tuned individually. Then, as soon as adiabatic elimination conditions are satisfied, the diagonal elements of the induced decoherence matrix $X$ are all minimized by taking the \emph{maximal value of $\kappa_k$ for each $L_k$ on the environment.}
\end{prop}
\begin{proof}
    We say ``as soon as'' because if the conditions are satisfied for some set of parameters, then they still hold when we increase the dissipation rates. We summarize the main ideas of the proof, whose full version is available in~\cite{Forni2019a}. First note that the unique steady state $\bar{\rho}_E$ of the environment under Hermitian $L_k$ must be $\propto \text{Identity}$, irrespective of the tuning choice. One can then write an optimization problem for each diagonal element of $X$, expressing the computations of the adiabatic elimination procedure recalled above as constraints with Lagrange multipliers. The necessary optimality conditions then ensure that $X$ can be minimal only at the extreme values of the $\kappa_k$. A local analysis shows that if its value matters, then $\kappa_k$ must be \emph{maximal} to minimize the $X$-element.
\end{proof}

Since we know that $X$ is positive semidefinite, the implications of Proposition~\ref{prop:D_2} on its diagonal have similar consequences for the induced dissipation rates. 

Finally, we can try to give criteria under which the induced dissipation on T can vanish at the limit where some $\kappa_k$ become very large. Note that this is the limit where the adiabatic elimination becomes more and more valid.\\

\begin{prop}\label{prop:D_3}
Consider \eqref{eq:Psec:setting} 
 and let
$$\mathcal{D}_{\text{total}} =  \mathcal{D}_a + \tfrac{1}{\delta} \mathcal{D}_b  \; ,$$
such that, for fixed $\mathcal{D}_a$ and $\mathcal{D}_b$, the $\kappa_k$ remain within their authorized domain as $\delta$ tends to $0$. In other words, $\mathcal{D}_b$ is the part of the dissipation on E whose rates can possibly tend to infinity. 

If $\mathcal{D}_b$ has a unique steady state, then the decoherence induced on the target system vanishes as $\delta$ tends to $0$. Otherwise, the dynamics on E can first be reduced by \emph{first-order} adiabatic elimination of $\mathcal{D}_b$.  The structure for this can be more general than the formulas of Section~\ref{ssec:formulas}, the procedure is recalled in Appendix \ref{app:ae_formulas}.1. The resulting system can in turn be analyzed, either exactly or with another round of adiabatic elimination, to show if there remains induced dissipation or not.
\end{prop}
\begin{proof}
When $\mathcal{D}_b$ has a unique steady state, we can consider just $\mathcal{D}_b$ as the fast dynamics in order to perform adiabatic elimination of E. The first-order formula gives only Hamiltonian dynamics on T, while the contributions of higher-order adiabatic elimination vanish as $1/\delta$ gets infinite.

When $\mathcal{D}_b$ induces no unique steady state for E, it cannot be used to adiabatically eliminate the whole Hilbert space of E. Instead, we must keep as reduced model a subspace of linear operators on E supporting all the steady states of $\mathcal{D}_b$.

To conclude on the behavior for $\delta \rightarrow 0$, we only need to keep the first-order contributions resulting from adiabatically eliminating $\mathcal{D}_b$. The following possibilities for the reduced system illustrate some possible rapid conclusions.
\begin{itemize}
\item In some cases, the coupling between T and what remains of E after adiabatic elimination of $\mathcal{D}_b$ can vanish; then, there would be no induced dissipation when $1/\delta$ gets infinite.
\item By definition of the original problem, the reduced dynamics on (what remains of) E has a unique steady state $\bar{\rho}_{\tilde{E}}$. In particular, the dissipation there cannot vanish, precluding the possibility of ending up with a purely Hamiltonian joint system on T and $\tilde{\text{E}}$.
\item If the remaining dynamics takes the form of dissipation on what remains of E, with weak Hamiltonian coupling to T, then we can readily apply the second-order adiabatic elimination formulas of Section~\ref{ssec:formulas} to the remaining system. This enables to directly either conclude to the negative (there already remains induced dissipation at second-order), or observe that at least the dominating order of dissipation vanishes (thus according to second-order adiabatic elimination formulas with finite $\mathcal{D}_a$). Note though that it does not seem true that the adiabatic elimination of $\mathcal{D}_b$ would always yield such structure.
\item In particular, in this last setting, if $\bar{\rho}_{\tilde{E}}$ has full rank, then induced dissipation cannot
vanish. If $\bar{\rho}_{\tilde{E}}$ has reduced rank, then the induced dissipation cannot vanish if the Hamiltonian
coupling acts inside the space supported by $\bar{\rho}_{\tilde{E}}$. The proof, worked out in appendix~\ref{sec:proof_D3_last},
follows similar steps as for proving that the dissipation matrix $X$ in Section~\ref{ssec:formulas} is always non-negative (see \cite{Azouit2017b} and related work).
\end{itemize}
\end{proof}

\subsection{Minimizing decoherence induced by a qubit environment (TLS)}

We now focus in more detail on the case of a two-level-system (TLS) environment. We consider Hermitian dissipation channels plus relaxation in a thermal environment:
\begin{eqnarray}
\label{eq:PsecTLS:totdiss}
\mathcal{D}_{\text{total}} &=& \kappa_x  \mathcal{D}_{\opus{\sigma_x}}  +  \kappa_y  \mathcal{D}_{\opus{\sigma_y}}  +  \kappa_z  \mathcal{D}_{\opus{\sigma_z}} \\
\nonumber & +& \kappa_{-}  \mathcal{D}_{\opus{\sigma_-}} + \kappa_+  \mathcal{D}_{\opus{\sigma_+}}  \;.
\end{eqnarray}
We assume that $\kappa_x,\kappa_y,\kappa_z$ are each tunable independently within a given interval e.g.~$[\,\underline{\kappa}_x,\; \overline{\kappa}_x\,]$, while $\kappa_{-} = \kappa_1 \, (1+n_{\text{th}})$, $\kappa_{+} = \kappa_1 \, n_{\text{th}}$,  with typically independent bounds on the coupling strength $\kappa_1$ and temperature characteristic $n_{\text{th}}$.

We will first consider two typical couplings between E and T --- almost-resonant and dispersive --- while assuming that E always has a strongly attractive unique steady state. This happens as soon as two of the $\kappa_x,\kappa_y,\kappa_z$ take significant nonzero values, or $\kappa_1$ takes a significant nonzero value. When this is not the case, we can still apply adiabatic elimination but on a modified state space splitting; we illustrate what this implies for induced dissipation in a third example.

\subsubsection{Dispersive coupling}\label{sssec:ex2:disp}

As a first case, we consider \eqref{eq:PsecTLS:totdiss} in conjunction with the coupling Hamiltonian:
$$ \opus{H}_{TE} = g\, \opus{T_z} \otimes \sigma_z \;.$$
This models the typical situation of dipolar coupling between the target system and a TLS which is far detuned (dispersive coupling limit).

Following the adiabatic elimination procedure, we first compute the steady state of the fast TLS relaxation alone:
$$\bar{\rho}_E = \frac{\mathbb{1}+\bar{z} \opus{\sigma_z}}{2} \;\; \text{ with } \;\; \bar{z} = \tfrac{-\kappa_1}{(1+2 n_{\text{th}}) \kappa_1 + 2(\kappa_x+\kappa_y)} \;.$$
From the coupling operator $\sigma_z$ in E, we then compute
$$\tilde{\sigma}_z \bar{\rho}_E  = \tfrac{1-\bar{z}^2}{2} \sigma_z \;.$$
Next we must solve
$$ \mathcal{D}_{\text{total}}(Q) = - \tilde{\sigma}_z \bar{\rho}_E \; ,$$
which fortunately reduces to a scalar equation on the coefficient of $\sigma_z$. Plugging the solution into the formula for the dissipation matrix gives:
$$X= \tfrac{1-\bar{z}^2}{(1+2 n_{\text{th}}) \kappa_1 + 2(\kappa_x+\kappa_y)} = \tfrac{4 c_+ c_-}{(c_+ + c_-)^3}$$
where $c_- = (1+n_{\text{th}}) \kappa_1 + \kappa_x+\kappa_y$ and $c_+ = n_{\text{th}} \kappa_1 + \kappa_x+\kappa_y$. This is the induced dissipation rate acting on T with the operator $T_z$. We notice that $\kappa_z$ plays no role here and we can make the following observations.
\begin{itemize}
\item One checks that, for any values of the other parameters, this induced dissipation rate decreases when $\kappa_1$ increases. Thus we should fix $\kappa_1$ at its maximal bound. For $\kappa_1$ dominating, the induced dissipation decreases as $1/\kappa_1$.
\item Once the value of $\kappa_1$ is fixed, we can write $\kappa_x+\kappa_y = \kappa_1 n_b$ such that the induced dissipation becomes a function of $n_{\text{eff}} = n_{\text{th}} + n_b$ only, namely
\begin{equation}\label{eq:vltm3}
X = \tfrac{4 n_{\text{eff}} (n_{\text{eff}+1})}{(2 n_{\text{eff}}+1)^3} \;.
\end{equation}
This function increases from $X=0$ at $n_{\text{eff}}=0$ towards a maximum at $n_{\text{eff}}= \tfrac{\sqrt{3}-1}{2} \approx 0.366$, then slowly decreases to $0$ as $n_{\text{eff}}$ tends to infinity. Note that the adiabatic elimination approximation remains well valid near $n_{th}=0$, as long as $\kappa_1$ itself is significantly larger than the coupling Hamiltonian.

Thus, the minimal induced dissipation will be obtained either at the lower or at the upper bound of $n_{\text{eff}}$, depending on their values. In other words, if a very low temperature can be achieved to keep the TLS close to its ground state then this is favorable, but otherwise it is better to just make it as mixing as possible. The judge about ``very low temperature'' is the formula \eqref{eq:vltm3}.

\item Comparing to Propositions: There is nothing special to say regarding Proposition~\ref{prop:D_1}.

Proposition~\ref{prop:D_2} applies rigorously when $\kappa_1=0$; taking $\kappa_1$ very low, we would indeed be in a regime where $n_{\text{eff}} \approx n_b = \tfrac{\kappa_x+\kappa_y}{\kappa_1}$ is large, and we have seen that as soon as $n_b>\tfrac{\sqrt{3}-1}{2}$ it is beneficial to increase it, be it through $\kappa_x$ or $\kappa_y$. On the contrary, if $\kappa_1$ is the dominating dissipation, then increasing $\kappa_x$ or $\kappa_y$ is not necessarily beneficial, as we may be in the regime $n_{\text{eff}} < \tfrac{\sqrt{3}-1}{2}$. This supports the condition that all dissipation operators must be Hermitian for Proposition~\ref{prop:D_2} to apply.

Regarding Proposition~\ref{prop:D_3}, as soon as $\kappa_1$ or two other dissipation channels can be increased indefinitely, we are in the situation where $\mathcal{D}_b$ has a unique steady state, and the induced dissipation goes to zero. There remains the case where only a single Hermitian channel can be increased indefinitely. 
\begin{itemize}
    \item Taking this channel to be $\kappa_z$, the elimination of $\mathcal{D}_b$ yields a reduced state space of the type $p_g \rho_g \otimes \ket{g}\bra{g} + (1-p_g) \rho_e \otimes \ket{e}\bra{e}$ with free parameters $p_g,\rho_g,\rho_e$. The remaining fast dynamics will stabilize the value $p_g = \bar{p}_g$ independently of the coupled target system T. The case $\kappa_x=\kappa_y=n_{th}=0$, thus with $\kappa_1$ stabilizing $\ket{g}\bra{g}$ as $\bar{p}_g=1$, would yield a rank-deficient $\bar{\rho}_E$ for which induced dissipation completely vanishes, even for finite $\kappa_z$; otherwise, induced dissipation will always be finite.
    \item Taking the possibly infinitely strong channel to be $\kappa_x$, the elimination of $\mathcal{D}_b$ yields a reduced space of similar form but with $\ket{+},\ket{-}$ replacing $\ket{g},\ket{e}$. At first-order adiabatic elimination of $\mathcal{D}_b$, the dispersive coupling Hamiltonian cancels and there only remains dissipation pushing $p_+$ towards $1/2$. Hence, driving $\kappa_x$ (or $\kappa_y$) towards infinity is sufficient to drive the induced dissipation on T towards 0. Since this stabilizes the most mixed environment state, this might not have been the most intuitive guess.
\end{itemize}
\end{itemize}

\subsubsection{Almost-resonant coupling (Jaynes-Cummings)}\label{sssec:ex2:res}

As a second case, we consider \eqref{eq:PsecTLS:totdiss} with adjustable $\kappa_k$, in conjunction with the fixed Hamiltonian:
\begin{eqnarray}\label{eq:PsecTLS:rescoup}
\opus{H}_{TE} &=& \tfrac{\Delta}{2} \mathbb{1} \otimes \opus{\sigma_z} + 2g ( \opus{T}_+ \otimes \opus{\sigma_-} + \opus{T}_- \otimes \opus{\sigma_+})  \\ \nonumber
&=&  \tfrac{\Delta}{2} \mathbb{1} \otimes \opus{\sigma_z} + g ( \opus{T_x} \otimes \opus{\sigma_x} + \opus{T_y} \otimes \opus{\sigma_y}) \; ,  
\end{eqnarray}
where $\opus{T_x} = \opus{T}_-+\opus{T}_+$ and $\opus{T_y} = -i(\opus{T}_- -\opus{T}_+)$.

In the adiabatic elimination formulas, the fast TLS dynamics now includes both $\opus{H_E} = \frac{\Delta}{2} \sigma_z$ and $\mathcal{D}_{\text{total}}$. Note that this remains valid when $\Delta$ is not dominating $g$, because we only need fast \emph{dissipation}. However, when $\Delta$ does take a large value, it enables to have a strongly attractive unique TLS steady state even if just $\kappa_x$ or $\kappa_y$ is nonzero. The steady state of the TLS alone is:
$$\bar{\rho}_E = \frac{\mathbb{1}+\bar{z} \opus{\sigma_z}}{2} \;\; \text{ with } \;\; \bar{z} = \tfrac{-\kappa_1}{(1+2 n_{\text{th}}) \kappa_1 + 2(\kappa_x+\kappa_y)} \;.$$
For the coupling operators $\sigma_x$ and $\sigma_y$, we then compute
$$\tilde{\sigma}_x \bar{\rho}_E = \tfrac{\opus{\sigma_x}-i \bar{z} \opus{\sigma_y}}{2}\;\; , \;\; \tilde{\sigma}_y \bar{\rho}_E = \tfrac{\opus{\sigma_y}+ i \bar{z} \opus{\sigma_x}}{2}\;.$$
The solution of
\begin{equation}\label{eq:PsecTLS:HEdyn}
-i[H_E, Q_{k} ] + \mathcal{D}_{\text{total}}(Q_k) = - \tilde{\sigma}_k \bar{\rho}_E
\end{equation}
for $k \in \{x,y\}$ is rather easy in Bloch coordinates, as the dynamics decouple the coefficients of $\sigma_x,\sigma_y$ from those of $\sigma_z,\mathbb{1}$. We can thus write $Q_k = q_{k,x} \sigma_x + q_{k,y} \sigma_y$ and the left side of \eqref{eq:PsecTLS:HEdyn} just becomes
\begin{eqnarray*}
&& \tfrac{\Delta}{2}  (q_{k,x} \sigma_y - q_{k,y} \sigma_x) \\
&& - 2 (\kappa_y+\kappa_z+ (1+2 n_{\text{th}}) \tfrac{\kappa_1}{4}) q_{k,y} \sigma_y \\
&& - 2 (\kappa_x+\kappa_z+ (1+2 n_{\text{th}}) \tfrac{\kappa_1}{4}) q_{k,x} \sigma_x \;.
\end{eqnarray*}
Equating the components in $\sigma_x$ and $\sigma_y$ gives the solutions, from which we construct the dissipation matrix:
$$ X= \frac{1}{\tfrac{\Delta^2}{4}+c_x c_y}\; \left( \begin{array}{cc} 
c_y & i \bar{z} \tfrac{c_y+c_x}{2} \\  -i \bar{z} \tfrac{c_y+c_x}{2} & c_x
\end{array}\right) \; ,
$$
with $c_{x,y} = \kappa_{x,y} + \kappa_z + \tfrac{(1+2 n_{\text{th}}) \kappa_1}{4}$.
The parameters now define not only the induced dissipation rate, but also the associated operators (unitary combinations of $T_x,T_y$). Considering any of them as equally bad for the target system, we typically look at the spectrum of $X$. We can make the following observations.
\begin{itemize}
\item $\Delta$: increasing the detuning between E and T always decreases induced dissipation, down to zero as $\Delta$ gets infinite.
\item Getting induced dissipation to zero at a finite value of $\Delta$, requires to increase \emph{both} $c_x$ and $c_y$ to infinity --- this will be impossible if only $\kappa_x$ or $\kappa_y$ can be made arbitrarily large.

\item $c_x,c_y$, sum of rates: The sum of induced dissipation rates (trace of $X$) as a function of $c_x,c_y$ looks like a saddle around the point $c_x=c_y=\Delta/2$, where induced dissipation is maximal as a function of $c_x+c_y$ and minimal as a function of $|c_x-c_y|$. Which side gives the minimum induced dissipation, will thus depend on the available range of $\kappa_k$.

In particular, for $\Delta=0$, induced dissipation will always decrease when we increase $\kappa_z, \kappa_1, n_{\text{th}}, \kappa_x, \kappa_y$. Thus even if we have the option $n_{\text{th}}=0$ to attract the TLS towards a pure state with only $\sigma_-$, it is better to not do so and rather increase the TLS temperature and other rates. This difference with respect to dispersive coupling of course stems from the fact that a ground state for E, although pure, has no particular advantage under resonant coupling. 

In particular, for large $\Delta$, it appears better to keep low dissipation on the TLS. This can be understood as keeping the TLS frequency well-defined, avoiding any leakage towards the frequencies to which the target system is sensitive. Be careful though that the formulas are only valid if the dissipation on the TLS remains significantly larger than its coupling with the target system. Otherwise, the correct viewpoint would rather be to first take the dispersive coupling limit and then analyze the system as in the previous Section~\ref{sssec:ex2:disp}.

\item $c_x,c_y$, individual rates: The difference between the two induced dissipation rates may be interesting to track when thinking e.g.~of the interest of having biased noise~\cite{Mirrahimi2014a}. At fixed value of the sum, the difference increases when $|c_x-c_y|$ gets larger (thus $\kappa_x$ up and $\kappa_y$ down), or when $\bar{z}^2$ gets larger (thus e.g.~$\kappa_1$ up and $n_{\text{th}}$ down).

In particular, for $\bar{z}=0$, increasing only e.g.~$\kappa_x$ and thus $c_x$, decreases one induced dissipation rate as $\frac{1}{c_x+\Delta^2/(4c_y)}$ (thus to $0$ as $\kappa_x$ gets infinite), but increases the other one as $\frac{1}{c_y+\Delta^2/(4c_x)}$ (or thus at best keeps it constant if $\Delta=0$, with finite limit $1/c_y$ as $\kappa_x$ gets infinite).

\item Comparing to Propositions: The two regimes of Proposition~\ref{prop:D_1} are well visible here. The one where $\Delta$ and all $\kappa_k$ are scaled by $\alpha$ is trivial. The case where $\Delta$ is fixed shows two things: if $\Delta$ is small, then Proposition~\ref{prop:D_1} says that it is better to increase the $\kappa_k$, as we see from the explicit formula here; however, if $\Delta$ is fixed and large, then Proposition~\ref{prop:D_1} does not apply and we see indeed with the present formula that the situation is not as clear. In other words, the saddle at  $c_x=c_y=\Delta/2$ is very consistent with the first case of Proposition~\ref{prop:D_1}.

Proposition~\ref{prop:D_2} applies at least when $\Delta=0$ (and $\kappa_1 = 0$). It predicts that in this setting, increasing any of $\kappa_x$  or $\kappa_y$ can only be beneficial. In the particular case $\bar{z}=0$ mentioned in the previous item, we see how a nonzero $\Delta$ moderates this conclusion.

Regarding Proposition~\ref{prop:D_3}, like for the dispersive coupling, the only nontrivial situation is when only a single Hermitian channel can be increased indefinitely. 
\begin{itemize}
\item Taking this channel to be $\kappa_z$, the elimination of $\mathcal{D}_b$ cancels the Hamiltonian coupling; thus, although convergence on E happens at a finite rate as we need $\mathcal{D}_a$ to finally converge to $\bar{\rho}_E$, the induced dissipation on T goes to $0$ as $\kappa_z$ gets infinite.
\item Taking this channel to be $\kappa_x$, the elimination of $\mathcal{D}_b$ yields a reduced state space of the type 
$$\hspace{1cm} p_+ \rho_+ \otimes \ket{+}\bra{+} \;+\; (1-p_+) \rho_- \otimes \ket{-}\bra{-},$$ 
with free parameters $p_+,\rho_+,\rho_-$. The remaining fast dynamics stabilizes $p_+ = 1/2$ independently of T, while the Hamiltonian coupling reduces to 
$$\hspace{1.5cm} -i g [T_x,\rho_+]\otimes \ket{+}\bra{+} + i g [T_x,\rho_-] \otimes \ket{-}\bra{-} \;.$$
Since $\bar{\rho}_E$ has full rank, the associated $T_x$ dissipation induced according to second-order adiabatic elimination is bound to stay finite, even when $\kappa_x$ tends to infinity.
\end{itemize}
\end{itemize}

\subsubsection{Partly dissipative environment}\label{sssec:ex3:nonuq}

We now address a setting where the fast decoherence of the TLS does not converge to a unique steady state $\bar{\rho}_E$. A typical example would be \eqref{eq:PsecTLS:totdiss} where only $\kappa_z$ is large. If this were the only dynamics on the environment qubit E, then implications for the target system T would depend on the environment's initial state. The intermediate case which we discuss hare, assumes that we also have the unavoidable $\kappa_-,\kappa_+$ dissipation, but with rates comparable to the coupling $g$ between E and T.

Since adiabatic elimination fundamentally works by splitting the fast and slow dynamics, it should thus eliminate only the fast decay of E  under $\kappa_z \mathcal{D}_{\sigma_z}$,
i.e.~the quickly vanishing coherences among $\ket{e}$ and $\ket{g}$ states of E.
The $\kappa_-,\kappa_+$ dissipation on E has to be taken with the slow dynamics, which thus cover both the target system and the populations on $\q{e}\qd{e}$ or $\q{g}\qd{g}$ of the environment E. To illustrate what this can imply for the target system, we again investigate the two typical coupling cases.\\

\paragraph*{Dispersive coupling:} First consider the case of a dispersive coupling:
\begin{equation}\label{eq:3rdexdc}
\tfrac{d}{dt}\rho = \kappa_z \mathcal{D}_{\sigma_z} + \kappa_- \mathcal{D}_{\sigma_-} + \kappa_+ \mathcal{D}_{\sigma_+}  - i g [T_z \otimes \sigma_z, \rho] \; ,
\end{equation}
where we recall that only $\kappa_z$ is supposed to be larger than the other rates.

The set of states of the form $\rho = \rho_g \otimes \q{g}\qd{g} + \rho_e \otimes \q{e}\qd{e}$, corresponding to the set where $\mathcal{D}_{\sigma_z}(\rho)=0$, is in fact exactly invariant under \eqref{eq:3rdexdc}.  The dynamics for the slow variables $\rho_g$ and $\rho_e$ (each positive Hermitian, but only sum of their traces must equal one) write as:
\begin{eqnarray*}
\tfrac{d}{dt} \rho_g  &=& \kappa_- \rho_e - \kappa_+ \rho_g + i g [T_z, \rho_g] \\
\tfrac{d}{dt} \rho_e  &=& \kappa_+ \rho_g - \kappa_- \rho_e - i g  [T_z, \rho_e] \;.
\end{eqnarray*}
Consider an initial separable state between T and E, thus $\rho = \rho_T \otimes (w \q{g}\qd{g} + (1-w) \q{e}\qd{e})$, where the environment populations are at steady-state value $w = \tfrac{\kappa_-}{\kappa_+ + \kappa_-}$. In the eigenbasis of $T_z$, the diagonal elements of $\rho_T$ do not change. However, as the environment jumps between $\q{e}$ and $\q{g}$ implying opposite rotations with $T_z$, the off-diagonal elements of $\rho_T$ will undergo induced decay. More precisely, for each pair of eigenvalues $\lambda_j,\lambda_k$ of $T_z$, the corresponding off-diagonal elements of $\rho_g$ and $\rho_e$ will decay according to the eigenvalues
\begin{eqnarray*}
r_{\pm} &=& -\kappa_1 (n_{\text{th}}+\tfrac{1}{2}) \pm \sqrt{ \kappa_1^2 (n_{\text{th}}+\tfrac{1}{2})^2 - L^2 + i \kappa_1 L} \quad \\
 && \text{with }\;  L = g  (\lambda_j-\lambda_k) \;.
\end{eqnarray*}
\begin{itemize}
\item When $L$ is small compared to $\kappa_1 (n_{\text{th}}+\tfrac{1}{2}) = (\kappa_- + \kappa_+)/2$, we would be in the regime where adiabatic elimination of E still holds. The slowest eigenvalue approximates as
\begin{eqnarray*}
r_- & \simeq & \;  i \tfrac{L}{(2n_{\text{th}}+1)} \\
& & - \tfrac{L^2}{\kappa_1 (2n_{\text{th}}+1)} (\;1-\tfrac{1}{(2 n_{\text{th}}+1)^2} \;) \;.
\end{eqnarray*}
In the second line we thus do find back the induced dissipation rate in $g^2/\kappa$, with an additional factor accounting for the fact that induced dissipation vanishes if the environment is exclusively in $\q{g}$. An optimization like in the previous examples applies, and larger dissipation on E implies lower induced dissipation on T.
\item When $L$ is large compared to $\kappa_1 (n_{\text{th}}+\tfrac{1}{2})$, the eigenvalues boil down to
$$r_{-} \simeq \kappa_- + i q \text{ and } r_{+} \simeq  \kappa_+  -i q$$
for some real parameter $q$. Thus the induced dissipation rates on T are equal to the ones of excitation and loss on E, irrespective of the value of $L$. Contrary to the previous case, it is thus better to keep environment dissipation low.
\end{itemize}
These two cases in fact illustrate the transition from the situation where highest environment dissipation is better (``surprising'' conclusion of adiabatic elimination) to the case where lowest dissipation is better (truly i.e.~not only according to the standard formula for adiabatic elimination of E, whose validity drops). According to both these limit cases, an intermediate rate of dissipation appears worst. Note that we are comparing the environment dissipation to $L = g  (\lambda_j-\lambda_k) $, thus in a single multi-level system the different cases can arise for different off-diagonal elements.\\

\paragraph*{Resonant coupling:} Consider the model
\begin{eqnarray}\label{eq:3rdexrc}
\tfrac{d}{dt}\rho &=& \kappa_z \mathcal{D}_{\sigma_z} + \kappa_- \mathcal{D}_{\sigma_-} + \kappa_+ \mathcal{D}_{\sigma_+}  \\
\nonumber && - i g [T_x \otimes \sigma_x + T_y \otimes \sigma_y,\; \rho] \; ,
\end{eqnarray}
where again only $\kappa_z$ is supposed to be larger than the other rates.

The Hamiltonian coupling makes it difficult to exactly identify the slow invariant subspace from intuition,
so we apply the mathematical adiabatic elimination procedure as recalled in Appendix~\ref{sec:computations_partly_dissipative}.
At order zero, the slow subspace is parametrized as $\rho_s = \rho_g \otimes \q{g}\qd{g} + \rho_e \otimes \q{e}\qd{e}$, with slow variables $\rho_g$ and $\rho_e$ (each positive Hermitian, but only sum of their traces must equal one).

At order one in $\varepsilon = (g,\kappa_{\pm}) / \kappa_z$, the slow dynamics  $\kappa_z \varepsilon \mathcal{L}_{s,1}$ correspond to:
$$\tfrac{d}{dt}\rho_g = \kappa_- \rho_e - \kappa_+ \rho_g \;\; , \;\; \tfrac{d}{dt}\rho_e = \kappa_+ \rho_g - \kappa_- \rho_e \;.$$
The coupling Hamiltonian thus vanishes and the state of T remains untouched in the sense that we have, at this order, $\tfrac{d}{dt} (\rho_g+\rho_e) = 0$.

At order two though, we get the dissipative dynamics:
\begin{eqnarray*}
\tfrac{d}{dt} \rho_s &=& \kappa_z \varepsilon \mathcal{L}_{s,1}(\rho_s) \\
&& + \tfrac{g^2}{\kappa_z}\mathcal{D}_{(T_x+iT_y) \otimes \q{e}\qd{g}}(\rho_s)\\
&& +  \tfrac{g^2}{\kappa_z}\mathcal{D}_{(T_x-iT_y) \otimes \q{g}\qd{e}}(\rho_s) \;.
\end{eqnarray*}
The second-order dissipation combines $(\q{e},\q{g})$-population mixing on E with induced decoherence on T. To get an idea of the latter, we can consider the (quite academic) special case where $T_x=T_y$ and again obtain autonomous dynamics for T, namely:
$$\tfrac{d}{dt} (\rho_g+\rho_e) = \tfrac{2 g^2}{\kappa_z} \mathcal{D}_{T_x}(\rho_g+\rho_e)\;.$$
Thus, unlike for dispersive coupling, the induced dissipation (up to second order included) appears to be independent of the values of $\kappa_-$, $\kappa_+$ as long as they remain small compared to $\kappa_z$.

\enlargethispage{-4.0\baselineskip}

\section{Conclusion}

Protecting a target quantum system from decoherence is a major objective towards quantum technology. Although quantum information loss on a target physical system is often expressed via Markovian decoherence channels, everyone acknowledges that this only approximates more intricate dynamics of a larger system. Adding dynamics at the fast timescales of this larger system may thus allow to change the induced decoherence on target, and ideally reduce it. This is essentially the idea behind $1/f$ noise mitigation methods, revised Floquet-Markov Lindbladians, and spin echo or quantum dynamical decoupling (QDD) techniques, among others.  

In the present paper, we express the not entirely Markovian dissipation on the target system T as a Hamiltonian coupling to a low-dimensional environment subsystem E, which itself undergoes Markovian / Lindbladian dissipation. This is in line with initial QDD settings \cite{DDofOpenQsystems}, which focus on the Hamiltonian part of T and E.

The specificity of our proposal is to mitigate decoherence of T by acting on the intermediate environment E instead of on the target system T. Such actions cannot be assumed as precise as on T, but they need not be. Indeed, we explicitly quantify how not only strong and imprecise coherent drives, but also adding pure decoherence channels on E (without introducing direct Markovian dissipation on T itself), effectively reduces the decoherence induced on T. Maybe surprisingly, we observe how only particular circumstances would favor a very pure environment as compared to a very mixing one.

The reduction of induced decoherence on T when increasing the decoherence on E should not be too unfamiliar to researchers used to adiabatic elimination and the ``$g^2/\kappa$'' formula. In light of the present paper, this is interpreted as a QDD effect, which can arise both through coherent or incoherent driving, and which can be quantified precisely in both cases. 

Indeed, having all fast dynamics on subsystem E, we can go beyond Hamiltonian decoupling arguments and develop an
adiabatic elimination procedure yielding explicit formulas for the decoherence of T induced by dissipation on E. The
resulting formulas are valid in the limit of strong dissipation on E, which is precisely the regime that is typically
favored. They allow to explicitly examine trade-offs and dependencies on parameters, as we illustrate on various typical settings when E is a two-level system.


\begin{acknowledgments}
The authors thank Pierre Rouchon, Mazyar Mirrahimi, Zaki Leghtas and Francesco Ticozzi for inspiring discussions. This work has been supported by the ANR grant JCJC-HAMROQS (French National Research Agency).
\end{acknowledgments}


\appendix

\section{Dissipation model of TLS under ultra-strong driving}\label{app:lindbladian_derivation}

In this section, we rederive a Lindbladian model for the dissipation of the environment TLS,
starting from a general model where this E subsystem is coupled to a large bath B through a Hamiltonian coupling.
The Lindbladian obtained will explicitly account for the possibly strong drive $H_c(t)$ on E, yielding a perturbative correction to the dissipators in the $(\sigma_-,\sigma_+$)-basis obtained in the undriven case. We will see that this correction becomes significant when the drive amplitude is non-negligible w.r.t.\ the bare frequency of the TLS. For the derivation we follow the standard approach of the Born-Markov approximation~\cite{Breuer2007}, followed by a secular approximation, averaging out over rapidly oscillating terms.

We consider a general bath B with bare Hamiltonian $H_B$, and denote the bare frequency of E by $\Omega_E =
\bestguess{\Omega}{E} + \delta\Omega_E$, with $\bestguess{\Omega}{E}$ the nominal user-known value.
In the lab frame of both systems, we assume an inter-system coupling $\gamma \sigma_x \otimes R$,
with $R$ some constant Hermitian operator acting on the bath Hilbert space, and $\gamma$ some small positive coupling rate with the dimension of a frequency.
For example, if the bath can be modeled as a thermal reservoir of harmonic oscillators, this leads to a Jaynes-Cummings-type interaction with each of the different modes in the bath.
Moving to the rotating frame of both systems, and introducing the drive as in~\eqref{eq:EAD_control}, we obtain the total Hamiltonian
\begin{equation}\label{eq:bath_model_zonder_T}
    H_E(t) \otimes \mathbb{1}_B +
     \gamma (\sigma_+ e^{i \bestguess{\Omega}{E} t} + \sigma_- e^{-i \bestguess{\Omega}{E} t}) \otimes \tilde{R}(t),
\end{equation}
with $\tilde{R}(t) = e^{i H_B t} R e^{- i H_B t}$. Performing the toggling frame transformation defined in~\eqref{eq:control_propagator} yields
\begin{equation}\label{eq:int_frame_bath}
    \gamma \tilde{E}(t) \otimes \tilde{R}(t),
\end{equation}
with $ \tilde{E}(t) = e^{i (\bestguess{\Omega}{E} + \bestguess{\omega}{1}) t} E_+(t) +e^{- i (\bestguess{\Omega}{E} + \bestguess{\omega}{1}) t} E_-(t)$, 
where we have defined
\begin{subequations}
    \label{eq:e_pm}
    \begin{align}
        E_+(t) &:= e^{i \frac{\Lambda}{2} \sigma_{\alpha x}} \sigma_+ e^{- i \frac{\Lambda}{2} \sigma_{\alpha x}}\nonumber\\
        &= \frac{\cos(\alpha)}{2} \sigma_{\alpha x} + i \frac{1 + \sin(\alpha)}{2} e^{i \Lambda t} \sigma_{\alpha+}  \nonumber\\
               &+ i \frac{1 - \sin(\alpha)}{2} e^{ - i \Lambda t} \sigma_{\alpha-}\label{eq:e_p}\\
        E_-(t) &:= E_+^\dag(t).
    \end{align}
\end{subequations}
At this point we introduce the modified bare E-frequency $\tilde{\Omega}_E = \bestguess{\Omega}{E} + \bestguess{\omega}{1}$.
In line with the conclusions in section~\ref{sec:coherent_decoupling}, we will consider $\bestguess{\omega}{1} \ll \Omega_E$, so $\tilde{\Omega}_E \simeq \Omega_E$.

In the interaction frame of~\eqref{eq:int_frame_bath}, the evolution equation of the joint density matrix $\rho_{EB}$ is thus
\begin{equation}\label{eq:startpoint_diff}
    \dot{\rho}_{EB}(t) = - i \gamma \comm{\tilde{E}(t) \otimes \tilde{R}(t)}{\rho_{EB}(t)}.
\end{equation}
We can write this as an integral equation,
\begin{equation*}
    \rho_{EB}(t) = \rho_{EB}(0) - i \gamma \int_0^t \comm{\tilde{E}(s) \otimes \tilde{R}(s)}{\rho_{EB}(s)} \dd s,
\end{equation*}
and reinjecting this into~\eqref{eq:startpoint_diff}, we obtain
\begin{align}
    \dot{\rho}_{EB}(t) &= - i \gamma \comm{\tilde{E}(t) \otimes \tilde{R}(t)}{\rho_{EB}(0)}\label{eq:startpoint_integrodiff}\\
                       &- \gamma^2 \int_0^t\comm{\tilde{E}(t) \otimes \tilde{R}(t)}{\comm{\tilde{E}(s) \otimes \tilde{R}(s)}{\rho_{EB}(s)}} \dd s.\nonumber
\end{align}

Up until here, no approximations have been made, so~\eqref{eq:startpoint_integrodiff} is exact.
At this point we follow the standard procedure of the Born-Markov approximation~\cite{cohenT_book2F,Breuer2007,Qnoise},
assuming the bath to be very large and unaffected by the weak coupling with the E system, so it remains in a steady state $\bar{\rho}_B$ that is invariant under $H_B$ ($\comm{H_B}{\bar{\rho}_B} = 0$).
Without loss of generality we can take $\Tr(R \; \bar{\rho}_B) = 0$, since otherwise this would just lead to a modification of the bare E-Hamiltonian.
Lastly, we assume the correlation time of the bath to be the shortest timescale present in the joint system.
Taking the partial trace of both sides in~\eqref{eq:startpoint_integrodiff} and performing these approximations
yields a Markovian equation for E:
\begin{align*}
    &\frac{\dot{\rho}_E(t)}{\gamma^2} =\\
    &\int_0^\infty \!\!\! \Tr(\comm{\comm{\tilde{E}(t \shortminus s) \otimes \tilde{R}(t \shortminus s)}{\rho_E(t) \otimes \bar{\rho}_B}}{\tilde{E}(t) \otimes \tilde{R}(t)}) \dd s
\end{align*}
The right-hand side can be further worked out by defining the two-point correlation function $g(z)$ of the bath as
\[g(z) := \Tr(\tilde{R}(t) \tilde{R}(t-z) \, \bar{\rho}_B), z, t \in \mathbb{R},\]
yielding
\begin{align}
\frac{\dot{\rho}_E(t)}{\gamma^2}  =  \int_{-\infty}^\infty g(z) \Big(&\comm{\tilde{E}(t-z) \rho_E(t)}{\tilde{E}(t)}\nonumber\\
    + &\comm{\tilde{E}(t)}{\rho_E(t) \tilde{E}(t+z)}\Big) \dd z.\label{eq:born_markov}
\end{align}
Plugging in the expression of $\tilde{E}(t)$ as in~\eqref{eq:e_pm}, terms oscillating at frequencies
$\pm 2 \tilde{\Omega}_E, \pm 2 \Lambda, 2 \tilde{\Omega}_E \pm 2 \Lambda, - 2 \tilde{\Omega}_E \pm 2 \Lambda$ appear.
Regarding oscillations as a function of $z$, we define the spectral density of the bath $G$ as
\begin{equation}\label{eq:defn_spectral_density}
    G(\nu) := \int_{-\infty}^\infty e^{i \nu z} g(z) \, \dd z, \forall \nu \in \mathbb{R} \; .
\end{equation}
There remains to treat the terms oscillating as a function of $t$. The bare TLS frequency $\tilde{\Omega}_E$ can always be assumed very large w.r.t.\ $\dot{\rho}_E$ in~\eqref{eq:born_markov},
justifying to average over terms oscillating at frequencies $\pm \tilde{\Omega}_E$. The case of ultra-strong driving
precisely assumes that we can similarly average over frequencies $\pm \Lambda$ and, avoiding parametric resonance, over the frequencies $2 \tilde{\Omega}_E \pm 2 \Lambda$ and $- 2 \tilde{\Omega}_E \pm 2 \Lambda$. Therefore, the final Lindbladian model is just obtained by performing a last secular approximation (i.e.\ RWA) as mentioned in the first paragraph, yielding:
\begin{equation}\label{eq:dissipator_bis}
    \dot{\rho}_E(t) = \kappa_{\alpha x}  \mathcal{D}_{\sigma_{\alpha x}} + \kappa_{\alpha-} \mathcal{D}_{\sigma_{\alpha-}} + \kappa_{\alpha+} \mathcal{D}_{\sigma_{\alpha+}},
\end{equation}
where
\begin{subequations}\label{eq:zeghiou}
    \begin{align}
        \kappa_{\alpha x} &:= \frac{\gamma^2}{2} \qty(G(\tilde{\Omega}_E) + G(\shortminus \tilde{\Omega}_E)) \cos^2(\alpha), \\
        \kappa_{\alpha-} &:= \frac{\gamma^2}{2} G(\tilde{\Omega}_E + \Lambda) {(1 + \sin(\alpha))}^2\nonumber\\
                        &\hphantom{:}+\frac{\gamma^2}{2} G(\shortminus \tilde{\Omega}_E + \Lambda) {(1 \shortminus \sin(\alpha))}^2, \\
        \kappa_{\alpha+} &:= \frac{\gamma^2}{2} G(\shortminus \tilde{\Omega}_E \shortminus \Lambda) {(1 + \sin(\alpha))}^2 \nonumber \\
                        &\hphantom{:}+\frac{\gamma^2}{2} G(\tilde{\Omega}_E \shortminus \Lambda) {(1 \shortminus \sin(\alpha))}^2
    \end{align} \;
\end{subequations}
while $\alpha$ and the associated operators are defined in the main text. The model~\eqref{eq:dissipator_bis} is used in the analysis of section~\ref{sec:dissipation_model}. In the final secular approximation, this model neglects 2nd-order RWA-terms of order $\frac{\kappa_{\alpha x,\alpha+,\alpha-}^2}{\tilde{\Omega}_E}$ and of order $\frac{\kappa_{\alpha x,\alpha+,\alpha-}^2}{\Lambda}$.

\subsubsection{Interpretation}

We can briefly comment on how to consider the dissipation rates \eqref{eq:zeghiou} as a function of our QDD parameters.
\begin{itemize}
\item The effect of $\frac{\bestguess{\omega}{1}}{2} \sigma_z$ just adds up to $\bestguess{\Omega}{E}$,
so for $\omega_2=0$ the bath noise spectrum $G$ is probed at altered frequencies $\pm(\bestguess{\Omega}{E} + \bestguess{\omega}{1})$ to evaluate the excitation and loss rates.
Knowing $\Omega_E$ up to $\delta\Omega_E$ anyways, if we want to use these equations we have to assume $\bestguess{\omega}{1}, \delta\Omega_E \ll \Omega_E$, and $G$ sufficiently flat for the induced frequency shift to have negligible effect on the induced Lindbladian.
\item The stronger drive of amplitude $\omega_2$ introduces the periodic time-dependence in the TLS Hamiltonian \eqref{eq:EAD_control}. According to the general Floquet-Markov theory~\cite{GRIFONI1998229}, the eigenbasis in which the TLS decoheres is then given by the Floquet Hamiltonian associated to E, in a frame given by a periodic change of variables (often called the micromotion), and Lindbladian dissipation rates are found by evaluating the bath noise spectrum at the Floquet quasi-energies. In our case, the periodic change of variables just corresponds to going to the rotating frame w.r.t.\ $\frac{\bestguess{\omega}{1}}{2} \sigma_z$, where we obtain a constant Hamiltonian $\frac{\Lambda}{2} \sigma_{\alpha x}$ on E. This special situation implies that the Floquet decomposition trivializes to the more standard rotating-frame and averaging approach, but thus still with correspondingly modified dissipation channels on E.
\end{itemize}

\subsubsection{Deriving the other Lindblad model}\label{App.errordisc}

We can briefly review the derivation using the first possible condition mentioned in Section
\ref{sec:dissipation_model}, namely for obtaining a Lindblad model whose dissipators do not depend on the drive when
$\Lambda \ll \Omega_E$ (strong, yet not ultra-strong driving).

The steps up to \eqref{eq:defn_spectral_density} remain the same. From there, with $\Lambda \ll \Omega_E$, we can still perform a final secular approximation over frequencies $\pm  \tilde{\Omega}_E$ as well as $2 \tilde{\Omega}_E \pm 2 \Lambda$ and $- 2 \tilde{\Omega}_E \pm 2 \Lambda$. However, averaging over $\pm \Lambda$ may not be justified and another standard type of approximation is applied to obtain a stationary Lindbladian. This consists in assuming that the bath spectral density $G$ is sufficiently flat to be considered invariant w.r.t.\ frequency shifts of $\pm \Lambda$ around $\Omega_E \gg \Lambda$. 

Thus concretely, averaging~\eqref{eq:born_markov} over $t$ with only the frequencies $\pm 2 \tilde{\Omega}_E$, yields
    \begin{align*}
        \int_{-\infty}^\infty\!\!\!\!\! g(z) e^{i \tilde{\Omega}_E z} \big(&\comm{E_-(t\!\shortminus\!z) \rho_E}{E_+(t)} + \comm{E_-(t) \rho_E}{E_+(t\!+\!z)}\big)\\
        +\!\!\int_{-\infty}^\infty \!\!\!\!\! g(z) e^{- i \tilde{\Omega}_E z} \big(&\comm{E_+(t\!\shortminus\!z) \rho_E}{E_-(t)} + \comm{E_+(t)}{\rho_E E_-(t\!+\!z)}\big).
    \end{align*}
We next shift the $z$ dependency of $E_-$ and $E_+$ towards $g(z)$ and assume $G(\tilde{\Omega}_E \pm \Lambda) \simeq G(\tilde{\Omega}_E) \simeq G(\Omega_E)$ when integrating over $z$. Finally, moving back to the lab frame by undoing~\eqref{eq:control_propagator}, we then readily obtain~\eqref{eq:pm_dissipation} where the drive has no impact on the Lindbladian dissipation.

As explained in the main text, the conclusions obtained with these two approaches do coincide (at least at leading orders) when both conditions --- averaging over $\pm \Lambda$, and assuming a locally flat noise spectrum $G(.)$ --- are satisfied.

\section{Adiabatic elimination method}\label{app:ae_formulas}

\subsection{Summary of the formalism}\label{app:adelgenform}

Consider dynamics with the following timescale separation
\begin{equation}\label{eq:ae_start}
    \dot{\rho} = \mathcal{L}_0(\rho) +  \varepsilon \mathcal{L}_1(\rho).
\end{equation}
Here, $\rho$ is a density operator acting on a Hilbert space $\mathcal{H}$,
$\mathcal{L}_0$ a stationary Lindbladian of order 1, and $\mathcal{L}_1$ an order-one
Lindbladian providing a perturbation, since $\varepsilon \ll 1$ is a small positive constant.
We use the term Lindbladian in the broad sense, as we assume any Hamiltonian parts of the dynamics to be included in $\mathcal{L}_0$ or $\mathcal{L}_1$. The starting point is that the fast dynamics are degenerate, i.e. the linear superoperator  $\mathcal{L}_0$, acting on the set of linear operators on $\mathcal{H}$,  has a nontrivial kernel $\mathcal{M}_0$ associated to eigenvalue 0. Furthermore, this kernel is strongly attractive, in other words all the non-zero eigenvalues of $\mathcal{L}_0$ have a strictly negative real part (spectral gap).

The goal of adiabatic elimination, as described in~\cite{Azouit2017b}, is then to obtain a reduced model describing the perturbation of this degenerate kernel under the full Lindblad dynamics $\mathcal{L}_0 + \varepsilon \mathcal{L}_1$, for $\varepsilon$ small. This reduced model involves an invariant space $\mathcal{M}_r$ --- dubbed the slow or reduced subspace --- of the same dimension as the kernel $\mathcal{M}_0$ of $\mathcal{L}_0$, and on which the perturbation of 0-eigenvalues now imply some slow dynamics associated to eigenvalues of order $\varepsilon$ of the superoperator.

Both the invariant space $\mathcal{M}_r$ and the associated slow dynamics $\mathcal{L}_{s}$ can be determined as a power series in $\varepsilon$.
For this we parameterize the reduced model with a variable $\rho_s \in \mathcal{M}_s \simeq \mathcal{M}_r$ undergoing the dynamics
\begin{equation}\label{eq:anszats_AELs}
    \dot{\rho}_s(t) = \mathcal{L}_{s,\varepsilon}(\rho_s(t)) = \sum_{k = 1}^\infty \varepsilon^k \mathcal{L}_{s,k}(\rho_s) \; ;
\end{equation}
and we express how this variable is embedded in the full system, thus mapping the parameterization space $\mathcal{M}_s$
to the actual invariant eigenspace $\mathcal{M}_r$, via the linear map:
\begin{equation}\label{eq:anszats_AE}
    \rho(t) = \mathcal{K}_\varepsilon(\rho_s(t)) = \sum_{k = 0}^\infty \varepsilon^k \mathcal{K}_k(\rho_s(t)).
\end{equation}
Ideally, we want $\mathcal{L}_{s,\varepsilon}$ to have the typical Lindblad structure of positivity-preserving dynamics, and $\mathcal{K}_\varepsilon$ to be a Kraus map, so density matrices in $\mathcal{M}_s$ are mapped to density operators in the total space. General expressions satisfying this structure have been obtained when truncating the series after 2nd order; we hence keep following the procedure of~\cite{Azouit2017b}.

Demanding that the equations \eqref{eq:anszats_AELs},\eqref{eq:anszats_AE} be solution of~\eqref{eq:ae_start}, the $\mathcal{L}_{s,k}$ and $\mathcal{K}_k$ can be progressively identified by matching terms of equal order in $\varepsilon$.
Explicitly, one obtains
\begin{align*}
    \mathcal{L}_0 (\mathcal{K}_0(\rho_s)) &= 0,\\
    \mathcal{K}_0 ( \mathcal{L}_{s,1}(\rho_s)) &= \mathcal{L}_0 (\mathcal{K}_{1}(\rho_s)) + \mathcal{L}_1(\mathcal{K}_0(\rho_s)),\\
    \mathcal{K}_0 ( \mathcal{L}_{s,2}(\rho_s)) + \mathcal{K}_1 ( \mathcal{L}_{s,1}(\rho_s)) &= \mathcal{L}_0 (\mathcal{K}_{2}(\rho_s)) + \mathcal{L}_1(\mathcal{K}_1(\rho_s)),\\
    & \hphantom{a} \vdots
\end{align*}
Since these equations should hold for any $\rho_s \in \mathcal{M}_s$, we write (with a slight abuse of notation, since all operators are linear):
\begin{align*}
    \mathcal{L}_0 \mathcal{K}_0 &= 0,\\
    \mathcal{K}_0  \mathcal{L}_{s,1} &= \mathcal{L}_0 \mathcal{K}_{1} + \mathcal{L}_1\mathcal{K}_0,\\
    \mathcal{K}_0  \mathcal{L}_{s,2} + \mathcal{K}_1  \mathcal{L}_{s,1} &= \mathcal{L}_0 \mathcal{K}_{2} + \mathcal{L}_1\mathcal{K}_1,\\
    & \hphantom{a} \vdots
\end{align*}
where a product of superoperators stands for their composition.

The solution is not unique since we have a choice in the parameterization of $\mathcal{M}_r$ via $\mathcal{M}_s$, but it has been proved that solutions exist~\cite{Azouit2017b}.
At each order, we can first solve for $\mathcal{L}_{s,k}$ by projecting the corresponding equation with
\[\mathcal{R} := \lim_{t\rightarrow + \infty}\exp(\mathcal{L}_0\, t)\]
onto the subspace corresponding to the zero eigenvalues of $\mathcal{L}_0$,
i.e.~the subspace whose perturbation we want to compute. Mathematically, this decouples the unknowns thanks to
$\mathcal{R} \mathcal{L}_0 =  \mathcal{L}_0 \mathcal{R} = 0$ and choosing $\mathcal{R}(\mathcal{K}_0) = \mathcal{K}_0$.
This choice for $\mathcal{K}_0$ is natural since $\mathcal{M}_s$ is isomorphic to $\mathcal{M}_0$.
Then in a second step, one can project the equations with $1 - \mathcal{R}$ to determine $\mathcal{K}_k$.

The results recalled at the beginning of Section~\ref{ssec:formulas} have been obtained by applying this procedure to $\mathcal{L}_0$
stabilizing one subsystem of a composite quantum system towards a unique steady state,
and $\mathcal{L}_1$ expressing Hamiltonian coupling between this subsystem and another one. The leading order adiabatic elimination results for this case have been explicitly computed in~\cite{Azouit2017b}.
The present paper has encountered situations where the set of steady states of $\mathcal{L}_0$ has a different structure.
We then resort to the general procedure outlined in this section.
This explains how we have treated the elimination of $\mathcal{D}_b$ in applications
of Proposition~\ref{prop:D_3} and how we have addressed Section~\ref{sssec:ex3:nonuq}. A more detailed discussion of these two cases is included in Sections~\ref{sec:proof_D3_last} and~\ref{sec:computations_partly_dissipative} respectively.

\subsection{Time-periodic extension}\label{sec:ae_floquet}

This section aims to develop an extension of the approach of adiabatic elimination in systems with strongly dissipative degrees of freedom, to the case where the perturbation displays a periodic time-dependence with a frequency comparable in magnitude to the fast dissipation rate.
We consider dynamics with a similar timescale separation as before:
\begin{equation}\label{eq:floquet_start}
    \dot{\rho} = \mathcal{L}_0(\rho) +  \varepsilon \mathcal{L}_1(\rho,t),
\end{equation}
where we have the same assumptions on $\mathcal{L}_0$ as before, but  now $\mathcal{L}_1(t)$ is a periodic Lindbladian perturbation of period $\frac{2 \pi}{\omega}$. Furthermore, this perturbation should be rapidly oscillating, i.e.~$\omega \gg \Vert \varepsilon \mathcal{L}_1 \Vert$. Thanks to Floquet theory, we can expect $\mathcal{M}_0$ to be perturbed into a slightly altered attractive subspace
which now moves periodically in time, and on which also some slow dynamics is present.

We again parametrize the slow dynamics using a variable $\rho_s$ living in a space $\mathcal{M}_s$ isomorphic to $\mathcal{M}_0$,
and propose 
\begin{subequations}
\label{eq:anszats_floquet}
\begin{eqnarray}
    \rho(t) &=& \mathcal{K}_\varepsilon(\rho_s(t),t),\\
    \dot{\rho}_s(t) &=& \mathcal{L}_{s,\varepsilon}(\rho_s(t)).
\end{eqnarray}
\end{subequations}
as a solution staying in the ``slow'' invariant subspace of~\eqref{eq:floquet_start}. Here, $\mathcal{K}_\varepsilon(\cdot,t)$
is a $\frac{2 \pi}{\omega}$-periodic map characterizing the embedding of the perturbed slow subspace in the total space
and $\mathcal{L}_{s,\varepsilon}: \mathcal{M}_s \rightarrow \mathcal{M}_s$
is a stationary superoperator parametrizing the slow dynamics.
Plugging~\eqref{eq:anszats_floquet} into~\eqref{eq:floquet_start}, we obtain an invariance equation:
\begin{equation}\label{eq:geom_inv_floquet}
    \pdv{\mathcal{K}_\varepsilon}{t}\!\,(t) + \mathcal{K}_\varepsilon(t) \mathcal{L}_{s,\varepsilon} = \mathcal{L}_0 \mathcal{K}_\varepsilon(t) + \varepsilon \mathcal{L}_1(t) \mathcal{K}_\varepsilon(t),
\end{equation}
where the domain of all terms is $\mathcal{M}_s$.
We again expand both the stationary superoperator $\mathcal{L}_{s,\varepsilon}$ and the periodic map $\mathcal{K}_{\varepsilon}(t)$ in powers of $\varepsilon$.
\begin{align}
    \mathcal{K}_{\varepsilon}(t) &= \mathcal{K}_0(t) + \varepsilon \mathcal{K}_1(t) + \varepsilon^2 \mathcal{K}_2(t) + \cdots,\\
    \mathcal{L}_{s,\varepsilon} &= \mathcal{L}_{s,0} + \varepsilon \mathcal{L}_{s,1} + \varepsilon^2 \mathcal{L}_{s,2} + \cdots,
\end{align}
Collecting~\eqref{eq:geom_inv_floquet} into powers in $\varepsilon$ yields the set of recursive equations
\begin{align}
    \pdv{\mathcal{K}_0}{t}\!\,(t) + \mathcal{K}_0(t) \mathcal{L}_{s,0} &= \mathcal{L}_0 \mathcal{K}_0(t), \text{ and for }k\geq 1:\\
    \pdv{\mathcal{K}_k}{t}\!\,(t) + \sum_{j=0}^k \mathcal{K}_j(t) \mathcal{L}_{s,k-j} &= \mathcal{L}_0 \mathcal{K}_k(t) + \mathcal{L}_1(t) \mathcal{K}_{k-1}(t)\label{eq:recursive_floquet} .
\end{align}
We can choose $\mathcal{L}_{s,0} = 0$ and $\mathcal{K}_0$ time-independent and injective such that $\mathcal{R} \mathcal{K}_{0} = \mathcal{K}_0$, since for $\varepsilon=0$, the solutions in the slow subspace are stationary.
For $k=1$ we then obtain the following equation, to be satisfied by $\mathcal{L}_{s,1}$ and $\mathcal{K}_1$:
\begin{equation}\label{eq:first_order_floquet}
    \pdv{\mathcal{K}_1}{t}\!\,(t) + \mathcal{K}_0 \mathcal{L}_{s,1} = \mathcal{L}_0 \mathcal{K}_1(t) + \mathcal{L}_1(t) \mathcal{K}_0 \; .
\end{equation}
We split this equation up into four parts, by projecting either with $\mathcal{R}$ or $1 - \mathcal{R}$ on the one hand,
and by considering the time-average ($\bar{\cdot}$) and ripple ($\tilde{\cdot}$) of the equation separately on the other hand.
Since in this way it will be clear which terms depend on time, we drop the $t$ argument in what follows.

Applying $\mathcal{R}$ and the time-average to \eqref{eq:first_order_floquet}, we obtain
\begin{equation}\label{eq:first_order_red}
    \mathcal{K}_0 \mathcal{L}_{s,1} = \mathcal{R} \bar{\mathcal{L}}_1 \mathcal{K}_0 \; .
\end{equation}
This equation determines $\mathcal{L}_{s,1}$ uniquely, since $\mathcal{K}_0$ can be inverted on its image.
The application of $\mathcal{R}$ to the perturbation (its average part here) corresponds to the Zeno-effect that is well-known for stationary systems.

Applying $\mathcal{R}$ to \eqref{eq:first_order_floquet} and taking the ripple of the resulting equation, we get
\begin{equation}
    \mathcal{R} \dot{\mathcal{K}}_1 = \mathcal{R} \tilde{\mathcal{L}}_1 \mathcal{K}_0,
\end{equation}
of which a solution can be obtained via taking the zero-average time primitive:
\begin{equation}\label{eq:primitivation}
    \mathcal{R} \mathcal{K}_1 = \mathcal{R} \partial_t^{-1} \tilde{\mathcal{L}}_1 \mathcal{K}_0 + \mathcal{R} \bar{\mathcal{G}}_1.
\end{equation}
Here $\bar{\mathcal{G}}_1$ is an integration constant, playing the role of a gauge choice.
Equation~\eqref{eq:primitivation} is reminiscent of an averaging procedure, where oscillating terms are transformed away using a coordinate change $\varepsilon$-close to identity (here thus $\mathcal{K}_0+ \varepsilon \mathcal{K}_1$) and generated by the integral of the oscillating terms; see~\cite{guckenheimer2002nonlinear} for a canonical treatment of this averaging procedure. This is not surprising, since within the slow subspace, the effect of $\mathcal{L}_0$ reduces to zero, and we retain a small oscillating perturbation, which is exactly the setting where averaging procedures work well.

Applying $1 - \mathcal{R}$ to Eq.~\eqref{eq:first_order_floquet} by $1 - \mathcal{R}$ and taking the average of the resulting equation, we get
\begin{equation}
    0 = \mathcal{L}_{0} (1 - \mathcal{R}) \bar{\mathcal{K}}_1 + (1 - \mathcal{R}) \bar{\mathcal{L}}_1 \mathcal{K}_0,
\end{equation}
which has the formal solution
\begin{equation}\label{eq:first_order_pseudo_inverse}
    (1 - \mathcal{R}) \bar{\mathcal{K}}_1 = - \mathcal{L}_0^{-1} (1 - \mathcal{R}) \bar{\mathcal{L}}_1 \mathcal{K}_0.
\end{equation}
    Since $\mathcal{L}_0$ has a spectral gap, its restriction to the image of  $(1 - \mathcal{R})$ can rigorously be inverted,
because it has no eigenvalue zero there.
This pseudo-inverse of $\mathcal{L}_0$ is equally present in stationary adiabatic elimination and it expresses how the stationary part of the perturbation perturbs the slow subspace up to first order.

Lastly, applying $1 - \mathcal{R}$ to \eqref{eq:first_order_floquet} and taking the ripple of the resulting equation, we get
\begin{equation}\label{eq:new_pseudo_inverse}
    (1 - \mathcal{R}) \dot{\mathcal{K}}_1 = \mathcal{L}_0 (1 - \mathcal{R}) \tilde{\mathcal{K}}_1 + (1 - \mathcal{R}) \tilde{\mathcal{L}}_1 \mathcal{K}_0.
\end{equation}
To determine $(1 - \mathcal{R}) \tilde{\mathcal{K}}_1$ from this equation, we  introduce a decomposition into Fourier modes.
We can write
\[(1 - \mathcal{R}) \tilde{\mathcal{L}}_1(t) = \sum_{n \in \mathbb{Z}, n \neq 0} e^{i n t} (1 - \mathcal{R}) \tilde{\mathcal{L}}_{1,n} ,\]
for some superoperators $\tilde{\mathcal{L}}_{1,n}$, since $\tilde{\mathcal{L}}_1$ has zero average.
Decomposing in the same way the tentative solution
\[(1 - \mathcal{R}) \tilde{\mathcal{K}}_1(t) = \sum_{n \in \mathbb{Z}, n \neq 0} e^{i n t} (1 - \mathcal{R}) \mathcal{K}_{1,n},\]
and plugging this into~\eqref{eq:new_pseudo_inverse}, we see that for every $n \neq 0$, we are looking for the stationary superoperator $(1 - \mathcal{R}) \mathcal{K}_{1,n}$ such that
\begin{equation}\label{eq:K_1_3rd_contr}
    \qty(\mathcal{L}_{0} - i n)(1 - \mathcal{R})\mathcal{K}_{1,n} = - (1 - \mathcal{R})\tilde{\mathcal{L}}_{1,n} \mathcal{K}_0.
\end{equation}
Here we can really see that, since the time-dependence of $\mathcal{L}_1$ is as fast as the dissipation $\mathcal{L}_0$, the combined effect of the two has to be inverted to obtain the oscillating part of the correction to the slow subspace in the $(1 - \mathcal{R})$ subspace. We thus need $\mathcal{L}_{0} - i n$ to be invertible on the image of $1 - \mathcal{R}$.
Because $\mathcal{L}_0$ restricted to the image of $1 - \mathcal{R}$ only has eigenvalues with strictly negative real part,
a shift in its spectrum by $- i n, n\in \mathbb{Z}$ can never move an eigenvalue to the origin, and hence $\mathcal{L}_{0} - i n$ can formally be inverted in the above equation. 
This can be done for every fixed $n$ separately, or if available a spectral decomposition of $\mathcal{L}_0$ could allow to define all inverses at once.

Equation~\eqref{eq:recursive_floquet} for $k \geq 2$ can be treated in an analogous way, and the general higher-order solution goes as follows:
\begin{subequations}
\label{eq:general_red_dyn}
\begin{align}
    \hspace*{-0.16cm}\mathcal{K}_0 \mathcal{L}_{s,k} &= \mathcal{R} \bar{\mathcal{A}}_k,\\
    \hspace*{-0.16cm}\mathcal{R} \mathcal{K}_k &= \mathcal{R} \partial_t^{- 1} \qty(\tilde{\mathcal{A}}_k - \tilde{\mathcal{B}}_k) + \mathcal{R} \bar{\mathcal{G}}_k,\\
    \hspace*{-0.16cm}(1  -  \mathcal{R}) \bar{\mathcal{K}}_k &= -  \mathcal{L}_0^{- 1}  (1 - \mathcal{R}) \qty(\bar{\mathcal{A}}_k - \bar{\mathcal{B}}_k),\\
    \hspace*{-0.16cm}(1  -  \mathcal{R}) \tilde{\mathcal{K}}_k &= -  \qty(\mathcal{L}_0 - \partial_t)^{- 1}   (1 - \mathcal{R}) \qty(\tilde{\mathcal{A}}_k - \tilde{\mathcal{B}}_k),\hphantom{::}
\end{align}
\end{subequations}
with
\begin{align*}
    \mathcal{A}_k &= \mathcal{L}_1 \mathcal{K}_{k \shortminus 1},\\
    \mathcal{B}_k &= \sum_{j=1}^{k - 1} \mathcal{K}_j \mathcal{L}_{s,k \shortminus j}.
\end{align*}
Here, $\mathcal{R} \bar{\mathcal{G}}_k$ is a general gauge choice that can be made at every order.
All inverses are well-defined for the same reasons as before, and it is easy to check that the above recursive relation provides a solution,
by plugging it into~\eqref{eq:recursive_floquet}.

In appendix~\ref{reduced_formulas}, we apply the method of this section to the model~\eqref{eq:model_strong} (resp.~\eqref{eq:model_ultrastrong} under ultra-strong driving) and we show that $\mathcal{K}_0 + \varepsilon \mathcal{K}_1$ can be written as a Kraus map up to $\mathcal{O}(\varepsilon^2)$ terms, choosing $\bar{\mathcal{G}}_1 = 0$, so $\mathcal{R} \tilde{\mathcal{K}}_1 = \mathcal{R} \partial_t^{-1} \tilde{\mathcal{L}}_1 \mathcal{K}_0$, and $\mathcal{R} \bar{\mathcal{K}}_1 = 0$.
Furthermore, $\mathcal{L}_{s,1}$ is a Hamiltonian on the target T, and $\mathcal{L}_{s,2}$ is the sum of a Hamiltonian and
a Lindbladian on T. Thus the proposed perturbative series preserves the quantum structure of Lindbladian reduced dynamics and CPTP mappings up to second order. In fact, one can prove that this remains the case for a general bipartite scenario.
\begin{thm}\label{thm:quantum_structure}
    Consider the model~\eqref{eq:floquet_start} where the Lindbladian $\mathcal{L}_0$ acts only on one subsystem (F) of
    a bipartite quantum system, and exponentially stabilizes F towards a unique steady state at a rate $\kappa$;
    and, $\mathcal{L}_1(t)$ expresses a $\frac{2 \pi}{\omega}$-periodic Hamiltonian coupling between the F-subsystem and the second one (S).
    Assume $\varepsilon \norm{\mathcal{L}_1} \ll \kappa$ and $\varepsilon \norm{\mathcal{L}_1} \ll \omega$.
    When choosing $\bar{\mathcal{G}}_1 = 0$ in~\eqref{eq:primitivation},
    $\mathcal{L}_{s,1}$ takes the form of a Hamiltonian, $\mathcal{L}_{s,2}$ is the sum of Hamiltonian and a
    Lindbladian term, and $\mathcal{K}_0 + \varepsilon \mathcal{K}_1$ can be written as a CPTP-map up to terms of order
    $\varepsilon^2$.
\end{thm}
\begin{proof}
    Since $\mathcal{L}_0$ only acts on F, it trivially corresponds to a Lindbladian $\mathcal{L}_F$ acting on F such that
    $\mathcal{L}_0 = \mathrm{identity} \otimes \mathcal{L}_F$.
    The proof then consists of a straightforward adaptation of Lemma 4 and 5 in appendix A of~\cite{Azouit2017c} to a general
    pseudo-inverse ${\qty(\mathcal{L}_F - i n)}^{-1}, n \in \mathbb{Z}$ instead of only ${\mathcal{L}_F}^{-1}$ in the
    original work.
\end{proof}

\subsection{Proof of Proposition~\ref{prop:D_3}, last item}\label{sec:proof_D3_last}

In Section \ref{sec:noisy_decoupling}, we also go back to the general procedure for adiabatic elimination, namely when eliminating fast degrees of freedom which do not necessarily coincide with a subsystem. A first point where this appears is Proposition~\ref{prop:D_3}, where we consider the possibility to first eliminate part of E, namely the one corresponding to strictly negative eigenvalues of $\mathcal{D}_b$, and reconsider the system from there. Our claims in Proposition~\ref{prop:D_3} involve nothing special and can only be further worked out on examples, except for the claim in the last item. We next provide its proof.

We consider the system obtained after first-order adiabatic elimination of $\mathcal{D}_b$, according to the procedure of Appendix~\ref{app:adelgenform}, as being the new target-environment model, and we denote things as if this was the starting situation (e.g.~writing $\bar{\rho}_E$ for the unique steady state of the environment after already having reduced it with $\mathcal{D}_b$). Without loss of generality, we assume that the $E_k$ have been redefined such that $\Tr(E_k\, \bar{\rho}) = 0$, and also that each $E_k$ is Hermitian. We denote by $\mathcal{D}_E$ the remaining Lindbladian dissipation on this reduced environment. The proof ideas are similar to those for proving positivity of $X$ in the adiabatic elimination theory paper~\cite{Azouit2017b}.

The goal is thus to investigate when the induced dissipation matrix $X$ in Section~\ref{ssec:formulas} might vanish. Since $X$ is nonnegative, we can
focus on its diagonal. This means, we want each diagonal element $x_k := \Tr(E_k (Q_k + Q_k^\dagger)) = 0$. Here $Q_k$ is the solution of
$\mathcal{D}_{E}(Q_k) = - E_k \bar{\rho}_E$. Using the integral formula for the inverse of a negative operator, we can write
$Q_k = \int_0^{\infty}  \exp[\mathcal{D}_{E} t] (E_k \bar{\rho}_E) dt$ and thus 
\begin{eqnarray*}
x_k &=& \Tr\left(E_k \; \int_0^{\infty}  \exp[\mathcal{D}_{E} t] (E_k \bar{\rho}_E) dt \right) \\
&=&  \Tr\left( \int_0^{\infty}  \exp[\mathcal{D}^*_{E} t](E_k) dt \; (E_k \bar{\rho}_E) \right) \\
&=& \Tr( M_k \;  (E_k \bar{\rho}_E))
\end{eqnarray*}
where $\mathcal{D}^*$ denotes the dual superoperator of $\mathcal{D}$, and  $M_k$ must satisfy $\mathcal{D}^*_{E}(M_k) = -E_k$. Replacing $E_k$ in this way in the expression of $x_k$ and using that $\mathcal{D}_{E}(\bar{\rho}_E) = 0$, we get after a few computations:
$$x_k = \sum_j \; \Tr([M_k,D_j] \; \bar{\rho}_E \; [M_k,D_j]^\dagger) \; ,$$
with $D_j$ the dissipation channel operators of $\mathcal{D}_{E}$.

Now, when $\bar{\rho}_E$ has full rank, the only way to get $x_k=0$ is to have $[M_k,D_j]=0$ for all $D_j$. But this would imply $\mathcal{D}^*_{E}(M_k) = 0$, contradicting how $M_k$ must be computed. When $\bar{\rho}_E$ has reduced rank we apply the same argument to the block-diagonal part corresponding to the support of $\bar{\rho}_E$.

\subsection{Adiabatic elimination computations for Section~\ref{sssec:ex3:nonuq}:}\label{sec:computations_partly_dissipative}

There is a second point in Section \ref{sec:noisy_decoupling} where we go back to the general formalism of Appendix~\ref{app:adelgenform}, because we eliminate degrees of freedom which do not necessarily coincide with a subsystem. Indeed, in Section~\ref{sssec:ex3:nonuq}, we consider how to treat a case where the environment E is allowed to keep slow degrees of freedom. We provide formulas showing how on the considered example, the ``$g^2 / \kappa$ '' scaling breaks down for dispersive coupling, while it appears to keep holding under resonant coupling. We here give some details on the computations behind those formulas, focusing on the case of resonant interaction.

The fast dynamics happens at timescale $\kappa_z$, while the slow one involves $\kappa_-, \kappa_+, g$ thus $(\kappa_-,
\kappa_+, g) / \kappa_z$ are all of order $\varepsilon$ in the notation of Appendix~\ref{app:adelgenform}.

Following the general structure explained there, we parameterize the slow dynamics with $\rho_s = (\rho_g,\; \rho_e)$ where both are nonnegative
operators and $\Tr(\rho_g+\rho_e) = 1$. Indeed, the linear superoperator 
$$K_0(\rho_s) = \rho_g \otimes \q{g}\qd{g}+\rho_e \otimes \q{e}\qd{e}$$
maps this reduced state onto the steady states of the fast dynamics $\mathcal{D}_{\sigma_z}$. Furthermore, the convergence under the fast dynamics happens according to 
$$\mathcal{R}(\rho) = (\qd{g}\rho\q{g}\;,\;\; \qd{e}\rho\q{e}) \;.$$

Applying $\mathcal{R}$ to the equation associated to $\varepsilon^1$ (see general expression above), we get
\begin{align*}
 \kappa_z \varepsilon\,L_{g,1} &=   \kappa_-  \qd{g} \mathcal{D}_{\sigma_-}( \rho_g \otimes \q{g}\qd{g}+\rho_e \otimes \q{e}\qd{e})  \q{g}\\
                                              &+   \kappa_+  \qd{g} \mathcal{D}_{\sigma_+}( \rho_g \otimes \q{g}\qd{g}+\rho_e \otimes \q{e}\qd{e})  \q{g} ,\\
\kappa_z \varepsilon\, L_{e,1} &=  \kappa_- \qd{e}  \mathcal{D}_{\sigma_-} ( \rho_g \otimes \q{g}\qd{g}+\rho_e \otimes \q{e}\qd{e})  \q{e} ,\\
                               &+  \kappa_+ \qd{e}  \mathcal{D}_{\sigma_+} ( \rho_g \otimes \q{g}\qd{g}+\rho_e \otimes \q{e}\qd{e})  \q{e} ,
\end{align*}
while the Hamiltonian moves $\q{g}\qd{g}$ and $\q{e}\qd{e}$ onto $\q{g}\qd{e}$ and $\q{e}\qd{g}$ which get canceled by $\mathcal{R}$. Working out the above yields the reported equation for $\kappa_z \varepsilon \mathcal{L}_{s,1}$.

Next, we go back to the $\varepsilon^1$ equation, without applying $\mathcal{R}$, and parameterize $\mathcal{K}_1(\rho_s) = \sum_{j,k \in \{g,e \}} \mathcal{K}_{j,k}(\rho_s) \otimes \q{j}\qd{k}$. We observe that $\mathcal{K}_{e,e} \otimes \q{e}\qd{e}$ and $\mathcal{K}_{g,g} \otimes \q{g}\qd{g}$ cancel under application of $\mathcal{L}_0 = \kappa_z \mathcal{D}_z$, and therefore these are gauge degrees of freedom associated to non-uniqueness of the parameterization; we can take them as $\mathcal{K}_{e,e} = \mathcal{K}_{g,g} = 0$ for simplicity. The remaining equations impose:
\begin{eqnarray*}
\varepsilon \, \mathcal{K}_{e,g}(\rho_s) &=& \tfrac{-g\, i}{2 \kappa_z} (T_x \rho_g - \rho_e T_x + i T_y \rho_g - i \rho_e T_y) \\
\varepsilon \, \mathcal{K}_{g,e}(\rho_s) &=& \tfrac{-g\, i}{2 \kappa_z} (T_x \rho_e - \rho_g T_x + i \rho_g T_y  - i T_y \rho_e) \;.
\end{eqnarray*}

This can be plugged into the equation associated to $\varepsilon^2$, to which again we apply $\mathcal{R}$ in order to obtain $\mathcal{L}_{s,2}$. In the term $\mathcal{L}_1(\mathcal{K}_1)$ from the abstract expression, now only the Hamiltonian contribution remains as it can map terms of the form $\q{g}\qd{e}$, $\q{e}\qd{g}$ in $\mathcal{K}_1$ towards terms in $\q{e}\qd{e}$, $\q{g}\qd{g}$ which are conserved by $\mathcal{R}$. Simple algebraic computations then yield the dynamics announced in the main text.

\begin{widetext}

\section{Derivation of the reduced model of section~\ref{sec:reduced_model}}\label{reduced_formulas}

In this section, we apply the general formulas derived in Appendix~\ref{app:ae_formulas} to derive the reduced model of Section~\ref{sec:reduced_model}.
We are dealing with the particular case of a bipartite system, so in the notation of app.~\ref{app:ae_formulas}, $\mathcal{H} = \mathcal{H}_T \otimes \mathcal{H}_E$.
For both the model of~\eqref{eq:model_strong} and~\eqref{eq:model_ultrastrong},
the fast dynamics acts only on E, and quickly drives it to a unique steady state $\bar{\rho}_E$ were it not for the T-E coupling, which is considered the perturbation with $g \ll \bestguess{\omega}{1}, \Lambda, \kappa_{\pm}, \kappa_{\alpha \pm}, \kappa_{\alpha x}$.
We will calculate $\bar{\rho}_E$ explicitly below for both cases, but it is clear that the unperturbed slow subspace $\mathcal{M}_0$ is given by the set of linear operators
\[X_T \otimes \bar{\rho}_E,\]
where $X_T$ acts on $\mathcal{H}_T$.
Hence, in the notation of app.~\ref{app:ae_formulas}, it is natural to choose $\mathcal{M}_s$ as the space of operators acting on $\mathcal{H}_T$,
and \[\mathcal{K}_0(\rho_s) = \rho_s \otimes \bar{\rho}_E.\]
In this way, $\mathcal{L}_{s,1}$, $\mathcal{L}_{s,2}$ are superoperators corresponding to the target Hilbert space $\mathcal{H}_T$ alone,
and the reduced model obtained can truly be seen as describing the induced decoherence on the target system.
How the target becomes entangled with the environment will be described by the map $\mathcal{K}_1$ up to first order in $g$.

\subsection{Case of strong driving}\label{sec:reduc_strong}

We recapitulate the full model here:
\begin{align}
    \dot{\rho} &= - i \frac{\Lambda}{2} \comm{\mathbb{1}_T \otimes \sigma_{\alpha x}}{\rho}
    + \kappa_- \mathcal{D}_{\mathbb{1}_T \otimes \sigma_-}(\rho) + \kappa_+ \mathcal{D}_{\mathbb{1}_T \otimes \sigma_+}(\rho)\nonumber\\
    &- ig \comm{T_z \otimes \sigma_z + e^{i \bestguess{\omega}{1} t} T_- \otimes \sigma_+ + e^{- i \bestguess{\omega}{1} t} T_+ \otimes \sigma_-}{\rho}.\nonumber
\end{align}

    In the notation of appendix~\ref{sec:ae_floquet}, assuming $g \ll \bestguess{\omega}{1}, \kappa_-$ at least, we can thus define $\varepsilon = \frac{g}{\bestguess{\omega}{1}}$,
\begin{equation}
    \mathcal{L}_0 = - i \frac{\Lambda}{2} \comm{\mathbb{1}_T \otimes \sigma_{\alpha x}}{\cdot}
    + \kappa_- \mathcal{D}_{\mathbb{1}_T \otimes \sigma_-} + \kappa_+ \mathcal{D}_{\mathbb{1}_T \otimes \sigma_+},
\end{equation}
and 
\begin{align*}
    \mathcal{L}_1(t) = - i \bestguess{\omega}{1} \comm{T_z \otimes \sigma_z + e^{i \bestguess{\omega}{1} t} T_- \otimes \sigma_+ + e^{- i \bestguess{\omega}{1} t} T_+ \otimes \sigma_-}{\cdot}.
\end{align*}

It is straightforward to verify that the fast dynamics $\mathcal{L}_0$ drives the environment to a unique steady state
\begin{equation}
    \bar{\rho}_E = \frac{\mathbb{1}_E + \xi_\infty \sigma_+ + \xi_\infty^* \sigma_- + z_\infty \sigma_z }{2},
\end{equation}
with
\begin{align}
    \xi_\infty &= -2 \Lambda \cos(\alpha) \frac{\kappa_\Delta}{\kappa_\Sigma} \frac{2 \Lambda \sin(\alpha) + i \kappa_\Sigma}{\kappa_\Sigma^2 + 2 \Lambda^2\qty(1 + \sin^2(\alpha))},\\
    z_\infty &= - \frac{\kappa_\Delta}{\kappa_\Sigma} \frac{4 \Lambda^2 \sin^2(\alpha) + \kappa_\Sigma^2}{\kappa_\Sigma^2 + 2 \Lambda^2\qty(1 + \sin^2(\alpha))},
\end{align}
where we have defined 
\begin{align}
    \kappa_\Sigma &:= \kappa_- + \kappa_+,\\
    \kappa_\Delta &:= \kappa_- - \kappa_+.
\end{align}

Since we are interested in the regime of strong driving where $\omega_2 \gg \kappa_\Sigma$, we also compute the leading order in $\frac{1}{\omega_2}$ of all quantities in this section.
For this, $\cos(\alpha)$ should be put to $1$ since $\alpha$ goes to zero with $\omega_2 \rightarrow \infty$ , and it
should be remembered that $\Lambda \sin(\alpha) = \NDelta$. Thus
\begin{align}
    \xi_\infty &= - \frac{\kappa_{\Delta}}{\kappa_{\Sigma}} \frac{\qty( i \kappa_{\Sigma} + 2 \NDelta)}{\omega_2} + \mathcal{O}\qty(\frac{1}{\omega_2^2}),\\
    z_\infty &= - \frac{\kappa_{\Delta}}{\kappa_{\Sigma}} \frac{\left(\kappa_{\Sigma}^{2} + 2 \NDelta^{2}\right)}{\omega_2^{2}}+ \mathcal{O}\qty(\frac{1}{\omega_2^3}) \; .
\end{align}
The steady state thus converges to the maximally mixed state in the limit of strong driving.

For the projector $\mathcal{R}$ we have \[\mathcal{R}(X_\textrm{TE}) = \Tr_E(X_\textrm{TE}) \otimes \bar{\rho}_E \qq*{,} \forall X_\textrm{TE}.\]
Equation~\eqref{eq:first_order_red} yields the following expression for the first-order reduced dynamics:
\begin{align*}
    \varepsilon \mathcal{L}_{s,1}(\rho_s) \otimes \bar{\rho}_E &= \varepsilon \mathcal{R}(\overline{\mathcal{L}_1(\rho_s \otimes \bar{\rho}_E)}) \nonumber\\
                                                   &= - i g \Tr_E(\comm{T_z \otimes \sigma_z}{\rho_s \otimes \bar{\rho}_E}) \otimes \bar{\rho}_E \nonumber\\
                                                   &= - i g z_\infty \comm{T_z}{\rho_s} \otimes \bar{\rho}_E,
\end{align*}
readily yielding
\begin{equation}
    \varepsilon \mathcal{L}_{s,1}(\rho_s) = - i g z_\infty \comm{T_z}{\rho_s}.
\end{equation}

Equation~\eqref{eq:primitivation} in turn yields
\begin{align}
    \varepsilon \mathcal{R} \mathcal{K}_1(\rho_s) &= \varepsilon \mathcal{R} \partial_t^{-1} \tilde{\mathcal{L}}_1 \mathcal{K}_0(\rho_s)\nonumber\\
                              &= - \frac{g}{\bestguess{\omega}{1}} \Tr_E(\comm{e^{i \bestguess{\omega}{1} t} T_- \otimes \sigma_+ - e^{- i \bestguess{\omega}{1} t} T_+ \otimes \sigma_-}{\rho_s \otimes \bar{\rho}_E}) \otimes \bar{\rho}_E\nonumber\\
                              &= - \frac{i g}{2 \bestguess{\omega}{1}} \comm{i \xi_\infty^* e^{i \bestguess{\omega}{1}
                              t} T_- - i \xi_\infty e^{- i \bestguess{\omega}{1} t} T_+}{\rho_s} \otimes
                              \bar{\rho}_E,\label{eq:K_1_a}
\end{align}
where we have put the integration constant to zero as a gauge choice.
Equation~\eqref{eq:first_order_pseudo_inverse} yields a second part of $\mathcal{K}_1$:
\begin{align}
    \varepsilon \mathcal{L}_0(1 - \mathcal{R}) \bar{\mathcal{K}}_1(\rho_s) &= - \varepsilon (1 - \mathcal{R}) \bar{\mathcal{L}}_1 \mathcal{K}_0(\rho_s)\nonumber\\
                                                                &= i g \comm{T_z \otimes \sigma_z}{\rho_s \otimes \bar{\rho}_E}\nonumber\\
                                                                &- i g \Tr_E(\comm{T_z \otimes \sigma_z}{\rho_s \otimes \bar{\rho}_E}) \otimes \bar\rho_E \nonumber\\
                                                                &= i g \qty(T_z \rho_s \otimes \bar{\sigma}_z
                                                                \bar{\rho}_E - \rho_s T_z \otimes \bar{\rho}_E
                                                                \bar{\sigma}_z),
\end{align}
with $\bar{\sigma}_z = \sigma_z - \Tr(\sigma_z \bar{\rho}_E)\mathbb{1}_E = \sigma_z - z_\infty \mathbb{1}_E$.
Note that taking the partial trace over E of the right-hand side gives zero, since $\Tr(\bar{\sigma}_z \bar{\rho}_E) = 0$.
Hence $\mathcal{L}_0$ can be inverted to obtain, formally,
\begin{equation}\label{eq:K_1_b}
    \varepsilon (1 - \mathcal{R}) \bar{\mathcal{K}}_1(\rho_s) = i g \qty(T_z \rho_s \otimes \mathcal{L}_0^{-1}(\bar{\sigma}_z \bar{\rho}_E) - \rho_s T_z \otimes \mathcal{L}_0^{-1}(\bar{\rho}_E \bar{\sigma}_z)).
\end{equation}
To carry out the inversion we use matrix representations in the Pauli basis. In the standard Pauli basis ($\sigma_x, \sigma_y, \sigma_z$), we obtain the following matrix representation for $\mathcal{L}_0$:
\begin{equation}
    \qty[\mathcal{L}_0] = \mqty(- \frac{\kappa_\Sigma}{2} & - \Lambda \sin(\alpha) & 0  \\ \Lambda \sin(\alpha) & - \frac{\kappa_\Sigma}{2} & - \Lambda \cos(\alpha) \\
    0 & \Lambda \cos(\alpha) & - \kappa_\Sigma ),
\end{equation}
with $\mathrm{det}\qty[\mathcal{L}_0] = - \frac{\kappa_\Sigma}{4} \qty(\kappa_\Sigma^2 + 2 \Lambda^2 (1 + \sin^2(\alpha)))$.
For its inverse $\qty[\mathcal{L}_0^{-1}]$ we hence obtain
\begin{equation*}
    \frac{1}{\mathrm{det}\qty[\mathcal{L}_0]}\mqty(\frac{\kappa_{\Sigma}^{2}}{2} + \Lambda^{2} \cos^{2}{\left(\alpha \right)} & - \kappa_{\Sigma} \Lambda \sin{\left(\alpha \right)} & \frac{\Lambda^{2} \sin{\left(2 \alpha \right)}}{2}\\\kappa_{\Sigma} \Lambda \sin{\left(\alpha \right)} & \frac{\kappa_{\Sigma}^{2}}{2} & - \frac{\kappa_{\Sigma} \Lambda \cos{\left(\alpha \right)}}{2}\\\frac{\Lambda^{2} \sin{\left(2 \alpha \right)}}{2} & \frac{\kappa_{\Sigma} \Lambda \cos{\left(\alpha \right)}}{2} & \frac{\kappa_{\Sigma}^{2}}{4} + \Lambda^{2} \sin^{2}{\left(\alpha \right)}).
\end{equation*}
In turn, $\bar{\sigma}_z \bar{\rho}_E$ takes the following vector representation in the Pauli basis:
\begin{equation}
    \qty[\bar{\sigma}_z \bar{\rho}_E] = \frac{1}{2} \mqty(- i y_\infty - x_\infty z_\infty \\ i x_\infty - y_\infty z_\infty \\ 1 - z_\infty^{2}).
\end{equation}
Straightforward but tedious calculations then give
\begin{equation}\label{eq:exact_A_z_bar}
    \qty[\mathcal{L}_0^{-1}(\bar{\sigma}_z \bar{\rho}_E)] = \frac{1}{8 \, \mathrm{det}\qty[\mathcal{L}_E]}
    \mqty(- 4 \kappa_{\Sigma} \omega_2 \left(i x_{\infty} - y_{\infty} z_{\infty}\right) \sin{\left(\alpha \right)} - 2 \omega_2^{2} \left(z_{\infty}^{2} - 1\right) \sin{\left(2 \alpha \right)} - 2 \left(\kappa_{\Sigma}^{2} + 2 \omega_2^{2} \cos^{2}{\left(\alpha \right)}\right) \left(x_{\infty} z_{\infty} + i y_{\infty}\right)\\2 \kappa_{\Sigma} \left(\kappa_{\Sigma} \left(i x_{\infty} - y_{\infty} z_{\infty}\right) + \omega_2 \left(z_{\infty}^{2} - 1\right) \cos{\left(\alpha \right)} - 2 \omega_2 \left(x_{\infty} z_{\infty} + i y_{\infty}\right) \sin{\left(\alpha \right)}\right)\\2 \kappa_{\Sigma} \omega_2 \left(i x_{\infty} - y_{\infty} z_{\infty}\right) \cos{\left(\alpha \right)} - 2 \omega_2^{2} \left(x_{\infty} z_{\infty} + i y_{\infty}\right) \sin{\left(2 \alpha \right)} + \left(1 - z_{\infty}^{2}\right) \left(\kappa_{\Sigma}^{2} + 4 \omega_2^{2} \sin^{2}{\left(\alpha \right)}\right)).
\end{equation}
Focussing on the leading-order in $\frac{1}{\omega_2}$ yields the following:
\begin{equation}
    \qty[\mathcal{L}_0^{-1}(\bar{\sigma}_z \bar{\rho}_E)] = 
    \mqty(\frac{- \NDelta + i \kappa_{\Delta}}{\kappa_{\Sigma} \omega_2}\\\frac{1}{2 \omega_2}\\\frac{\NDelta \left(- \NDelta + i \kappa_{\Delta}\right) - \frac{\kappa_{\Sigma} \left(2 i \kappa_{\Delta} + \kappa_{\Sigma}\right)}{4}}{\kappa_{\Sigma} \omega_2^{2}}),
\end{equation}
and further
\begin{equation*}
    \varepsilon (1 - \mathcal{R}) \bar{\mathcal{K}}_1(\rho_s) = \frac{g}{\Lambda} (i T_z \otimes \bar{M}_z) (\rho_s \otimes \bar{\rho}_E) +  \frac{g}{\Lambda} (\rho_s \otimes \bar{\rho}_E) {\qty(i T_z \otimes \bar{M}_z)}^\dag,
\end{equation*}
with
\begin{equation}
    \qty[\bar{M}_z] = \mqty(- \frac{2 \NDelta}{\kappa_{\Sigma}}\\1\\\frac{- 4 \NDelta^{2} + 4 \kappa_{\Delta}^{2} - \kappa_{\Sigma}^{2}}{2 \kappa_{\Sigma} \omega_2}\\\frac{\kappa_{\Delta} \left(2 \NDelta - \kappa_{\Sigma}\right)}{\kappa_{\Sigma} \omega_2})
              + i \mqty(\frac{2 \kappa_{\Delta}}{\kappa_{\Sigma}}\\0\\\frac{4 \NDelta \kappa_{\Delta}}{\kappa_{\Sigma} \omega_2}\\- \frac{2 \kappa_{\Delta}^{2}}{\kappa_{\Sigma} \omega_2}).
\end{equation}

For the last part of $\mathcal{K}_1$, consider~\eqref{eq:new_pseudo_inverse}:
\begin{equation}\label{eq:laatste_K}
     \varepsilon (\mathcal{L}_0 - \partial_t) (1 - \mathcal{R}) \tilde{\mathcal{K}}_1(\rho_s) =  i g (1 - \mathcal{R})
    \qty(\comm{e^{i \bestguess{\omega}{1} t} T_- \otimes \sigma_+ + e^{- i \bestguess{\omega}{1} t} T_+ \otimes \sigma_-}{\rho_s \otimes \bar{\rho}_E}).
\end{equation}
Introducing
\begin{align}
    \bar{\sigma}_+ &:= \sigma_+ - \Tr(\sigma_+ \bar{\rho}_E) \mathbb{1}_E = \sigma_+ - \frac{\xi_\infty^*}{2} \mathbb{1}_E,\nonumber \\
    \bar{\sigma}_- &:= \sigma_- - \Tr(\sigma_- \bar{\rho}_E) \mathbb{1}_E = \sigma_- - \frac{\xi_\infty}{2} \mathbb{1}_E,\nonumber
\end{align}
we can write the right-hand side of~\eqref{eq:laatste_K} as
\begin{equation*}
     i g \comm{e^{i \bestguess{\omega}{1} t} T_- \otimes \bar{\sigma}_+ + e^{- i \bestguess{\omega}{1} t} T_+ \otimes \bar{\sigma}_-}{\rho_s \otimes \bar{\rho}_E}
    =i g e^{i \bestguess{\omega}{1} t} T_- \rho_s \otimes \bar{\sigma}_+ \bar{\rho}_E + i g e^{- i \bestguess{\omega}{1} t} T_+ \rho_s \otimes \bar{\sigma}_- \bar{\rho}_E + \textrm{h.c.}
\end{equation*}
At this point we can split $(1 - \mathcal{R}) \tilde{\mathcal{K}}_1$ up into two parts:
\begin{align}
    \varepsilon (1 - \mathcal{R}) \tilde{\mathcal{K}}_1(\rho_s) &= i g e^{i \bestguess{\omega}{1} t} {\qty(\mathcal{L}_0 - i \bestguess{\omega}{1})}^{-1}\qty(T_- \rho_s \otimes \bar{\sigma}_+ \bar{\rho}_E)
                  + i g e^{-i \bestguess{\omega}{1} t} {\qty(\mathcal{L}_0 + i \bestguess{\omega}{1})}^{-1}\qty(T_+ \rho_s \otimes \bar{\sigma}_- \bar{\rho}_E) + \textrm{h.c.}  \nonumber\\
                  &= i g e^{i \bestguess{\omega}{1} t} T_- \rho_s \otimes {\qty(\mathcal{L}_0 - i \bestguess{\omega}{1})}^{-1}\qty(\bar{\sigma}_+ \bar{\rho}_E)
                  + i g e^{-i \bestguess{\omega}{1} t} T_+ \rho_s \otimes {\qty(\mathcal{L}_0 + i
                  \bestguess{\omega}{1})}^{-1}\qty(\bar{\sigma}_- \bar{\rho}_E) + \textrm{h.c.} \label{eq:K_1_c}
\end{align}
We obtain the following matrix representations:
\begin{equation*}
    \qty[\bar{\sigma}_+ \bar{\rho}_E] = \frac{1}{8} \mqty(- 2 z_{\infty} - \left(\xi_{\infty} + \xi^*_{\infty}\right) \xi^*_{\infty} + 2\\i \left(- 2 z_{\infty} - 1 \left(\xi_{\infty} - \xi^*_{\infty}\right) \xi^*_{\infty} + 2\right)\\\left(2 - 2 z_{\infty}\right) \xi^*_{\infty}),
\end{equation*}
and

\begin{align*}
    &\hspace{5cm}{\mathrm{det}\qty[\mathcal{L}_0 \mp i \bestguess{\omega}{1}]} \qty[{\qty(\mathcal{L}_0 \mp i \bestguess{\omega}{1})}^{-1}] = \\
        &\frac{1}{2} \mqty(
\kappa_{\Sigma}^{2} + 3 i \kappa_{\Sigma} \bestguess{\omega}{1} + 2 \omega_2^{2} \cos^{2}{\left(\alpha \right)} - 2 \bestguess{\omega}{1}^{2} & - 2 \omega_2 \left(\kappa_{\Sigma} + i \bestguess{\omega}{1}\right) \sin{\left(\alpha \right)} & \omega_2^{2} \sin{\left(2 \alpha \right)}\\2 \omega_2 \left(\kappa_{\Sigma} + i \bestguess{\omega}{1}\right) \sin{\left(\alpha \right)} & \kappa_{\Sigma}^{2} + 3 i \kappa_{\Sigma} \bestguess{\omega}{1} - 2 \bestguess{\omega}{1}^{2} & - \omega_2 \left(\kappa_{\Sigma} + 2 i \bestguess{\omega}{1}\right) \cos{\left(\alpha \right)}\\\omega_2^{2} \sin{\left(2 \alpha \right)} & \omega_2 \left(\kappa_{\Sigma} + 2 i \bestguess{\omega}{1}\right) \cos{\left(\alpha \right)} & \frac{\kappa_{\Sigma}^{2}}{2} + 2 i \kappa_{\Sigma} \bestguess{\omega}{1} + 2 \omega_2^{2} \sin^{2}{\left(\alpha \right)} - 2 \bestguess{\omega}{1}^{2}
        ),
\end{align*}
with
\begin{equation}
    \mathrm{det}\qty[\mathcal{L}_0 \mp i \bestguess{\omega}{1}] = - \frac{\kappa_{\Sigma}^{3}}{4} + \frac{\kappa_{\Sigma} \omega_2^{2} \cos^{2}{\left(\alpha \right)}}{2} - \kappa_{\Sigma} \omega_2^{2} + 2 \kappa_{\Sigma} \bestguess{\omega}{1}^{2}
     \pm i \qty(\frac{5 \kappa_{\Sigma}^{2} \bestguess{\omega}{1}}{4} - \omega_2^{2} \bestguess{\omega}{1} + \bestguess{\omega}{1}^{3}).
\end{equation}
Tedious calculations then show that
\begin{subequations}
\label{eq:M_pm_bar_exact}
    \begin{align}
        {\qty(\mathcal{L}_0 - i \bestguess{\omega}{1})}^{-1}\qty(\bar{\sigma}_+ \bar{\rho}_E) = \frac{1}{\kappa_\Sigma + i \bestguess{\omega}{1}} \bar{M}_+ \bar{\rho}_E,\\
        {\qty(\mathcal{L}_0 + i \bestguess{\omega}{1})}^{-1}\qty(\bar{\sigma}_- \bar{\rho}_E) = \frac{1}{\kappa_\Sigma - i \bestguess{\omega}{1}} \bar{M}_- \bar{\rho}_E,
    \end{align}
\end{subequations}
with $\bar{M}_\pm$ operators such that $\Tr(\bar{M}_\pm \bar{\rho}_E) = 0$, and that, up to
leading-order in $\frac{1}{\omega_2}$ take the form
\begin{align}
    \qty[\bar{M}_+] = \frac{1}{2}
    \mqty(-1 + \mathcal{O}\qty(\frac{1}{\omega_2}) \\
    - \frac{\left(\kappa_{\Sigma} + 2 i \bestguess{\omega}{1}\right) \left(4 \NDelta^{2} \kappa_{\Delta} + \kappa_{\Sigma} \left(2 \NDelta \kappa_{\Delta} + 2 i \NDelta \kappa_{\Sigma} - 2 \NDelta \bestguess{\omega}{1} + 2 \kappa_{\Delta} \kappa_{\Sigma} + 2 i \kappa_{\Delta} \bestguess{\omega}{1} - \kappa_{\Sigma}^{2} - 3 i \kappa_{\Sigma} \bestguess{\omega}{1} + 2 \bestguess{\omega}{1}^{2}\right)\right)}{2 \kappa_{\Sigma} \omega_2^{2} \left(i \kappa_{\Sigma} - 2 \bestguess{\omega}{1}\right)} + \mathcal{O}\qty(\frac{1}{\omega_2^3})\\
    \frac{- \NDelta + i \kappa_{\Delta} - \frac{i \kappa_{\Sigma}}{2} + \bestguess{\omega}{1}}{\omega_2} + \mathcal{O}\qty(\frac{1}{\omega_2^2})\\
    \frac{\kappa_{\Delta}}{\omega_2} + \mathcal{O}\qty(\frac{1}{\omega_2^2})),\\
    \qty[\bar{M}_-] = \frac{1}{2}
    \mqty(-1 + \mathcal{O}\qty(\frac{1}{\omega_2})\\
    - \frac{\left(\kappa_{\Sigma} - 2 i \bestguess{\omega}{1}\right) \left(4 \NDelta^{2} \kappa_{\Delta} + \kappa_{\Sigma} \left(2 \NDelta \kappa_{\Delta} + 2 i \NDelta \kappa_{\Sigma} + 2 \NDelta \bestguess{\omega}{1} + 2 \kappa_{\Delta} \kappa_{\Sigma} - 2 i \kappa_{\Delta} \bestguess{\omega}{1} + \kappa_{\Sigma}^{2} - 3 i \kappa_{\Sigma} \bestguess{\omega}{1} - 2 \bestguess{\omega}{1}^{2}\right)\right)}{2 \kappa_{\Sigma} \omega_2^{2} \left(i \kappa_{\Sigma} + 2 \bestguess{\omega}{1}\right)} + \mathcal{O}\qty(\frac{1}{\omega_2^3})\\
    \frac{- \NDelta + i \kappa_{\Delta} + \frac{i \kappa_{\Sigma}}{2} + \bestguess{\omega}{1}}{\omega_2} + \mathcal{O}\qty(\frac{1}{\omega_2^2})\\
    \frac{\kappa_{\Delta}}{\omega_2} + \mathcal{O}\qty(\frac{1}{\omega_2^2})).
\end{align}
Hence we can write 
\begin{equation}
    \bar{M}_+ = \bar{M}_- = - \frac{\sigma_x}{2} + \mathcal{O}\qty(\frac{1}{\omega_2}).
\end{equation}

Putting all this together, we can write
\begin{align*}
    \varepsilon (1 - \mathcal{R}) \tilde{\mathcal{K}}_1(\rho_s)
    =& \qty(\frac{i g}{\kappa_\Sigma + i \bestguess{\omega}{1}} e^{i \bestguess{\omega}{1} t} T_- \otimes \bar{M}_+) (\rho_s \otimes \bar{\rho}_E) + (\rho_s \otimes \bar{\rho}_E) {\qty(\frac{i g}{\kappa_\Sigma + i \bestguess{\omega}{1}} e^{i \bestguess{\omega}{1} t} T_- \otimes \bar{M}_+)}^\dag\\
    +& \qty(\frac{i g}{\kappa_\Sigma - i \bestguess{\omega}{1}} e^{- i \bestguess{\omega}{1} t} T_+ \otimes \bar{M}_-) (\rho_s \otimes \bar{\rho}_E) + (\rho_s \otimes \bar{\rho}_E) {\qty(\frac{i g}{\kappa_\Sigma - i \bestguess{\omega}{1}} e^{- i \bestguess{\omega}{1} t} T_+ \otimes \bar{M}_-)}^\dag.
\end{align*}

For the second-order reduced dynamics,~\eqref{eq:general_red_dyn} for $k = 2$ gives
\begin{align}
    &\mathcal{K}_0 \mathcal{L}_{s,2}(\rho_s) = \mathcal{L}_{s,2}(\rho_s) \otimes \bar{\rho}_E\\
    =& \mathcal{R} \overline{\mathcal{L}_1 \mathcal{K}_1}(\rho_s) = \Tr_E(\overline{\mathcal{L}_1 \mathcal{K}_1}(\rho_s)) \otimes \bar{\rho}_E,
\end{align}
so 
\begin{equation}
    \mathcal{L}_{s,2}(\rho_s) = \Tr_E(\bar{\mathcal{L}}_1 (1 - \mathcal{R}) \bar{\mathcal{K}}_1(\rho_s))
                              +\Tr_E\qty(\overline{\tilde{\mathcal{L}}_1 \mathcal{R} \tilde{\mathcal{K}}_1(\rho_s)})
                              +\Tr_E\qty(\overline{\tilde{\mathcal{L}}_1 (1 - \mathcal{R}) \tilde{\mathcal{K}}_1(\rho_s)})\label{eq:ls_2_terms}.
\end{equation}

It is straightforward to verify that
\begin{equation}
    \Tr_E(\bar{\mathcal{L}}_1 (1 - \mathcal{R}) \bar{\mathcal{K}}_1(\rho_s)) = \frac{\bestguess{\omega}{1}^2}{\omega_2} \Tr(\sigma_z \bar{M}_z \bar{\rho}_E) \qty(T_z^2 \rho_s - T_z \rho_s T_z)
                                                                             + \frac{\bestguess{\omega}{1}^2}{\omega_2} \Tr(\sigma_z \bar{\rho}_E \bar{M}_z^\dag) \qty(\rho_s T_z^2 - T_z \rho_s T_z),
\end{equation}
and using \[\Tr(\sigma_z \bar{M}_z \bar{\rho}_E) = - \frac{4 \NDelta^{2} + \kappa_{\Sigma}^{2}}{2 \kappa_{\Sigma} \omega_2}
+ i \frac{\kappa_{\Delta} \left(2 \NDelta - \kappa_{\Sigma}\right)}{\kappa_{\Sigma} \omega_2},\]
we obtain that
\begin{equation}
    \Tr_E(\bar{\mathcal{L}}_1 (1 - \mathcal{R}) \bar{\mathcal{K}}_1(\rho_s)) = \bestguess{\omega}{1}^2 \frac{\left(4 \NDelta^{2} + \kappa_{\Sigma}^{2}\right)}{\kappa_{\Sigma} \omega_2^{2}} \mathcal{D}_{T_z}(\rho_s)
                - i \bestguess{\omega}{1}^2 \frac{\kappa_{\Delta} \left(- 2 \NDelta + \kappa_{\Sigma}\right)}{\kappa_{\Sigma} \omega_2^{2}} \comm{T_z^2}{\rho_s}.
\end{equation}

For the second term in equation~\eqref{eq:ls_2_terms} we obtain
\begin{equation}
    \Tr_E\qty(\overline{\tilde{\mathcal{L}}_1 \mathcal{R} \tilde{\mathcal{K}}_1(\rho_s)}) = 
    i \frac{\xi_\infty}{2} \bestguess{\omega}{1} \Tr(\sigma_+ \bar{\rho}_E) \comm{T_-}{\comm{T_+}{\rho_s}} - i \frac{\xi_\infty^*}{2} \bestguess{\omega}{1}\Tr(\sigma_- \bar{\rho}_E) \comm{T_+}{\comm{T_-}{\rho_s}},
\end{equation}
and using $\Tr(\sigma_+ \bar{\rho}_E) = \frac{\xi_\infty^*}{2}$ we obtain that
\begin{equation}
    \Tr_E\qty(\overline{\tilde{\mathcal{L}}_1 \mathcal{R} \tilde{\mathcal{K}}_1(\rho_s)}) = 
    - i \frac{\xi_\infty^* \xi_\infty}{4} \bestguess{\omega}{1} \comm{\comm{T_+}{T_-}}{\rho_s}.
\end{equation}

For the third term in equation~\eqref{eq:ls_2_terms} we obtain
\begin{align}
    \Tr_E\qty(\overline{\tilde{\mathcal{L}}_1 (1 - \mathcal{R}) \tilde{\mathcal{K}}_1(\rho_s)})
     &= a_+ \qty(T_- T_+ \rho_s - T_+ \rho_s T_-) - a_-^* \qty(T_- \rho_s T_+ - \rho_s T_+ T_-)\\
     &+ a_- \qty(T_+ T_- \rho_s - T_- \rho_s T_+) - a_+^* \qty(T_+ \rho_s T_- - \rho_s T_- T_+),
\end{align}
with 
\begin{align}
    a_+ &= \frac{\Tr(\sigma_+ \bar{M}_- \bar{\rho}_E)}{\kappa_\Sigma - i \bestguess{\omega}{1}} \bestguess{\omega}{1}^2,\\
    a_- &= \frac{\Tr(\sigma_- \bar{M}_+ \bar{\rho}_E)}{\kappa_\Sigma + i \bestguess{\omega}{1}} \bestguess{\omega}{1}^2.
\end{align}
Retaining the leading-order terms in $\frac{1}{\omega_2}$ for $a_+$ and $a_-$, we readily obtain
\begin{equation}
    \Tr_E\qty(\overline{\tilde{\mathcal{L}}_1 (1 - \mathcal{R}) \tilde{\mathcal{K}}_1(\rho_s)})
    = \frac{\kappa_{\Sigma} \bestguess{\omega}{1}^{2}}{\kappa_{\Sigma}^{2} + 4 \bestguess{\omega}{1}^{2}}
    \qty(\mathcal{D}_{T_-} + \mathcal{D}_{T_-}) + i \frac{\bestguess{\omega}{1}^3}{\kappa_\Sigma^2 + 4 \bestguess{\omega}{1}^2} \comm{\comm{T_+}{T_-}}{\rho_s}.
\end{equation}

Putting all of the calculations of this section together, we obtain the following second-order reduced model.
For the slow dynamics we obtain an explicit Lindbladian model
\begin{align}
    \mathcal{L}_{s,g}(\rho_s) &=  - i  \comm{\omega_{s,z,1} T_z + \omega_{s,z,2} T_z^2 + \omega_{s,c} \comm{T_+}{T_-} + \omega_{s,a} (T_+ T_- + T_- T_+)}{\rho_s}\nonumber \\
    &+ \kappa_{s,z} \mathcal{D}_{T_z}(\rho_s) + \kappa_{s,\pm} \qty(\mathcal{D}_{T_-} + \mathcal{D}_{T_+})
    + \mathcal{O}\qty(g \varepsilon^2),\label{eq:lindbladian_strong}
\end{align}
with, up to leading order in $\frac{1}{\omega_2}$,
\begin{align}
    \omega_{s,z,1}   &= - \frac{\kappa_{\Delta} g \left(4 \NDelta^{2} + \kappa_{\Sigma}^{2}\right)}{2 \kappa_{\Sigma} \omega_2^{2}},\\
    \omega_{s,z,2}   &= \frac{\kappa_{\Delta} g^{2} \left(- 2 \NDelta + \kappa_{\Sigma}\right)}{2 \kappa_{\Sigma} \omega_2^{2}},\\
    \omega_{s,c} &= \frac{g^{2} \bestguess{\omega}{1}}{\kappa_{\Sigma}^{2} + 4 \bestguess{\omega}{1}^{2}},\\
    \kappa_{s,z}   &= \frac{g^{2} \left(4 \NDelta^{2} + \kappa_{\Sigma}^{2}\right)}{\kappa_{\Sigma} \omega_2^{2}},\\
    \kappa_{s,\pm} &= \frac{\kappa_{\Sigma} g^{2}}{\kappa_{\Sigma}^{2} + 4 \bestguess{\omega}{1}^{2}} \, .
\end{align}
For the embedding of the slow subspace we obtain a completely positive map up to second order terms:
\begin{equation}\label{eq:kraus_map_strong}
    \mathcal{K}_{s,g}(\rho_s) = K_g (\rho_s \otimes \bar{\rho}_E) K_g^\dag + \mathcal{O}(\varepsilon^2),
\end{equation}
with, up to leading-order in $\frac{1}{\omega_2}$ for every term,
\begin{align}
    K_g :=  1 - i \frac{\kappa_\Delta}{\kappa_\Sigma} \frac{g}{\omega_2} H_s \otimes \mathbb{1}_E
    + i \frac{g}{\omega_2} T_z \otimes \sigma_y - i \frac{2 g \NDelta}{\kappa_\Sigma \omega_2} T_z \otimes \sigma_x
    - i \frac{g}{\sqrt{\kappa_\Sigma^2 + 4 \bestguess{\omega}{1}^2}} H_{s,\pm} \otimes \sigma_x - 2 \frac{\kappa_\Delta}{\kappa_\Sigma} \frac{g}{\omega_2} T_z \otimes \sigma_x,
\end{align}
and we have defined
\begin{align}
    H_s       &= - \frac{\kappa_\Sigma + 2 i \NDelta}{2 \bestguess{\omega}{1}} e^{i \bestguess{\omega}{1} t} T_- -
    \frac{\kappa_\Sigma - 2 i \NDelta}{2 \bestguess{\omega}{1}} e^{- i \bestguess{\omega}{1} t} T_+,\\
    H_{s,\pm} &= \frac{(\kappa_\Sigma - 2 i \bestguess{\omega}{1}) e^{i \bestguess{\omega}{1} t} T_- + (\kappa_\Sigma + 2 i \bestguess{\omega}{1}) e^{- i \bestguess{\omega}{1} t} T_+}{\sqrt{\kappa_\Sigma^2 + 4 \bestguess{\omega}{1}^2}}.
\end{align}

We here reported the leading-order of all different terms in $\frac{1}{\omega_2}$, hence approximating the exact expression of
$\mathcal{K}_1$ as defined in~\eqref{eq:K_1_a},~\eqref{eq:K_1_b} and~\eqref{eq:K_1_c} in the limit of large $\omega_2$.
When using the exact expressions, it is straightforward to show that $\Tr(\mathcal{K}_1(\rho_s)) = 0$,
since~\eqref{eq:K_1_a},~\eqref{eq:exact_A_z_bar} and~\eqref{eq:M_pm_bar_exact}, are traceless expressions.
We then obtain that $\Tr_E(\mathcal{K}_{s,g}(\rho_s)) =   \Tr(\mathcal{K}_0(\rho_s)) + \varepsilon
\Tr(\mathcal{K}_1(\rho_s)) + \mathcal{O}(\varepsilon^2) = \Tr(\rho_s) + \mathcal{O}(\varepsilon^2) $,
and thus up to order $\varepsilon^2$, $\mathcal{K}_{s,g}$ is also trace-preserving, and hence CPTP.

\subsubsection*{Discussion of Hamiltonian terms}

The exact first-order slow dynamics $\mathcal{L}_{s,1}$ is given by the Hamiltonian
$\omega_{s,z,1} T_z$, with
$$\omega_{s,z,1} = - \frac{\kappa_- - \kappa_+}{\kappa_- + \kappa_+} \frac{4 \NDelta^{2} + {(\kappa_- + \kappa_+)}^{2}}{4 \NDelta^{2} + {(\kappa_- + \kappa_+)}^{2} + 2 \omega_{2}^{2}}.$$
Regarding the system parameters, we can see that this contribution is largest for a TLS coupled to a cold bath, and disappears in the limit of a hot bath, where $\kappa_- = \kappa_+$.
Since the imperfect detuning $\NDelta$ appears, we cannot expect to have exact knowledge of $\omega_{s,z,1}$.
However, if $\NDelta$ can be assumed constant, then the term can be calibrated experimentally and corrected for.
Remark that such a Lamb-shift type Hamiltonian is present in the absence of driving as well, and only the frequency is altered through the driving.
Regarding the QDD control, the term goes like $\sim \frac{1}{\omega_2^2}$ for large $\omega_2$, and hence it is suppressed for strong driving, although this was not explicitly part of our goal.

For the first Hamiltonian term at second order, we obtain
$$\omega_{s,z,2} = - \frac{16 \NDelta \kappa_{\Delta} \omega_{2}^{2} g^{2}}{\kappa_{\Sigma}
\left(4 \NDelta^{2} + \kappa_{\Sigma}^{2} + 2 \omega_{2}^{2}\right)^{2}}.$$
We obtain the same conclusion as for $\omega_{s,z,1}$, namely $\omega_{s,z,2}$ is minimal for a hot bath,
and decreases like $\frac{1}{\omega_2^2}$ under the QDD controls.
The full expressions for the remaining two Hamiltonian terms are more involved.
Directly focussing in the regime for large $\omega_2$, we obtain a Hamiltonian $\omega_{s,c} \comm{T_+}{T_-}$ with
\begin{equation}
    \omega_{s,c} = - \frac{\bestguess{\omega}{1}}{\kappa_{\Sigma}^{2} + 4 \bestguess{\omega}{1}^{2}} + \mathcal{O}\qty(\frac{1}{\omega_2^2}),
\end{equation}
as above, but also an additional Hamiltonian $\omega_{s,a} \qty(T_+ T_- + T_- T_+)$, with
\begin{equation}
    \omega_{s,a} = - \frac{\kappa_{\Delta} \bestguess{\omega}{1} \left(4 \NDelta \bestguess{\omega}{1} + \kappa_{\Sigma}^{2}\right)}{2 \kappa_{\Sigma} \omega_{2}^{2} \left(\kappa_{\Sigma}^{2} + 4 \bestguess{\omega}{1}^{2}\right)}
    + \mathcal{O}\qty(\frac{1}{\omega_2^4}).
\end{equation}
We can again see that the QDD controls suppress these Hamiltonian contributions asymptotically for large $\bestguess{\omega}{1}$ and $\omega_2$.

\subsection{Case of ultra-strong driving}\label{sec:calculations_ultra_strong}

We recapitulate the full model here, the dissipation model being given in~\eqref{eq:dissipator_bis}:
\begin{align}
    \dot{\rho} &= - i \frac{\Lambda}{2} \comm{\mathbb{1}_T \otimes \sigma_{\alpha x}}{\rho}
    + \kappa_{\alpha x} \mathcal{D}_{\mathbb{1}_T \otimes \sigma_{\alpha x}}(\rho)
    + \kappa_{\alpha-} \mathcal{D}_{\mathbb{1}_T \otimes \sigma_{\alpha-}}(\rho) + \kappa_{\alpha+} \mathcal{D}_{\mathbb{1}_T \otimes
    \sigma_{\alpha+}}(\rho)\nonumber\\
    &- ig \comm{T_z \otimes \sigma_z + e^{i \bestguess{\omega}{1} t} T_- \otimes \sigma_+ + e^{- i \bestguess{\omega}{1} t} T_+ \otimes \sigma_-}{\rho}.\nonumber
\end{align}

In the notation of appendix~\ref{sec:ae_floquet}, we can similarly write $\varepsilon = \frac{g}{\bestguess{\omega}{1}}$,
\begin{equation}
    \mathcal{L}_0 = - i \frac{\Lambda}{2} \comm{\mathbb{1}_T \otimes \sigma_{\alpha x}}{\cdot}
    + \kappa_{\alpha x} \mathcal{D}_{\mathbb{1}_T \otimes \sigma_{\alpha x}}
    + \kappa_{\alpha-} \mathcal{D}_{\mathbb{1}_T \otimes \sigma_{\alpha-}} + \kappa_{\alpha+} \mathcal{D}_{\mathbb{1}_T \otimes \sigma_{\alpha+}},
\end{equation}
and we still have
\begin{align*}
    \mathcal{L}_1(t) = - i \bestguess{\omega}{1} \comm{T_z \otimes \sigma_z + e^{i \bestguess{\omega}{1} t} T_- \otimes \sigma_+ + e^{- i \bestguess{\omega}{1} t} T_+ \otimes \sigma_-}{\cdot}.
\end{align*}

It is straightforward to verify that the fast dynamics $\mathcal{L}_0$ drives the environment to a unique steady state
\begin{equation}
    \bar{\rho}_E = \frac{\mathbb{1}_E +  x_{\alpha, \infty} \sigma_{\alpha x} }{2},
\end{equation}
with
\begin{align}
    x_{\alpha, \infty} &= \frac{\kappa_{\alpha_+} - \kappa_{\alpha_-}}{\kappa_{\alpha_+} + \kappa_{\alpha_-}}.
\end{align}
For the following it is instructive to define
\begin{align}
    \kappa_{\alpha_\Sigma} &:= \kappa_- + \kappa_+,\\
    \kappa_{\alpha_\Delta} &:= \kappa_- - \kappa_+,
\end{align}
so $x_{\alpha, \infty} = - \frac{\kappa_{\alpha_\Delta}}{\kappa_{\alpha_\Sigma}}$.
Remark that the steady-state is independent of the driving amplitude $\Lambda$.

For the projector $\mathcal{R}$ we have \[\mathcal{R}(X_\textrm{TE}) = \Tr_E(X_\textrm{TE}) \otimes \bar{\rho}_E \qq*{,} \forall X_\textrm{TE}.\]
Equation~\eqref{eq:first_order_red} yields the following expression for the first-order reduced dynamics:
\begin{align*}
    \varepsilon \mathcal{L}_{s,1}(\rho_s) \otimes \bar{\rho}_E &= \varepsilon \mathcal{R}(\overline{\mathcal{L}_1(\rho_s \otimes \bar{\rho}_E)}) \nonumber\\
                                                   &= - i g \Tr_E(\comm{T_z \otimes \sigma_z}{\rho_s \otimes \bar{\rho}_E}) \otimes \bar{\rho}_E \nonumber\\
                                                   &= - i g x_{\alpha, \infty} \sin(\alpha) \comm{T_z}{\rho_s} \otimes \bar{\rho}_E,
\end{align*}
readily yielding
\begin{equation}
    \varepsilon \mathcal{L}_{s,1}(\rho_s) = - i g x_{\alpha, \infty} \sin(\alpha) \comm{T_z}{\rho_s}.
\end{equation}

Equation~\eqref{eq:primitivation} in turn yields
\begin{align}
    \varepsilon \mathcal{R} \mathcal{K}_1(\rho_s) &= \varepsilon \mathcal{R} \partial_t^{-1} \tilde{\mathcal{L}}_1 \mathcal{K}_0(\rho_s)\nonumber\\
                              &= - \frac{g}{\bestguess{\omega}{1}} \Tr_E(\comm{e^{i \bestguess{\omega}{1} t} T_- \otimes \sigma_+ - e^{- i \bestguess{\omega}{1} t} T_+ \otimes \sigma_-}{\rho_s \otimes \bar{\rho}_E}) \otimes \bar{\rho}_E\nonumber\\
                              &= - i x_{\alpha, \infty} \cos(\alpha) \frac{g}{2 \bestguess{\omega}{1}}\comm{i e^{i \bestguess{\omega}{1} t} T_- - i e^{- i \bestguess{\omega}{1} t} T_+}{\rho_s} \otimes \bar{\rho}_E,
\end{align}
where we have put the integration constant to zero as a gauge choice.
Equation \eqref{eq:first_order_pseudo_inverse} yields a second part of $\mathcal{K}_1$:
\begin{align}
    \varepsilon \mathcal{L}_0(1 - \mathcal{R}) \bar{\mathcal{K}}_1(\rho_s) &= - \varepsilon (1 - \mathcal{R}) \bar{\mathcal{L}}_1 \mathcal{K}_0(\rho_s)\nonumber\\
                                                                &= i g \comm{T_z \otimes \sigma_z}{\rho_s \otimes \bar{\rho}_E}
                                                                - i g \Tr_E(\comm{T_z \otimes \sigma_z}{\rho_s \otimes \bar{\rho}_E}) \otimes \bar\rho_E \nonumber\\
                                                                &= i g \qty(T_z \rho_s \otimes \bar{\sigma}_z \bar{\rho}_E - \rho_s T_z \otimes \bar{\rho}_E \bar{\sigma}_z),
\end{align}
with $\bar{\sigma}_z = \sigma_z - \Tr(\sigma_z \bar{\rho}_E)\mathbb{1}_E = \sigma_z - x_{\alpha, \infty} \sin{\left(\alpha \right)} \mathbb{1}_E$.
Remark that taking the partial trace over E of the right-hand side gives zero, since $\Tr(\bar{\sigma}_z \bar{\rho}_E) = 0$.
Hence $\mathcal{L}_0$ can be inverted to obtain, formally, 
\begin{equation*}
    \varepsilon (1 - \mathcal{R}) \bar{\mathcal{K}}_1(\rho_s) = i g \qty(T_z \rho_s \otimes \mathcal{L}_0^{-1}(\bar{\sigma}_z \bar{\rho}_E) - \rho_s T_z \otimes \mathcal{L}_0^{-1}(\bar{\rho}_E \bar{\sigma}_z)).
\end{equation*}
For this inversion we again use matrix representations in the Pauli basis.\\
In a rotated Pauli basis ($ \cos(\alpha) \sigma_z - \sin(\alpha) \sigma_x, \; \sigma_y,  \;\sigma_{\alpha x}$),
we obtain the following matrix representation for $\mathcal{L}_0$:
\begin{equation}
    \qty[\mathcal{L}_0] = \mqty(- \frac{\kappa_{\alpha_\Sigma}}{2} - 2 \kappa_{\alpha x} & - \Lambda & 0\\\Lambda & - \frac{\kappa_{\alpha_\Sigma}}{2} - 2 \kappa_{\alpha x} & 0\\0 & 0 & - \kappa_{\alpha_\Sigma}),
\end{equation}
with
\begin{equation}
    \mathrm{det}\qty[\mathcal{L}_0] = - \frac{\kappa_{\alpha_\Sigma} \left(4 \Lambda^{2} + \kappa_{\alpha_\Sigma}^{2} + 8 \kappa_{\alpha_\Sigma} \kappa_{\alpha x} + 16 \kappa_{\alpha x}^{2}\right)}{4}.
\end{equation}

For its inverse we obtain
\begin{equation*}
    \qty[\mathcal{L}_0^{-1}] = \frac{1}{\mathrm{det}\qty[\mathcal{L}_0]}\mqty(\frac{\kappa_{\alpha_\Sigma} \left(\kappa_{\alpha_\Sigma} + 4 \kappa_{\alpha x}\right)}{2} & - \Lambda \kappa_{\alpha_\Sigma} & 0\\\Lambda \kappa_{\alpha_\Sigma} & \frac{\kappa_{\alpha_\Sigma} \left(\kappa_{\alpha_\Sigma} + 4 \kappa_{\alpha x}\right)}{2} & 0\\0 & 0 & \Lambda^{2} + \frac{\kappa_{\alpha_\Sigma}^{2}}{4} + 2 \kappa_{\alpha_\Sigma} \kappa_{\alpha x} + 4 \kappa_{\alpha x}^{2}).
\end{equation*}
In turn, $\bar{\sigma}_z \bar{\rho}_E$ takes the following vector representation in the Pauli basis:
\begin{equation}
    \qty[\bar{\sigma}_z \bar{\rho}_E] = \mqty(- \frac{i \kappa_{\alpha_\Delta} \cos{\left(\alpha \right)}}{2 \kappa_{\alpha_\Sigma}}\\\frac{\cos{\left(\alpha \right)}}{2}\\\frac{\left(- \kappa_{\alpha_\Delta}^{2} + \kappa_{\alpha_\Sigma}^{2}\right) \sin{\left(\alpha \right)}}{2 \kappa_{\alpha_\Sigma}^{2}}\\0).
\end{equation}
Straightforward calculations then give
\begin{equation}
    \qty[\mathcal{L}_0^{-1}(\bar{\sigma}_z \bar{\rho}_E)] =
    \mqty(\frac{\left(2 \Lambda \kappa_{\alpha_\Sigma} + i \kappa_{\alpha_\Delta} \left(\kappa_{\alpha_\Sigma} + 4 \kappa_{\alpha x}\right)\right) \cos{\left(\alpha \right)}}{\kappa_{\alpha_\Sigma} \left(4 \Lambda^{2} + \kappa_{\alpha_\Sigma}^{2} + 8 \kappa_{\alpha_\Sigma} \kappa_{\alpha x} + 16 \kappa_{\alpha x}^{2}\right)}\\\frac{\left(2 i \Lambda \kappa_{\alpha_\Delta} - \kappa_{\alpha_\Sigma} \left(\kappa_{\alpha_\Sigma} + 4 \kappa_{\alpha x}\right)\right) \cos{\left(\alpha \right)}}{\kappa_{\alpha_\Sigma} \left(4 \Lambda^{2} + \kappa_{\alpha_\Sigma}^{2} + 8 \kappa_{\alpha_\Sigma} \kappa_{\alpha x} + 16 \kappa_{\alpha x}^{2}\right)}\\\frac{\left(\kappa_{\alpha_\Delta} - \kappa_{\alpha_\Sigma}\right) \left(\kappa_{\alpha_\Delta} + \kappa_{\alpha_\Sigma}\right) \sin{\left(\alpha \right)}}{2 \kappa_{\alpha_\Sigma}^{3}}).
\end{equation}

Focussing on the leading-order in $\frac{1}{\omega_2}$ yields the following:
\begin{equation}
    \qty[\mathcal{L}_0^{-1}(\bar{\sigma}_z \bar{\rho}_E)] = 
    \frac{1}{2 \omega_2} \mqty(1 \\ \frac{i \kappa_{\alpha_\Delta}}{\kappa_{\alpha_\Sigma}}\\\frac{\NDelta \left(\kappa_{\alpha_\Delta}^{2} - \kappa_{\alpha_\Sigma}^{2}\right)}{\kappa_{\alpha_\Sigma}^{3}}) + \mathcal{O}\qty(\frac{1}{\omega_2^2}),
\end{equation}
and further
\begin{equation*}
    \varepsilon (1 - \mathcal{R}) \bar{\mathcal{K}}_1(\rho_s) = i \frac{g}{\omega_2} \comm{T_z \otimes \bar{M}_z}{\rho_s \otimes \bar{\rho}_E},
\end{equation*}
with
\begin{equation}
    \qty[\bar{M}_z] = \mqty(1\\0\\- \frac{\NDelta}{\kappa_{\alpha_\Sigma}}\\- \frac{\NDelta \kappa_{\alpha_\Delta}}{\kappa_{\alpha_\Sigma}^{2}})
    + \mathcal{O}\qty(\frac{1}{\omega_2}).
\end{equation}

For the last part of $\mathcal{K}_1$, consider~\eqref{eq:new_pseudo_inverse}:
\begin{equation}\label{eq:laatste_K_ultra}
    \varepsilon (\mathcal{L}_0 - \partial_t) (1 - \mathcal{R}) \tilde{\mathcal{K}}_1(\rho_s) = i g (1 - \mathcal{R})
    \qty(\comm{e^{i \bestguess{\omega}{1} t} T_- \otimes \sigma_+ + e^{- i \bestguess{\omega}{1} t} T_+ \otimes \sigma_-}{\rho_s \otimes \bar{\rho}_E}).
\end{equation}
Introducing
\begin{align}
    \bar{\sigma}_+ &:= \sigma_+ - \Tr(\sigma_+ \bar{\rho}_E) \mathbb{1}_E = \sigma_+ + \frac{\kappa_{\alpha_\Delta} \cos{\left(\alpha \right)}}{2 \kappa_{\alpha_\Sigma}} \mathbb{1}_E,\nonumber \\
    \bar{\sigma}_- &:= \sigma_- - \Tr(\sigma_- \bar{\rho}_E) \mathbb{1}_E = \sigma_- + \frac{\kappa_{\alpha_\Delta} \cos{\left(\alpha \right)}}{2 \kappa_{\alpha_\Sigma}} \mathbb{1}_E,\nonumber
\end{align}
we can write the right-hand side of~\eqref{eq:laatste_K_ultra} as
\begin{equation*}
    i g \comm{e^{i \bestguess{\omega}{1} t} T_- \otimes \bar{\sigma}_+ + e^{- i \bestguess{\omega}{1} t} T_+ \otimes \bar{\sigma}_-}{\rho_s \otimes \bar{\rho}_E}
    =i g e^{i \bestguess{\omega}{1} t} T_- \rho_s \otimes \bar{\sigma}_+ \bar{\rho}_E + i g e^{- i \bestguess{\omega}{1} t} T_+ \rho_s \otimes \bar{\sigma}_- \bar{\rho}_E + \textrm{h.c.}
\end{equation*}
At this point we can split $(1 - \mathcal{R}) \tilde{\mathcal{K}}_1$ up into two parts:
\begin{align}
    \varepsilon (1 - \mathcal{R}) \tilde{\mathcal{K}}_1(\rho_s) &= i g e^{i \bestguess{\omega}{1} t} {\qty(\mathcal{L}_0 - i \bestguess{\omega}{1})}^{-1}\qty(T_- \rho_s \otimes \bar{\sigma}_+ \bar{\rho}_E)
                  + i g e^{-i \bestguess{\omega}{1} t} {\qty(\mathcal{L}_0 + i \bestguess{\omega}{1})}^{-1}\qty(T_+ \rho_s \otimes \bar{\sigma}_- \bar{\rho}_E) + \textrm{h.c.}  \nonumber\\
                  &= i g e^{i \bestguess{\omega}{1} t} T_- \rho_s \otimes {\qty(\mathcal{L}_0 - i \bestguess{\omega}{1})}^{-1}\qty(\bar{\sigma}_+ \bar{\rho}_E)
                  + i g e^{-i \bestguess{\omega}{1} t} T_+ \rho_s \otimes {\qty(\mathcal{L}_0 + i \bestguess{\omega}{1})}^{-1}\qty(\bar{\sigma}_- \bar{\rho}_E) + \textrm{h.c.}  \nonumber
\end{align}
We obtain the following matrix representations:
\begin{align*}
    \qty[\bar{\sigma}_+ \bar{\rho}_E] = \mqty(\frac{i \left(\kappa_{\alpha_\Delta} \sin{\left(\alpha \right)} + \kappa_{\alpha_\Sigma}\right)}{4 \kappa_{\alpha_\Sigma}}\\- \frac{\kappa_{\alpha_\Delta} + \kappa_{\alpha_\Sigma} \sin{\left(\alpha \right)}}{4 \kappa_{\alpha_\Sigma}}\\\frac{\left(- \kappa_{\alpha_\Delta}^{2} + \kappa_{\alpha_\Sigma}^{2}\right) \cos{\left(\alpha \right)}}{4 \kappa_{\alpha_\Sigma}^{2}}),\\
    \qty[\bar{\sigma}_- \bar{\rho}_E] = \mqty(\frac{i \left(\kappa_{\alpha_\Delta} \sin{\left(\alpha \right)} - \kappa_{\alpha_\Sigma}\right)}{4 \kappa_{\alpha_\Sigma}}\\\frac{\kappa_{\alpha_\Delta} - \kappa_{\alpha_\Sigma} \sin{\left(\alpha \right)}}{4 \kappa_{\alpha_\Sigma}}\\\frac{\left(- \kappa_{\alpha_\Delta}^{2} + \kappa_{\alpha_\Sigma}^{2}\right) \cos{\left(\alpha \right)}}{4 \kappa_{\alpha_\Sigma}^{2}}),
\end{align*}
and

    \begin{equation*}
        \hspace{-0.7cm} \qty[{\qty(\mathcal{L}_0 \mp i \bestguess{\omega}{1})}^{-1}] =
        \mqty(\frac{- 2 \kappa_{\alpha_\Sigma} - 8 \kappa_{\alpha x} \mp 4 i \bestguess{\omega}{1}}{4 \Lambda^{2} + \kappa_{\alpha_\Sigma}^{2} + 8 \kappa_{\alpha_\Sigma} \kappa_{\alpha x} \pm 4 i \kappa_{\alpha_\Sigma} \bestguess{\omega}{1} + 16 \kappa_{\alpha x}^{2} \pm 16 i \kappa_{\alpha x} \bestguess{\omega}{1} - 4 \bestguess{\omega}{1}^{2}} & \frac{4 \Lambda}{4 \Lambda^{2} + \kappa_{\alpha_\Sigma}^{2} + 8 \kappa_{\alpha_\Sigma} \kappa_{\alpha x} \pm 4 i \kappa_{\alpha_\Sigma} \bestguess{\omega}{1} + 16 \kappa_{\alpha x}^{2} \pm 16 i \kappa_{\alpha x} \bestguess{\omega}{1} - 4 \bestguess{\omega}{1}^{2}} & 0\\- \frac{4 \Lambda}{4 \Lambda^{2} + \kappa_{\alpha_\Sigma}^{2} + 8 \kappa_{\alpha_\Sigma} \kappa_{\alpha x} \pm 4 i \kappa_{\alpha_\Sigma} \bestguess{\omega}{1} + 16 \kappa_{\alpha x}^{2} \pm 16 i \kappa_{\alpha x} \bestguess{\omega}{1} - 4 \bestguess{\omega}{1}^{2}} & \frac{- 2 \kappa_{\alpha_\Sigma} - 8 \kappa_{\alpha x} \mp 4 i \bestguess{\omega}{1}}{4 \Lambda^{2} + \kappa_{\alpha_\Sigma}^{2} + 8 \kappa_{\alpha_\Sigma} \kappa_{\alpha x} \pm 4 i \kappa_{\alpha_\Sigma} \bestguess{\omega}{1} + 16 \kappa_{\alpha x}^{2} \pm 16 i \kappa_{\alpha x} \bestguess{\omega}{1} - 4 \bestguess{\omega}{1}^{2}} & 0\\0 & 0 & \frac{1}{- \kappa_{\alpha_\Sigma} \mp i \bestguess{\omega}{1}}).
    \end{equation*}

Again focussing on the leading-order in $\frac{1}{\omega_2}$, putting $\cos(\alpha)$  to $1$, and using $\Lambda \sin(\alpha) = \NDelta$, we obtain
\begin{align}
    {\qty(\mathcal{L}_0 - i \bestguess{\omega}{1})}^{-1}\qty(\bar{\sigma}_+ \bar{\rho}_E) = - \qty( \frac{1}{2 \omega_2}
    B^\dag + \frac{\cos(\alpha)}{2\qty(\kappa_{\alpha_\Sigma} + i \bestguess{\omega}{1})}
    \qty(\frac{\kappa_{\alpha_\Delta}}{\kappa_{\alpha_\Sigma}} \mathbb{1}_E + \sigma_{\alpha x})) \bar{\rho}_E,\\
    {\qty(\mathcal{L}_0 + i \bestguess{\omega}{1})}^{-1}\qty(\bar{\sigma}_- \bar{\rho}_E) = - \qty( \frac{1}{2 \omega_2}
    B + \frac{\cos(\alpha)}{2\qty(\kappa_{\alpha_\Sigma} - i \bestguess{\omega}{1})}
    \qty(\frac{\kappa_{\alpha_\Delta}}{\kappa_{\alpha_\Sigma}} \mathbb{1}_E + \sigma_{\alpha x})) \bar{\rho}_E,
\end{align}
with 
\begin{align}
    \qty[B] = \mqty(\frac{\NDelta - \frac{i \kappa_{\alpha_\Sigma}}{2} - 2 i \kappa_{\alpha x} - \bestguess{\omega}{1}}{\omega_2} + \mathcal{O}\qty(\frac{1}{\omega_2^2})\\- i + \mathcal{O}\qty(\frac{1}{\omega_2})\\0\\0).
\end{align}
Putting all this together, we can write
\begin{align*}
    \varepsilon (1 - \mathcal{R}) \tilde{\mathcal{K}}_1(\rho_s)
                    =& - i \frac{g}{2 \Lambda} \comm{e^{i \bestguess{\omega}{1} t} T_- \otimes B^\dag + e^{- i \bestguess{\omega}{1} t} T_+ \otimes B}{\rho_s \otimes \bar{\rho}_E}\\
                    &- i \frac{g \cos(\alpha)}{2\sqrt{\kappa_\Sigma^2 +
                    \bestguess{\omega}{1}^2}}\comm{\qty(\frac{\kappa_\Sigma - i
                    \bestguess{\omega}{1}}{\sqrt{\kappa_\Sigma^2 + \bestguess{\omega}{1}^2}} e^{i \bestguess{\omega}{1}
                    t} T_- + \frac{\kappa_\Sigma + i \bestguess{\omega}{1}}{\sqrt{\kappa_\Sigma^2 +
                    \bestguess{\omega}{1}^2}} e^{- i \bestguess{\omega}{1} t} T_+) \otimes
                    \qty(\frac{\kappa_{\alpha_\Delta}}{\kappa_{\alpha_\Sigma}} \mathbb{1}_E + \sigma_{\alpha x})}{\rho_s \otimes \bar{\rho}_E}.
\end{align*}

For the second-order reduced dynamics,~\eqref{eq:general_red_dyn} for $k = 2$ gives
\begin{equation}
    \mathcal{K}_0 \mathcal{L}_{s,2}(\rho_s) = \mathcal{L}_{s,2}(\rho_s) \otimes \bar{\rho}_E
    = \mathcal{R} \overline{\mathcal{L}_1 \mathcal{K}_1}(\rho_s) = \Tr_E(\overline{\mathcal{L}_1 \mathcal{K}_1}(\rho_s)) \otimes \bar{\rho}_E
\end{equation}
so 
\begin{equation}
    \mathcal{L}_{s,2}(\rho_s) = \Tr_E(\bar{\mathcal{L}}_1 (1 - \mathcal{R}) \bar{\mathcal{K}}_1(\rho_s))
                              +\Tr_E\qty(\overline{\tilde{\mathcal{L}}_1 \mathcal{R} \tilde{\mathcal{K}}_1(\rho_s)})
                              +\Tr_E\qty(\overline{\tilde{\mathcal{L}}_1 (1 - \mathcal{R})
                              \tilde{\mathcal{K}}_1(\rho_s)})\label{eq:ls_2_terms_ultra}.
\end{equation}

It is straightforward to verify that
\begin{equation}
    \Tr_E(\bar{\mathcal{L}}_1 (1 - \mathcal{R}) \bar{\mathcal{K}}_1(\rho_s)) = \frac{\bestguess{\omega}{1}^2}{\omega_2} \Tr(\sigma_z \bar{M}_z \bar{\rho}_E) \qty(T_z^2 \rho_s - T_z \rho_s T_z)
                                                                             + \frac{\bestguess{\omega}{1}^2}{\omega_2} \Tr(\sigma_z \bar{\rho}_E \bar{M}_z) \qty(\rho_s T_z^2 - T_z \rho_s T_z),
\end{equation}
and using \[\Tr(\sigma_z \bar{M}_z \bar{\rho}_E) =  \frac{i \kappa_{\alpha_\Delta}}{\kappa_{\alpha_\Sigma}} + \frac{\NDelta^{2} \left(\frac{\kappa_{\alpha_\Delta}^{2}}{\kappa_{\alpha_\Sigma}^{2}} - 1\right)}{\kappa_{\alpha_\Sigma} \omega_2}- \frac{\kappa_{\alpha_\Sigma}}{2 \omega_2} - \frac{2 \kappa_{\alpha x}}{\omega_2}+ \mathcal{O}\qty(\frac{1}{\omega_2^2}),\]
we obtain that
\begin{equation}
    \Tr_E(\bar{\mathcal{L}}_1 (1 - \mathcal{R}) \bar{\mathcal{K}}_1(\rho_s)) = \qty( 2 \frac{\NDelta^{2}
    \left(1 - \frac{\kappa_{\alpha_\Delta}^{2}}{\kappa_{\alpha_\Sigma}^{2}}\right)}{\kappa_{\alpha_\Sigma}} +
    \kappa_{\alpha_\Sigma} + 4 \kappa_{\alpha x})
    \frac{\bestguess{\omega}{1}^2}{\omega_2^2} \mathcal{D}_{T_z}(\rho_s)
                - i \frac{\kappa_{\alpha_\Delta} \bestguess{\omega}{1}^2}{\kappa_{\alpha_\Sigma} \omega_2} \comm{T_z^2}{\rho_s} + \mathcal{O}\qty(\frac{1}{\omega_2^3}).
\end{equation}

For the second term in equation~\eqref{eq:ls_2_terms_ultra} we obtain
\begin{equation}
    \Tr_E\qty(\overline{\tilde{\mathcal{L}}_1 \mathcal{R} \tilde{\mathcal{K}}_1(\rho_s)}) = 
    - i \cos(\alpha) \frac{\kappa_{\alpha_\Delta}}{2 \kappa_{\alpha_\Sigma}} \bestguess{\omega}{1} \Tr(\sigma_+ \bar{\rho}_E) \comm{T_-}{\comm{T_+}{\rho_s}} + i \cos(\alpha) \frac{\kappa_{\alpha_\Delta}}{2 \kappa_{\alpha_\Sigma} }\bestguess{\omega}{1} \Tr(\sigma_- \bar{\rho}_E) \comm{T_+}{\comm{T_-}{\rho_s}},
\end{equation}
and using $\Tr(\sigma_+ \bar{\rho}_E) = \Tr(\sigma_- \bar{\rho}_E) = - \frac{\kappa_{\alpha_\Delta} \cos{\left(\alpha \right)}}{2 \kappa_{\alpha_\Sigma}}$ we obtain that
\begin{equation}
    \Tr_E\qty(\overline{\tilde{\mathcal{L}}_1 \mathcal{R} \tilde{\mathcal{K}}_1(\rho_s)}) = 
    - i {\qty(\frac{\kappa_{\alpha_\Delta}}{\kappa_{\alpha_\Sigma}})}^2 \frac{\cos^2(\alpha)}{4} \bestguess{\omega}{1} \comm{\comm{T_+}{T_-}}{\rho_s}.
\end{equation}

For the third term in equation~\eqref{eq:ls_2_terms_ultra} we obtain
\begin{align}
    \Tr_E\qty(\overline{\tilde{\mathcal{L}}_1 (1 - \mathcal{R}) \tilde{\mathcal{K}}_1(\rho_s)})
     &= \frac{\bestguess{\omega}{1}^2}{2 \Lambda} \qty(c_+ \qty( T_+ \rho_s T_- - T_- T_+ \rho_s) + c_-^* \qty(T_- \rho_s T_+ - \rho_s T_+ T_-))\\
     &+ \frac{\bestguess{\omega}{1}^2}{2 \Lambda} \qty(c_- \qty(T_- \rho_s T_+ - T_+ T_- \rho_s) + c_+^* \qty(T_+ \rho_s T_- - \rho_s T_- T_+))\\
     &+ \frac{d \; \bestguess{\omega}{1}^2}{\kappa_\Sigma^2 + \bestguess{\omega}{1}^2} (\kappa_\Sigma + i \bestguess{\omega}{1}) \qty(T_- \rho_s T_+ - T_- T_+ \rho_s + T_+ \rho_s T_- - \rho_s T_+ T_-)\\
     &+ \frac{d \; \bestguess{\omega}{1}^2}{\kappa_\Sigma^2 + \bestguess{\omega}{1}^2} (\kappa_\Sigma - i \bestguess{\omega}{1}) \qty(T_+ \rho_s T_- - T_+ T_- \rho_s + T_- \rho_s T_+ - \rho_s T_- T_+),
\end{align}
with 
\begin{align}
    c_+ &= \Tr(\sigma_+ B \bar{\rho}_E),\\
    c_- &= \Tr(\sigma_- B^\dag \bar{\rho}_E),\\
    d   &= \Tr(\sigma_+ (\mathbb{1}_E + \sigma_{\alpha x}) \bar{\rho}_E) = \Tr(\sigma_- (\mathbb{1}_E + \sigma_{\alpha x}) \bar{\rho}_E)
         = \frac{1}{2} \qty(1 - \frac{\kappa_{\alpha_\Delta}^2}{\kappa_{\alpha_\Sigma}^2})\cos^2(\alpha).
\end{align}
Using
\begin{align}
    c_+ &= - \frac{i \kappa_{\alpha_\Delta}}{2 \kappa_{\alpha_\Sigma}} + \frac{ 4 i \NDelta + \kappa_{\alpha_\Sigma} + 4 \kappa_{\alpha x} - 2 i \bestguess{\omega}{1}}{4 \omega_2}
    + \mathcal{O}\qty(\frac{1}{\omega_2^2}),\\
    c_- &= - \frac{i \kappa_{\alpha_\Delta}}{2 \kappa_{\alpha_\Sigma}} + \frac{- 4 i \NDelta + \kappa_{\alpha_\Sigma} + 4 \kappa_{\alpha x} + 2 i \bestguess{\omega}{1}}{4 \omega_2}
    + \mathcal{O}\qty(\frac{1}{\omega_2^2}),
\end{align}
we readily obtain
\begin{align}
    \Tr_E\qty(\overline{\tilde{\mathcal{L}}_1 (1 - \mathcal{R}) \tilde{\mathcal{K}}_1(\rho_s)})
    &= \qty(\frac{\kappa_{\alpha_\Sigma} d \; \bestguess{\omega}{1}^2}{\kappa_{\alpha_\Sigma}^{2} +
    \bestguess{\omega}{1}^{2}} + \frac{\kappa_{\alpha_\Sigma} + 4 \kappa_{\alpha x} }{4
    \omega_2^2}\bestguess{\omega}{1}^2
            + \mathcal{O}\qty(\frac{1}{\omega_2^3}))
            \qty(\mathcal{D}_{T_-}(\rho_s) + \mathcal{D}_{T_+}(\rho_s))\\
    &+ i \qty(\frac{\bestguess{\omega}{1}^3 \, d}{2 \qty(\kappa_{\alpha_\Sigma}^{2} + \bestguess{\omega}{1}^{2})} + \mathcal{O}\qty(\frac{1}{\omega_2^2})) \comm{\comm{T_+}{T_-}}{\rho_s}\\
    &+ i \qty(\frac{\bestguess{\omega}{1}^2}{4 \omega_2} \frac{\kappa_{\alpha_\Delta}}{\kappa_{\alpha_\Sigma}} + \mathcal{O}\qty(\frac{1}{\omega_2^3})) \comm{T_+ T_- + T_- T_+}{\rho_s}.
\end{align}

Putting all of the calculations of this section together, we obtain the following second-order reduced model.
For the slow dynamics we obtain an explicit Lindbladian model, where we have kept the leading-order in
$\frac{1}{\omega_2}$ for every different type of term:
\begin{align}
    \mathcal{L}_{s,g}(\rho_s) &=  - i  \comm{\omega_{s,z,1} T_z +  \omega_{s,z,2} T_z^2 + \omega_{s,\textrm{c}} \comm{T_+}{T_-} + \omega_{s,\textrm{a}} (T_+ T_- + T_- T_+)}{\rho_s}\nonumber\\
    &+ \kappa_{s,z} \mathcal{D}_{T_z}(\rho_s) + \kappa_{s,\pm} \qty(\mathcal{D}_{T_-} + \mathcal{D}_{T_+}) +
    \mathcal{O}\qty(g \varepsilon^2) \label{eq:lindbladian_ultrastrong}
\end{align}
with 
\begin{align}
    \omega_{s,z,1} &= - \frac{\kappa_{\alpha_\Delta}}{\kappa_{\alpha_\Sigma}} \frac{g \NDelta}{\omega_2},\\
    \omega_{s,z,2} &= \frac{\kappa_{\alpha_\Delta} g ^2}{\kappa_{\alpha_\Sigma} \omega_2},\\
    \omega_{s,\textrm{c}} &= {\qty(\frac{\kappa_{\alpha_\Delta}}{\kappa_{\alpha_\Sigma}})}^2 \frac{g^2}{4 \bestguess{\omega}{1}}
                    - \frac{1}{4} \qty(1 - \frac{\kappa_{\alpha_\Delta}^2}{\kappa_{\alpha_\Sigma}^2}) \frac{\bestguess{\omega}{1} g^2}{\kappa_{\alpha_\Sigma}^{2} + \bestguess{\omega}{1}^{2}},\\
    \omega_{s,\textrm{a}} &= - \frac{\kappa_{\alpha_\Delta}}{\kappa_{\alpha_\Sigma}} \frac{g^2}{4 \omega_2},\\
    \kappa_{s,z} &= \qty(\frac{\NDelta^{2} \left(1 - \frac{\kappa_{\alpha_\Delta}^{2}}{\kappa_{\alpha_\Sigma}^{2}} \right)}{\kappa_{\alpha_\Sigma} \omega_2}
        + \frac{\kappa_{\alpha_\Sigma}}{2 \omega_2} + \frac{2 \kappa_{\alpha x}}{\omega_2}) \frac{2 g^2}{\omega_2},\\
    \kappa_{s,\pm} &= \frac{1}{2} \qty(1 - \frac{\kappa_{\alpha_\Delta}^2}{\kappa_{\alpha_\Sigma}^2}) \frac{\kappa_{\alpha_\Sigma} g^2}{\kappa_{\alpha_\Sigma}^{2} + \bestguess{\omega}{1}^{2}}
        + \frac{g^{2} \left(\kappa_{\alpha_\Sigma} + 4 \kappa_{\alpha x}\right)}{4 \omega_2^2} \, .
\end{align}
For the embedding of the slow subspace we obtain, up to second-order terms:
\begin{equation}\label{eq:kraus_map_ultra_strong}
    \mathcal{K}_{s,g}(\rho_s) = e^{- i H_g} (\rho_s \otimes \bar{\rho}_E) e^{i H_g} + \mathcal{O}(\varepsilon^2),
\end{equation}
with
\begin{align}
    H_g := \frac{\kappa_{\alpha_\Delta}}{\kappa_{\alpha_\Sigma}} \frac{g}{2 \bestguess{\omega}{1}} H_s \otimes \mathbb{1}_E
    + \frac{g}{2 \sqrt{\kappa_{\alpha_\Sigma}^{2} + \bestguess{\omega}{1}^{2}}} H_{s,\pm} \otimes \qty(\mathbb{1}_E + \sigma_{\alpha,x})
    - \frac{g}{\omega_2} T_z \otimes \bar{M}_z
    - \frac{g}{2 \omega_2} H_s \otimes \sigma_{y}.
\end{align}
Here we have defined
\begin{align}
    H_s &= i e^{- i \bestguess{\omega}{1} t} T_+ - i e^{i \bestguess{\omega}{1} t} T_-,\\
    H_{s,\pm} &= \frac{(\kappa_{\alpha_\Sigma} - i \bestguess{\omega}{1}) e^{i \bestguess{\omega}{1} t} T_- + (\kappa_{\alpha_\Sigma} + i \bestguess{\omega}{1}) e^{- i \bestguess{\omega}{1} t} T_+}{\sqrt{\kappa_{\alpha_\Sigma}^2 + \bestguess{\omega}{1}^2}},\\
    \bar{M}_z &= \sigma_{\alpha x} - \frac{\NDelta}{\kappa_{\alpha_\Sigma}} \qty(\frac{\kappa_{\alpha_\Delta}}{\kappa_{\alpha_\Sigma}} \mathbb{1}_E + \sigma_{\alpha x}).
\end{align}

It is easy to verify that $H_g$ is Hermitian, since $H_s$, $H_{s,\pm}$ and $\bar{M}_z$ are Hermitian, and hence $\mathcal{K}_g$ can be written as an entangling unitary up to $\mathcal{O}(\varepsilon^2)$. In particular, $\mathcal{K}_g$ is therefore a CPTP map up to $\mathcal{O}(\varepsilon^2)$ terms.

\subsubsection*{Discussion of Hamiltonian terms}

All the Hamiltonian terms are suppressed asymptotically for large
$\bestguess{\omega}{1}$ and $\omega_2$, although this was not explicitly part of our goal.
Note that in contrast to the case of strong driving where $\kappa_-$ dominates $\kappa_+$ for a cold bath, with ultrastrong driving we typically keep $\kappa_{\alpha -}$ and $\kappa_{\alpha +}$ of the same order. The resulting conclusions are consistent with the induced dissipations and Hamiltonians derived when both viewpoints hold, i.e. taking the limit of large $\omega_2$ in the expressions obtained with the dissipation model of strong driving on E.

\end{widetext}


\bibliographystyle{plain}
\bibliography{refs.bib}

\end{document}